\documentclass[notitlepage,superscriptaddress,nofootinbib,tightenlines,twocolumn]{revtex4-1}
\usepackage{amsmath,verbatim,latexsym,amssymb,graphicx,indentfirst,mathrsfs,mathtools,amsthm,bbm,bm,hyperref,url,cancel,subcaption,enumitem,tabularx}
\usepackage[font=small,labelfont=bf,format=plain,justification=raggedright,singlelinecheck=false]{caption}

\newcolumntype{M}[1]{>{\centering\arraybackslash}m{#1}}
\newcolumntype{N}{@{}m{0pt}@{}}

\usepackage{verbatim,indentfirst}
\usepackage[title,titletoc]{appendix}

\renewcommand{\thesection}{\arabic{section}}
\renewcommand{\thesubsection}{\arabic{section}.\arabic{subsection}}
\renewcommand{\thesubsubsection}{
\arabic{section}.\arabic{subsection}.\arabic{subsubsection}}
\makeatletter
\def\p@subsection{}
\makeatother
\makeatletter
\def\p@subsubsection{}
\makeatother

\newtheorem{proposition}{Proposition}

\makeatletter
\newcommand\footnoteref[1]{\protected@xdef\@thefnmark{\ref{#1}}\@footnotemark}
\makeatother

\usepackage{soul}

\usepackage{color}
\usepackage[dvipsnames]{xcolor}
\usepackage[normalem]{ulem}

\def\Mg{{\rm Mg}}
\def\Ca{{\rm Ca}}

\def\H{{\rm H}}
\def\Two{ ^{2+} }
\def\PO{{\rm PO}_4^{3-}}
\newcommand{\spin}{{\rm spin}}
\newcommand{\orb}{{\rm orb}}
\newcommand{\Posner}{{\rm Pos}}
\newcommand{\nuc}{{\rm nuc}}
\newcommand{\lab}{{\rm lab}}  
\newcommand{\zlab}{z_\lab}
\newcommand{\In}{{\rm in}}  
\newcommand{\zinter}{z_\In}  
\newcommand{\xinter}{x_\In}  
  
\newcommand{\zenz}{z_{\rm enz}} 
\newcommand{\TPos}{t_\Posner}
\newcommand{\comp}{{\rm comp}}
\newcommand{\HilPos}{\Hil_{ \text{no-coll.} }^-}
\newcommand{\PiPos}{\Pi_{ \text{no-coll.} }^-}
\newcommand{\Log}{{\rm L}}
\newcommand{\Logg}{{\rm \mathcal{L}}}
\newcommand{\rep}{{\rm rep}}
\newcommand{\PiStick}{ \Pi_{ AB} }
\newcommand{\PiStickBC}{ \Pi_{ BC} }
\newcommand{\qutrit}{{\rm qutrit}}

\newcommand{\CThree}{C}

\newcommand{\GC}{\mathcal{G}_{ \CThree }}

\newcommand{\rot}{{\rm rot}}
\newcommand{\fire}{{\rm fire}}
\newcommand{\diffuse}{{\rm diff}}

\newcommand{\Sites}{N}    

\newcommand{\AKLT}{{\rm AKLT}}   
\newcommand{\AKLTOne}{\ket{ \AKLT_{1 \text{D}} }}   
\newcommand{\AKLTHon}{\ket{ \AKLT_{\rm hon} }}   
\newcommand{\AKLTPrime}{\ket{ \AKLT'_{\rm hon} }}   
\newcommand{\Thirtyone}{^{31}{\rm P}}
\newcommand{\Graph}{G}
\newcommand{\zerot}{0_\tau}
\newcommand{\onet}{1_\tau}
\newcommand{\twot}{2_\tau}
\newcommand{\jt}{j_\tau}
\newcommand{\plust}{ +_\tau }
\newcommand{\inter}{ {\rm int} }   

\newcommand{\tot}{ {\rm tot} }
\def\const{ {\rm const.} }   
\newcommand{\Tr}{{\rm Tr}}   
\def\id{\mathbbm{1}}   
\newcommand{\kB}{k_\mathrm{B}}  

\newcommand{\Hil}{\mathcal{H}}  
\newcommand{\Basis}{\mathcal{B}}  

\newcommand{\LParen}{ \bm{(} }
\newcommand{\RParen}{ \bm{)} }
\newcommand*{\Set}[1]{\left\{  #1  \right\}}

\renewcommand\th{ {\rm th} }

\newcommand*{\bra}[1]{\langle #1\rvert}
\newcommand*{\ket}[1]{\lvert #1 \rangle}

\newcommand*{\ketbra}[2]{\lvert #1 \rangle\!\langle #2 \rvert}
\newcommand*{\expval}[1]{\left\langle  #1  \right\rangle}

\usepackage{graphicx}
\usepackage{epstopdf}
\newcommand{\caphead}[1]{{\bf #1}}

\begin{document}

\title{Quantum information in the Posner model of quantum cognition}
\author{Nicole Yunger Halpern}
\email{nicoleyh.11@gmail.com}
\affiliation{Institute for Quantum Information and Matter, Caltech, Pasadena, CA 91125, USA}
\affiliation{Kavli Institute for Theoretical Physics, University of California, Santa Barbara, CA 93106, USA}
\affiliation{ITAMP, Harvard-Smithsonian Center for Astrophysics, Cambridge, MA 02138, USA}
\affiliation{Department of Physics, Harvard University, Cambridge, MA 02138, USA}
\author{Elizabeth Crosson}
\email{crosson@unm.edu}
\affiliation{Institute for Quantum Information and Matter, Caltech, Pasadena, CA 91125, USA}
\affiliation{Kavli Institute for Theoretical Physics, University of California, Santa Barbara, CA 93106, USA}
\affiliation{Center for Quantum Information and Control (CQuIC)
Dept. of Physics and Astronomy, University of New Mexico, 
Albuquerque, NM 87131, USA}
\date{\today}

\begin{abstract} 
Matthew Fisher recently postulated a mechanism by which quantum phenomena could influence cognition: Phosphorus nuclear spins may resist decoherence for long times, especially when in Posner molecules. The spins would serve as biological qubits. We imagine that Fisher postulates correctly. How adroitly could biological systems process quantum information (QI)? We establish a framework for answering. Additionally, we construct applications of biological qubits to quantum error correction, quantum communication, and quantum computation. First, we posit how the QI encoded by the spins transforms as Posner molecules form. The transformation points to a natural computational basis for qubits in Posner molecules. From the basis, we construct a quantum code that detects arbitrary single-qubit errors. Each molecule encodes one qutrit. Shifting from information storage to computation, we define the model of Posner quantum computation. To illustrate the model's quantum-communication ability, we show how it can teleport information incoherently: A state's weights are teleported. Dephasing results from the entangling operation's simulation of a coarse-grained Bell measurement. Whether Posner quantum computation is universal remains an open question. However, the model's operations can efficiently prepare a Posner state usable as a resource in universal measurement-based quantum computation. The state results from deforming the Affleck-Kennedy-Lieb-Tasaki (AKLT) state and is a projected entangled-pair state (PEPS). Finally, we show that entanglement can affect molecular-binding rates, boosting a binding probability from 33.6\% to 100\% in an example. This work opens the door for the QI-theoretic analysis of biological qubits and Posner molecules.
\end{abstract}
\maketitle{}

\section{Introduction}

Fisher recently proposed 
a mechanism by which 
quantum phenomena might affect cognition~\cite{Fisher15}.
Phosphorus atoms populate biochemistry.
A phosphorus nucleus's spin, he argued,
can store quantum information (QI) for long times.
The nucleus has a spin quantum number $s = \frac{1}{2}$.
Hence the nucleus forms a \emph{qubit},
a quantum two-level system.
The qubit is the standard unit of QI.

Fisher postulated physical processes that might
entangle phosphorus nuclei.
Six phosphorus atoms might, with other ions,
form \emph{Posner molecules}, Ca$_9$(PO$_4$)$_6$~\cite{Treboux_00_Posner,Kanzaki_01_Posner,Yin_03_Posner}.\footnote{
Ca$_9$(PO$_4$)$_6$ has been called
the \emph{Posner cluster} and \emph{Posner molecule}.
We call it the \emph{Posner}, for short.}
The molecules might protect the spins' states
for long times.
Fisher also described how the QI stored in the spins
might be read out.
This QI, he conjectured, could impact neuron firing.
The neurons could participate in \emph{quantum cognition}.

These conjectures require empirical testing.
Fisher has proposed experiments~\cite{Fisher15},
including with Radzihovsky~\cite{Fisher_17_Quantum}.
Some of those experiments have begun~\cite{Fisher_17_Personal}.

Suppose that Fisher conjectures correctly.
How effectively could the spins process QI?
We provide a framework for answering this question,
and we begin answering.
We translate Fisher's physics and chemistry
into information theory.
The language of molecular binding, heat dissipation, etc.
is replaced with the formalism of
positive operator-valued measures (POVMs), 
computational bases, etc.
Additionally, we identify and quantify QI-storage, -communication, 
and -computation
capacities of the phosphorus nuclear spins and Posners.

The constructions and analyses consist largely of QI theory:
We leave primarily to Fisher
conjectures about which chemical processes 
occur in biological systems.
We suppose that Fisher conjectures correctly,
identifying the QI-theoretic implications of his proposal.
Granted, experiments might rule against the proposal,
but they might rule in favor.
Furthermore, Fisher's arguments are compelling enough
that their ramifications merit exploration.

To initiate that exploration,
we identify QI-processing tasks 
that Posners could undertake in principle.
We do not claim that Posners do process QI.
Such claims require justification with biochemistry,
whereas this paper focuses on QI theory.
This in-principle analysis forms a crucial starting point:
Characterizing a system's QI-processing power is difficult. 
To render the problem tractable, 
we sometimes imagine that atoms and molecules
can be manipulated with complete control.
Though impractical, this assumption provides a footing
on which to build an analysis.
(Nor does every part of this paper rely on this assumption.)
We chalk out boundaries on Posners' QI-processing power.

\section{Theory}

This paper is intended for QI scientists, 
for chemists, and for biophysicists.
Some readers may require background about QI theory.
They may wish to consult App.~\ref{section:QI_backgrnd} 
and~\cite{NielsenC10,Preskill_99_QEC}.
Next, we overview this paper's contributions.

\textbf{Computational bases before and after 
molecule formation:}
Phosphorus nuclear spins originate outside Posners,
in Fisher's narrative.
The spins occupy phosphate ions that join together to form Posners.
Molecular formation changes
how QI is encoded physically.

Outside of molecules, phosphorus nuclear spins
couple little to orbital degrees of freedom (DOFs).
Spin states form an obvious choice of computational basis.\footnote{
In QI, computations are expressed in terms of 
a \emph{computational basis}
for the system's Hilbert space~\cite{NielsenC10}.
Basis elements are often represented by bit strings, as in
$\Set{ \ket{ 0 0 \ldots 0 } , \ket{ 0 0 \ldots 0 1 } ,  \ldots  \ket{ 1 1 \ldots 1 } }$.}
In a Posner molecule, the nuclei are indistinguishable.
They occupy a totally antisymmetric state~\cite{Fisher15,Fisher_17_Quantum}:
The spins entangle with orbital DOFs.
Which physical states form a useful computational basis
is not obvious.

We identify such a basis.
Molecule formation, we posit further, maps premolecule spin states
to antisymmetric molecule states deterministically.
The premolecule orbital state determines the map.
We formalize the map with a projector-valued measure (PVM).
The mapped-to antisymmetric states form the computational basis,
in terms of which Posners' QI processing 
can be expressed cleanly.

\textbf{Quantum error-correcting and -detecting codes:}
The basis elements may decohere quickly:
Posners' geometry protects only spins.
The basis elements are spin-and-position entangled states.
Do the dynamics protect any states against errors?

Hamiltonians' ground spaces may form
quantum error-correcting and -detecting codes 
(QECD codes)~\cite{Preskill_99_QEC}.
One might hope to relate the Posner Hamiltonian $H_\Posner$
to a QECD code.
$H_\Posner$ was characterized shortly after
the present paper's initial release~\cite{Swift_17_Posner}.
Even without knowing the form of $H_\Posner$, however,
one can construct QECD codes that respect
charges expected to be conserved.

$H_\Posner$ likely preserves two observables.
One, $\GC$, generates cyclic permutations of the spins.
One such permutation shuffles the spins 
about the molecule's symmetry axis,
through an angle $2 \pi / 3$.
This permutation preserves the Posner's geometry~\cite{Treboux_00_Posner,Kanzaki_01_Posner,Yin_03_Posner,Fisher15}.
The other charge, $S^{ \zlab }_{1 \ldots 6 }$,
is the spins' total $z$-component relative to the lab frame.

The dynamics likely preserve
eigenstates shared by $\GC$ and $S^{ \zlab }_{1 \ldots 6}$.
Yet $\GC$ shares many eigenbases with $S^{ \zlab }_{1 \ldots 6}$:
The charges fail to form a complete set of commuting observables (CSCO).
We identify a useful operator to break the degeneracy:
$\mathbf{S}^2_{123}  \otimes  \mathbf{S}^2_{456}$
equals a product of the spin-squared operators $\mathbf{S}^2$
of trios of a Posner's spins.
This operator (i) respects the Posner's geometry and
(ii) facilitates the construction of Posner states 
that can fuel universal quantum computation (discussed below).

From the eigenbasis shared by $\GC$, 
$S^{ \zlab }_{1 \ldots 6}$,
and $\mathbf{S}^2_{123}  \otimes  \mathbf{S}^2_{456}$,
we form QECD codes.
A state $\ket{ \psi }$ in one charge sector of $\GC$
and one sector of $S^{ \zlab }_{1 \ldots 6}$
likely cannot transform, under the dynamics,
into a state $\ket{ \phi }$ in a second sector of $\GC$
and a second sector of $S^{ \zlab }_{1 \ldots 6}$.
Hence $\ket{ \psi }$ and $\ket{ \phi }$
suggest themselves as codewords.
Charge preservation would prevent
one codeword from evolving into another.

We construct two quantum codes,
each partially protected by charge preservation.
Via one code, each Posner encodes one qutrit.
The codewords correspond to
distinct eigenvalues of $\GC$.
This code detects arbitrary single-physical-qubit errors.
Via the second code, each Posner encodes one qubit.
This repetition code corrects two bit flips.
The codewords correspond to
distinct eigenvalues of $S_{ 1 \ldots 6 }^{ \zlab }$.

\textbf{Model of Posner quantum computation:}
Fisher posits physical processes, such as binding,
that Posners may undergo~\cite{Fisher15}.
We abstract away the physics, 
formalizing the computations effected by the processes.
These effected \emph{Posner operations} form
the \emph{model of Posner quantum computation}.

The model includes the preparation of
singlets, $\frac{1}{ \sqrt{2} } ( \ket{ 0 1 } - \ket{ 1 0 } )$.
The logical state evolves trivially,
under the identity $\id$,
when Posners form.
But Posner creation associates a hextuple of qubits
with a geometry
and with an observable $\GC$.

A Posner's six qubits may rotate through angles of up to $\pi$,
though typical angles are expected to be much smaller.
Also, measurements can be performed:
$\GC$ has eigenvalues $\tau = 0, \pm 1$.
Whether Posners $A$ and $B$ satisfy
$\tau_A + \tau_B = 0$ can be measured projectively.
If the equation is satisfied, the twelve qubits 
can rotate jointly.

Finally, hextuples can cease to correspond 
to geometries or to $\GC$'s,
as Posners break down into their constituent ions.
Thereafter, qubits can group together into new hextuples.
This model enables us to recast Fisher's narrative~\cite{Fisher15}
as a quantum circuit.

%
%
%
\textbf{Entanglement generated by, 
and quantum communication with,
molecular binding:}
Two Posners, Fisher conjectures, 
can bind together~\cite{Fisher15}.
Quantum-chemistry calculations support the conjecture~\cite{Swift_17_Posner}.
The binding is expected to entangle the Posners~\cite{Fisher15}.
How much entanglement does binding generate,
and entanglement of what sort?

We characterize the entanglement in two ways.
First, we compare Posner binding to a Bell measurement~\cite{NielsenC10}.
A Bell measurement yields one of four possible outcomes---two
bits of information.
Posner binding transforms a subspace
as a coarse-grained Bell measurement.
A Bell measurement is performed,
and one bit is discarded, effectively.

Second, we present a quantum-communication protocol
reliant on Posner binding.
We define a qutrit (three-level) subspace
of the Posner Hilbert space.
A Posner $P$ may occupy a state 
$\ket{ \psi }  =  \sum_{ j = 0}^2  c_j  \ket{ j }$
in the subspace.
The coefficients $| c_j |^2$ form a probability distribution $Q$.
This distribution has a probability $p$ of being teleported 
to another Posner, $P'$.
Another distribution, $\tilde{Q}$, consists of
combinations of the $| c_j |^2$'s.
$\tilde{Q}$ has a probability $1 - p$
of being teleported.
Measuring $P'$ in the right basis
would yield an outcome distributed 
according to $Q$ or according to $\tilde{Q}$.
A random variable is teleported,
though $P$ never interacts with $P'$ directly.

The weights of $\ket{ \psi }$ 
(or combinations of the weights)
are teleported~\cite{Bennett_93_Teleporting}.
The coherences are not.
We therefore dub the protocol \emph{incoherent teleportation}.
The dephasing comes from the binding's simulation of
a coarse-grained Bell measurement.
Bell measurements teleport QI coherently.

Incoherent teleportation effects
a variant of superdense coding~\cite{Bennett_92_Communication}.
A trit (a classical three-level system) 
is communicated effectively,
while a bit is communicated directly.
The trit is encoded superdensely in the bit,
with help from Posner binding.

\textbf{Posner-molecule state
that can serve as a universal resource for
measurement-based quantum computation:}
Measurement-based quantum computation (MBQC)~\cite{Briegel_01_Persistent,Raussendorf_03_Measurement,Briegel_09_Measurement}
is a quantum-computation model
alternative to the circuit model~\cite{Deutsch_89_Quantum}.
MBQC begins with a many-body entangled state $\ket{ \psi }$.
Single qubits are measured adaptively.

MBQC can efficiently simulate universal quantum computation
if begun with the right $\ket{ \psi }$.
Most quantum states cannot serve as universal resources~\cite{Gross_09_Most}.
Cluster states~\cite{Briegel_01_Persistent,Raussendorf_01_One,Hein_06_Entanglement} 
on 2D square lattices can~\cite{Briegel_01_Persistent,Raussendorf_03_Measurement,VandenNest_06_Universal,Miyake_11_Quantum}.
So can the Affleck-Kennedy-Lieb-Tasaki (AKLT) state~\cite{Affleck_87_Rigorous,Affleck_88_Valence,Kennedy_88_Existence}
on a honeycomb lattice, $\AKLTHon$.
Local measurements can transform $\AKLTHon$
into the universal cluster state.
Hence $\AKLTHon$ can fuel universal MBQC~\cite{Wei_12_Two,Miyake_11_Quantum}.

We define a variation $\AKLTPrime$ on $\AKLTHon$.
$\AKLTPrime$ can be prepared efficiently with Posner operations.
Preparing $\AKLTHon$, one projects onto
a spin-$\frac{3}{2}$ subspace.
Preparing $\AKLTPrime$, one projects onto
a larger subspace.
Local measurements (supplemented by Posner hydrolyzation,
singlet formation, and Posner creation)
can transform $\AKLTPrime$ into the universal cluster state.
Hence $\AKLTPrime$ can fuel universal MBQC
as $\AKLTHon$ can.

Whether Posner operations can implement
the extra local measurements,
or the adaptive measurements in MBQC,
remains an open question.
Yet the universality of a Posner state,
efficiently preparable by a (conjectured) biological system,
is remarkable.
Most states cannot fuel universal MBQC~\cite{Gross_09_Most}.
The universality of $\AKLTPrime$ follows from
(i) Posners' geometry and (ii) their ability to share singlets.

Like $\AKLTHon$, $\AKLTPrime$ is
a projected entangled-pair state (PEPS)~\cite{Verstraete_08_Matrix}.
The state is formed from two basic tensors.
Each tensor has three physical qubits
and three virtual legs.
One virtual leg has bond dimension six.
Each other virtual leg
has bond dimension two.
$\AKLTPrime$ is
the unique ground state of 
some frustration-free Hamiltonian $H_{ \AKLT' }$~\cite{Perez_07_PEPS,Molnar_17_Generalization}.
The relationship between $H_{ \AKLT' }$ and $H_\Posner$
remains an open question.
So does whether $H_{ \AKLT' }$ has a constant-size gap.

%
%
%
\textbf{Entanglement's influence on binding probabilities:}
Entanglement, Fisher conjectures, can affect
Posners' probability of binding together~\cite{Fisher15}.
He imagined a Posner $A$ entangled with a Posner $A'$ 
and a $B$ entangled with a $B'$.
Suppose that $A$ has bound to $B$.
$A'$ more likely binds to $B'$, Fisher argues,
than in the absence of entanglement.

We formulate a scheme for quantifying
entanglement's influence on binding probabilities.
Two Posners, $A$ and $B$, illustrate the scheme.
First, we suppose that the pair contains no singlets.
Then, we raise the number of singlets incrementally.
In the final case, $A$ and $B$ are maximally entangled.
The binding probability rises from 33.6\% to 100\%.
Our technique can be scaled up
to Fisher's four-Posner example~\cite{Fisher15} and
to clouds of entangled Posners.

\textbf{Comparison with DiVincenzo's criteria:}
DiVincenzo codified the criteria required
for realizing quantum computation and communication~\cite{diVincenzo_00_Physical}.
We compare the criteria with Fisher's narrative.
At least most criteria are satisfied,
if sufficient control is available.
Whether the gate set is universal remains an open question.

\textbf{Organization of this paper:}
Section~\ref{section:Backgrnd_MF} reviews Fisher's proposal.
Section~\ref{section:One_Pos_set_up} 
details the physical set-up
and models Posner creation.
How Posner creation changes the physical encoding of QI
appears in Sec.~\ref{section:How_to_enc}.
QECD codes are presented in Sec.~\ref{section:Code_constraints}.

The model of Posner quantum computation is defined
in Sec.~\ref{section:Abstract_logic}.
Posner binding is analyzed, and applied to incoherent teleportation,
in Sec.~\ref{section:Stick_apps}.
Section~\ref{section:AKLT} showcases
the universal resource state $\AKLTPrime$.

Section~\ref{section:Bio_Bell} quantifies
entanglement's effect on molecular-binding probabilities.
Quantum cognition is compared with 
DiVincenzo's criteria in Sec.~\ref{section:diV}.
Opportunities for further study are detailed in Sec.~\ref{section:Outlook}.


%
%
%
\subsection{Review: Fisher's quantum-cognition proposal}
\label{section:Backgrnd_MF}

Biological systems are warm, wet, and large.\footnote{
We focus on Fisher's quantum-cognition proposal~\cite{Fisher15}.
Alternative proposals appear in, e.g.,~\cite{Penrose_89_Emperor's,Hameroff_14_Consciousness,Tegmark_15_Consciousness}.
}
Such environments quickly diminish quantum coherences.
Fisher catalogued the influences
that could decohere nuclear spins
in biofluids.
Examples include electric and magnetic fields
generated by other nuclear spins and by electrons.

These sources, Fisher estimated, 
decohere the phosphorus-31 ($\Thirtyone$) nuclear spin slowly.
Coherence times might reach $\sim 1$ s,
if the phosphorus occupies a free-floating phosphate ion,
or $10^5 - 10^6$ s,
if the phosphorus occupies a Posner.
No other biologically prevalent atom, Fisher conjectures,
has such a long-lived nuclear spin.

Phosphorus atoms inhabit many biological ions and molecules.
Examples include the phosphate ion, PO$_4^{3-}$.
Three phosphates feature in the molecule 
\emph{adenosine triphosphate} (ATP).
ATP stores energy that powers chemical reactions.
Two phosphates can detach from an ATP molecule,
forming a \emph{diphosphate} ion.

A diphosphate can break into two phosphates,
with help from the enzyme \emph{pyrophosphatase}.
The two phosphates' phosphorus nuclear spins 
form a \emph{singlet},
Fisher and Radzihovsky (F\&R) conjecture~\cite{Fisher15,Fisher_17_Quantum}.
A singlet is a maximally entangled state.
Entanglement is a correlation, shareable by quantum systems,
stronger than any achievable by classical systems~\cite{NielsenC10}.

Many biomolecules contain phosphate ions.
Occupying a small molecule, Fisher argues, could shelter 
the phosphorus nuclear spin:
Entanglement with other particles
could decohere the spin.
Dipole-dipole interactions with external protons
threaten the spin most.
But protons and small molecules tumble around each other in fluids.
The potential experienced by the phosphorus spin
is expected to average to zero.

Which small biomolecules could a phosphorus inhabit?
An important candidate is Ca$_9$(PO$_4$)$_6$. 
A Posner consists of 
six phosphate ions (PO$_4^{3-}$) and nine calcium ions (Ca$^{2+}$)~\cite{Treboux_00_Posner,Kanzaki_01_Posner,Yin_03_Posner}.
Posners form in simulated biofluids
and might form in vivo~\cite{Onuma98Posner,Ayako99Posner,Dey10Posner}.
A Posner could contain a phosphate 
that forms a singlet with a phosphate in another Posner.
The Posners would share entanglement.

Two Posners can bind together, according to
quantum-chemistry calculations~\cite{Fisher15,Swift_17_Posner}.
The binding projects the Posners
onto a possibly entangled state.
Moreover, pre-existing entanglement could affect
the probability that Posners bind.

Bindings, influenced by entanglement, could influence neuron firing.
Suppose that a Posner $A$ shares 
entanglement with a Posner $A'$
and that a $B$ shares entanglement
with a $B'$.
Posners $A$ and $B$ could enter one neuron,
while $A'$ and $B'$ enter another.
Suppose that $A$ binds with $B$.
The binding, with entanglement,
could raise the probability that
$A'$ binds to $B'$.

Bound-together Posners move slowly, Fisher argues.
Compound molecules must displace 
many water molecules,
which slow down the pair.
Relatedly, the Posner pair has a large moment of inertia.
Hence the pair rotates more slowly than separated Posners
by the conservation of angular momentum.

Hydrogen ions H$^+$ can attach easily to slow molecules,
Fisher expects.
H$^+$ \emph{hydrolyzes} Posners,
breaking the molecules into their constituent ions.
Hence entanglement might correlate 
hydrolyzation of $A$ and $B$ with
hydrolyzation of $A'$ and $B'$.
Hydrolyzation would release Ca$^{2+}$ ions
into the neurons.
Suppose that many entangled Posners hydrolyzed in these two neurons.
The neurons' Ca$^{2+}$ concentrations
could rise.
The neurons could fire coordinatedly due to entanglement.

\section{Results}

\subsection{Physical set-up and Posner-molecule creation}
\label{section:One_Pos_set_up}


This section concerns (i) the physical set-up
and (ii) the joining together of phosphates (and calcium ions) 
in Posner molecules.
Part of the material appears in~\cite{Fisher15,Fisher_17_Quantum}
and is reviewed.
Part of the material has not, according to our knowledge,
appeared elsewhere.

The phosphorus nuclei are associated with
spin and spatial Hilbert spaces in Sec.~\ref{section:H_spaces}.
Section~\ref{section:Posner_geo} reviews, and introduces notation for,
the Posner's geometry. 
Section~\ref{section:When_Pos_form} models 
the creation of a Posner
from close-together ions.

%
%
%
\subsubsection{Spin and spatial Hilbert spaces}
\label{section:H_spaces}

Each phosphorus nucleus has two relevant DOFs:
a spin and a position.
We will sometimes call the position
the \emph{orbital} or \emph{spatial} DOF.
Let $\Hil_\nuc^\spin$ and $\Hil_\nuc^\orb$ denote the associated Hilbert spaces.
The nucleus has a spin quantum number $s = \frac{1}{2}$.
Hence $\Hil_\nuc^\spin  =  \mathbb{C}^2$.
The orbital Hilbert space is infinite-dimensional:
$\dim ( \Hil_\nuc^\orb )  =  \infty$.
Each phosphorus nucleus's Hilbert space decomposes as
$\Hil_\nuc  =  \Hil_\nuc^\spin  \otimes  \Hil_\nuc^\orb$.

The electrons' states
transform trivially under all relevant operations,
Fisher and Radzihovsky (F\&R) conjecture~\cite{Fisher_17_Quantum}.
We therefore ignore the electronic DOFs.
We ignore calcium ions similarly.
We focus on the DOFs that might store QI for long times.

\subsubsection{Posner-molecule geometry and notation}
\label{section:Posner_geo}


Quantum-chemistry calculations have shed light on
the shapes available to Posners~\cite{Treboux_00_Posner,Kanzaki_01_Posner,Yin_03_Posner,Swift_17_Posner}.
A Posner's shape depends on the environment.
Posners in biofluids have begun to be studied~\cite{Yin_03_Posner}.
We follow~\cite{Fisher15,Fisher_17_Quantum},
supposing that more-detailed studies will support~\cite{Yin_03_Posner}.\footnote{
Reference~\cite{Swift_17_Posner},
released shortly after this paper's initial release,
supports~\cite{Yin_03_Posner}.}

The Posner forms a cube (Fig.~\ref{fig:Geo}).
At each face's center sits a phosphate.
The Posner lacks cubic symmetry,
due to the tetrahedral phosphates' orientations.
But (a stable proposed configuration of) the Posner 
retains S$_6$ symmetry
and one $\CThree_3$ symmetry.
The $\CThree_3$ symmetry is an invariance under
$2 \pi / 3$ rotations about a cube diagonal.

This cube diagonal serves as
the $z$-axis $\hat{z}_\In$ of 
a reference frame fixed in the molecule.
The atoms' positions remain constant relative to 
this \emph{internal frame}.
The internal frame can move relative to the lab frame,
denoted by the subscript ``lab.''
The spins' Bloch vectors are defined with respect to the lab frame.

Imagine gazing down the diagonal, 
as in Fig~\ref{fig:h_axis}.
You would see two triangles
whose vertices consisted of phosphates.
The triangles would occupy parallel planes
pierced by the $\hat{z}_\In$-axis.
The three black dots in Fig.~\ref{fig:Phi} represent
the phosphates closest to you. 
We label this trio's $\zinter$-coordinate by $h_+$.
Farther back, at $\zinter = h_-$, sits
the trio represented by gray dots.
$\hat{z}_\In$ points oppositely 
the direction in which we imagined gazing,
such that $h_+ > h_-$.

\begin{figure}[tb]
\centering
\begin{subfigure}{0.4\textwidth}
\centering
\includegraphics[width=.9\textwidth]{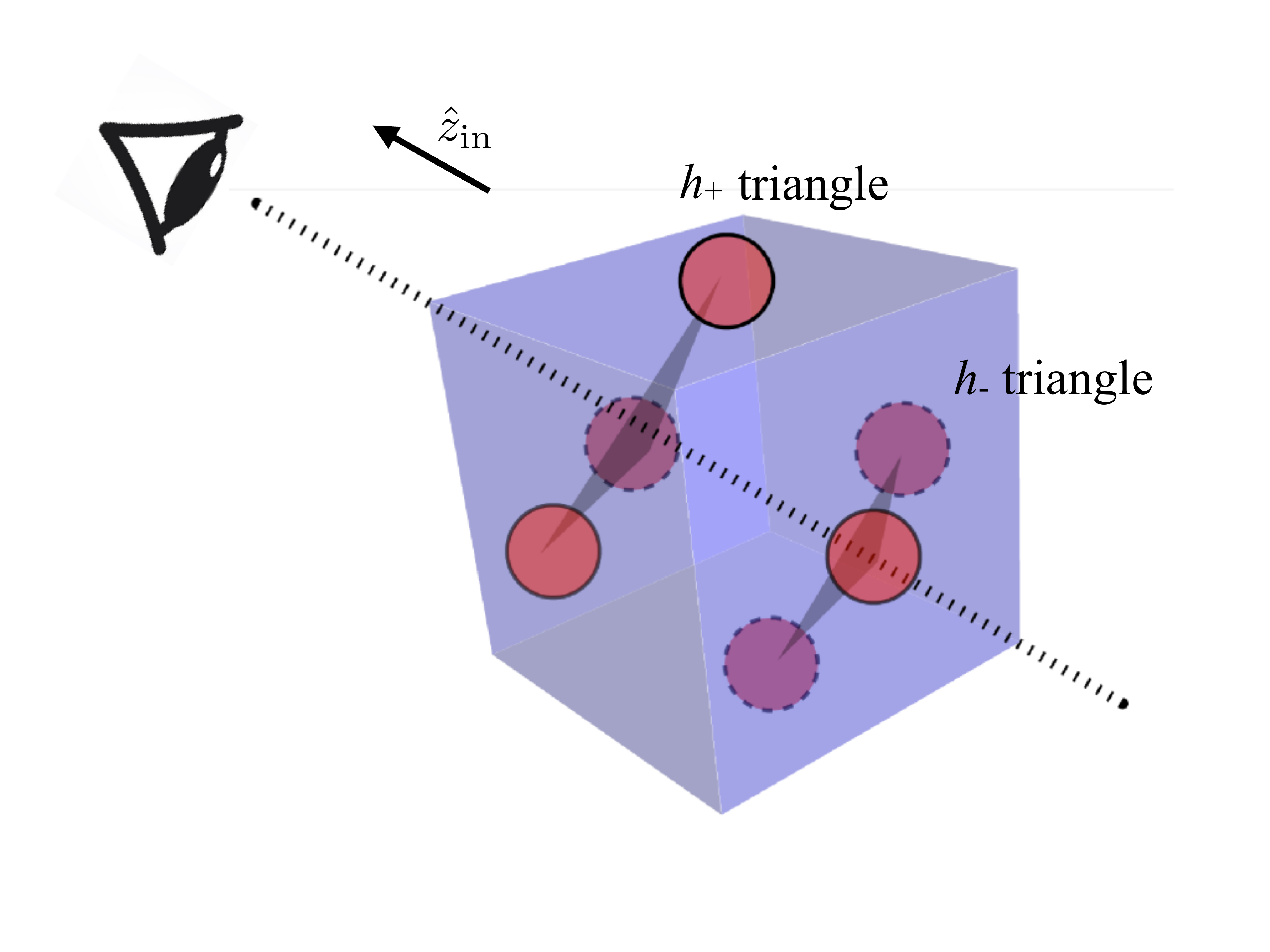}
\caption{}
\label{fig:h_axis}
\end{subfigure}
\begin{subfigure}{0.4\textwidth}
\centering
\includegraphics[width=.9\textwidth]{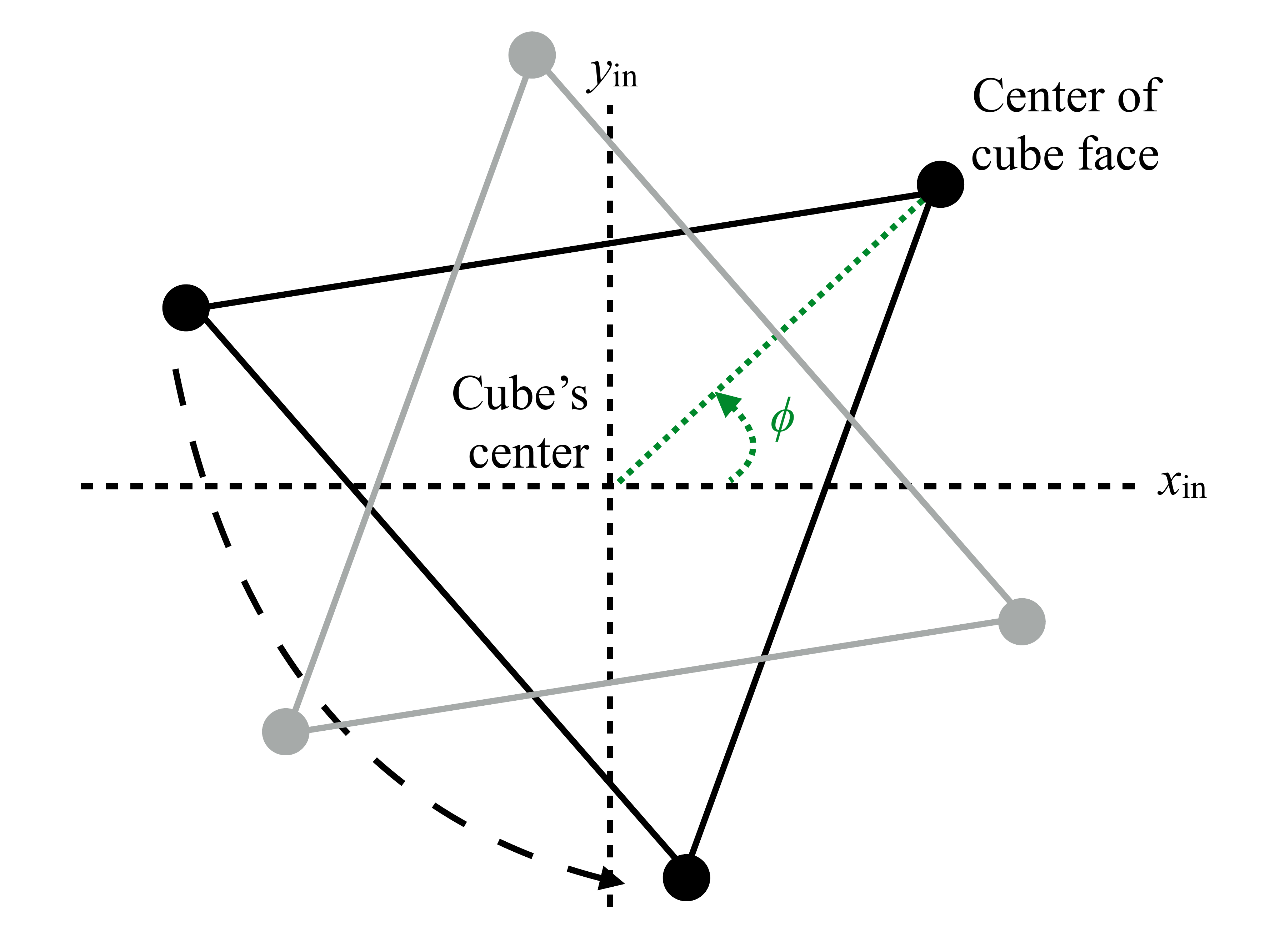}
\caption{}
\label{fig:Phi}
\end{subfigure}
\caption{\caphead{Posner-molecule geometry and coordinates:}
Quantum-chemistry calculations have shed light on 
the Posner molecule's cubic geometry~\cite{Swift_17_Posner,Treboux_00_Posner,Kanzaki_01_Posner,Yin_03_Posner}.
At each cube face's center
sits one phosphate ion (${\rm PO}_4^{3-}$).
The molecule appears to have one threefold symmetry axis
when in biofluids~\cite{Yin_03_Posner,Fisher15}.
The axis coincides with a cube diagonal.
Imagine gazing down the diagonal, as in Fig.~\ref{fig:h_axis}.
We orient the internal $z$-axis, $\hat{z}_\In$,
in the opposite direction.
(The internal reference frame remains fixed relative to the atoms' positions.)
Gazing down the diagonal, one sees
a triangle of phosphate ions
(the black dots in Fig.~\ref{fig:Phi}).
We denote the triangle's $\zinter$-coordinate by $h_+$.
$\phi$ denotes the least angle swept out counterclockwise
from the $+\xinter$-axis to a phosphate.
Behind the black-dot phosphates, at $\zinter = h_-$, sit
phosphates represented by gray dots in Fig.~\ref{fig:Phi}.
The gray dots form a triangle
rotated relative to the black-dot triangle 
through an angle $\pi / 4$.
The triangle pair remains invariant
under rotations, about $\hat{z}_\In$,
through an angle $2 \pi / 3$.
The long-dash line in Fig.~\ref{fig:Phi} illustrates such a rotation.
The invariance endows the Posner with $\text{C}_3$ symmetry.}
\label{fig:Geo}
\end{figure}

$\phi$ labels the triangles' orientation, 
as shown in Fig.~\ref{fig:Phi}.
We denote by $\varphi_j$ the angular orientation
of cube face $j$ (the site of a phosphate):
Consider a top-triangle face $j$, at $\zinter = h_+$.
Imagine rotating the $\xinter$-axis counterclockwise 
until it intersects a phosphate.
The angle swept out is $\varphi_j$.
One $h_+$ phosphate's $\varphi_j = \phi$,
another's $\varphi_j = \phi + 2 \pi / 3$,
and another's $\varphi_j = \phi + 4 \pi / 3$.
Now, consider the triangle at $\zinter = h_-$.
Each phosphate sits at an angle $\varphi_j + \pi / 4$,
for $\varphi_j = \phi$, $\phi + 2 \pi / 3$, or $\phi + 4 \pi / 3$.
We label the site of phosphate $j$ with an angle and a height: 
$( \varphi_j, h_j )$.

\subsubsection{Qualitative model for the creation of a Posner molecule}
\label{section:When_Pos_form}

Posners form from phosphate and calcium ions.
We propose a qualitative model for the formation process.
We first review how, 
according to Fisher, phosphorus nuclear spins
might come to form singlets.
We then envision phosphates falling into 
a potential well generated by the ions' mutual attraction
as a Posner forms.
F\&R have discussed the indistinguishability
of phosphorus nuclei in a Posner~\cite{Fisher15,Fisher_17_Quantum}.
We expand upon this discussion,
considering how distinguishable ions
become indistinguishable.
This discussion lays the foundation for constructing
a computational basis for in-Posner phosphorus nuclei.

Several molecules contain phosphate ions ${\rm PO}_4^{3-}$.
Examples include ATP (Sec.~\ref{section:Backgrnd_MF}).
Each ATP molecule contains three phosphates.
Two of the phosphates can break off, 
forming a diphosphate ion.
The enzyme pyrophosphatase can hydrolyze a diphosphate,
cleaving the ion into separated phosphates.
The separated phosphates contain phosphorus nuclear spins that,
Fisher conjectures~\cite{Fisher15}, form a singlet.

Let 1 and 2 label the phosphorus nuclear spins. 
Let $\hat{z}_{\rm enz}$ denote the $z$-axis of
a reference frame fixed in the enzyme.
Let $\hat{S}_{ \zenz }$ denote
the $\zenz$-component
of a phosphorus nucleus's spin operator.
Let $\ket{ \uparrow }$ and $\ket{ \downarrow }$
denote the $\hat{S}_{ \zenz }$ eigenstates:
$\hat{S}_{ \zenz }  \ket{ \uparrow }  =  \frac{ \hbar }{2}  \ket{ \uparrow }$,
and  $\hat{S}_{ \zenz }  \ket{ \uparrow } 
=  - \frac{ \hbar }{2}  \ket{ \uparrow }$.
The singlet has the form
\begin{align}
   \label{eq:Singlet}
   \ket{ \Psi^- }  :=  \frac{1}{ \sqrt{2} }  
   ( \ket{ \uparrow } \ket{ \downarrow }
   -  \ket{ \downarrow } \ket{ \uparrow } ) \, .
\end{align}
The singlet is one of the four \emph{Bell pairs}.
The Bell pairs are mutually orthogonal,
maximally entangled states of pairs of qubits~\cite{NielsenC10}.
Bell pairs serve as units of entanglement in QI.

Phosphorus nuclei are identical fermions,
as F\&R emphasize~\cite{Fisher15,Fisher_17_Quantum}.
But some of the nuclei's DOFs might be distinguishable 
before Posners form.
Consider, for example, two ATP molecules
on opposite sides of a petri dish.
Call the molecules $A$ and $B$.
A diphosphate could break off from each ATP molecule.
Each diphosphate could hydrolyze into two phosphates,
$A_1$ and $A_2$ or $B_1$ and $B_2$.
Consider the phosphorus nuclear spins
of one phosphate pair---say, of $A_1$ and $A_2$.
These spins would be indistinguishable:
Neither nucleus could be associated with 
an upward-pointing spin or with a downward-pointing spin.

But the spatial DOF of $A_1$ and $A_2$
could be distinguished from
the spatial DOF of $B_1$ and $B_2$:
We can imagine painting phosphate pair $A$ red 
and phosphate pair $B$ blue.
The phosphate pairs could diffuse to the dish's center.
The red pair and the blue pair could be tracked
along their trajectories.

\begin{figure}[tb]
\centering
\includegraphics[width=.4\textwidth, clip=true]{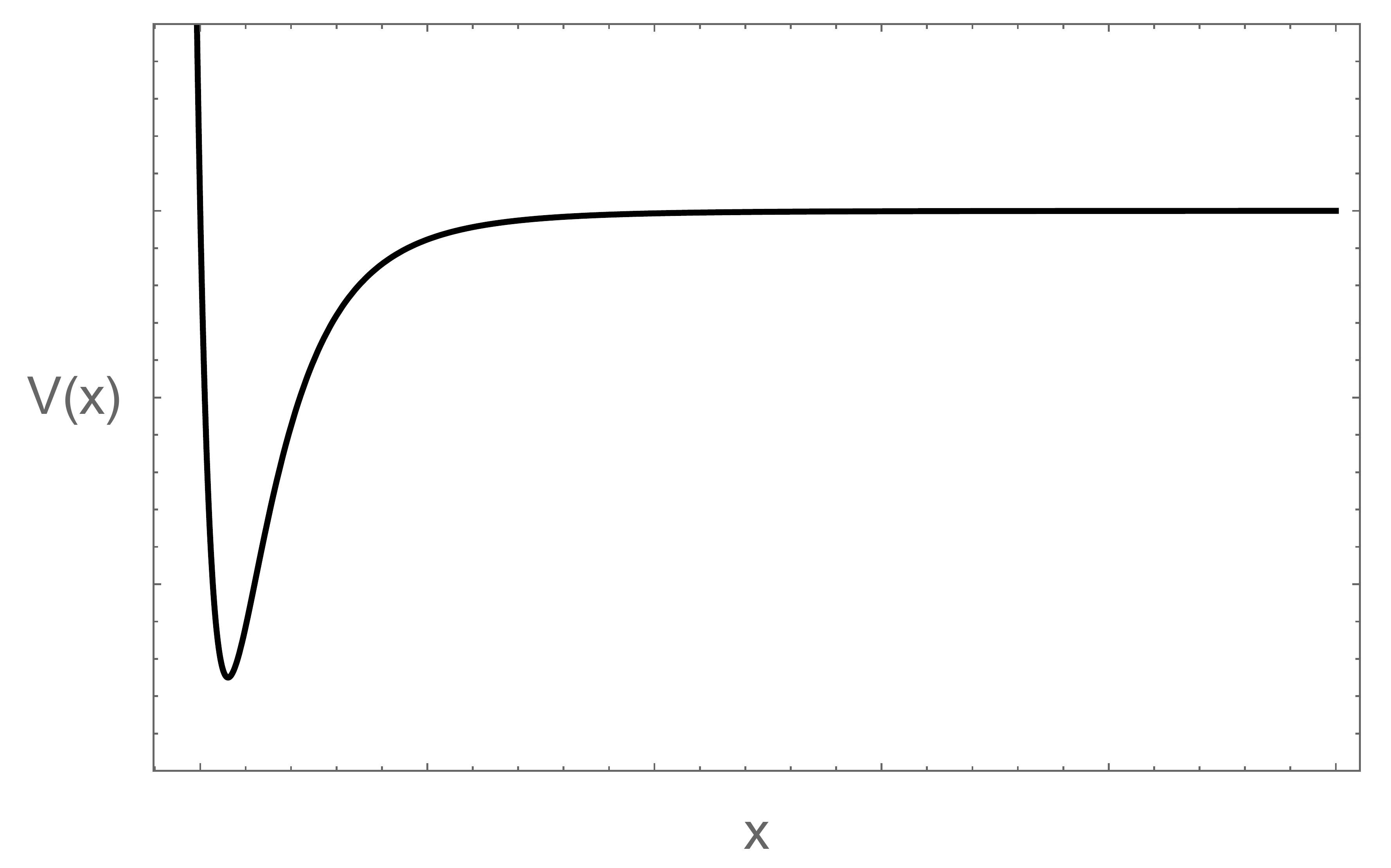}
\caption{\caphead{
Potential that models the ions' long-range attraction
and short-range repulsion:}  
Van der Waals forces draw particles together weakly at long range
and force particles apart strongly at short range.
The Lennard-Jones potential,
$(x)  =  \frac{ a }{ x^{12} }  -  \frac{b}{ x^6 }$,
forms a toy model for this qualitative behavior.
The real parameters $a, b > 0$.
We approximate qualitatively, with $V (x)$, 
the potential experienced by 
phosphate ions coalescing into a Posner molecule.
$x$ denotes the distance from a phosphate
to the system's center of mass.
}
\label{fig:Lennard_Jones}
\end{figure}

Consider six phosphates (and nine $\Ca\Two$ ions) approaching each other.
They are expected to attract each other weakly when far apart
and to repel strongly when close together.
The Lennard-Jones potential (Fig.~\ref{fig:Lennard_Jones}), 
used in molecular-dynamics simulations~\cite{Rapaport_04_Art},
captures these qualitative behaviors:
We temporarily approximate each phosphate as having
a classical position.
If $x$ denotes some phosphate's distance from 
the ions' center of mass,
\begin{align}
   \label{eq:LJ}
   V(x)  =  \frac{ a }{ x^{12} }  -  \frac{b}{ x^6 }  \, .
\end{align}
The real parameters $a, b  > 0$.

Where the concavity changes from negative to positive,
$\frac{ d^2  V (x) }{ d x^2 }  =  0$,
the potential has a ``lip.''
The ions have more energy, separated, than
they would have in a molecule.
The ions slide down the potential well,
releasing binding energy as heat.
The heat disrupts the environment,
which effectively measures the ions' state.\footnote{
That the environment measures the state via heat transfer
was proposed in~\cite{Fisher15}.}

At the well's bottom, the ions constitute a Posner molecule.
The phosphorus nuclei's quantum states
have position representations (wave functions) that overlap significantly.
The nuclei are indistinguishable~\cite{Fisher_17_Quantum}:
No nuclear pair can be identified as red-painted
or as blue-painted.
The six phosphorus nuclei occupy 
a totally antisymmetric spin-and-spatial state.
We will abbreviate ``totally antisymmetric'' as ``antisymmetric.''

\subsubsection{Formalizing the model for Posner-molecule creation}
\label{section:Formal_Pos_form}

Let us model, with mathematical tools of QI,
the environment's measuring of the ions,
the creation of a Posner,
and the antisymmetrization process.
Let $\TPos$ denote the scale of the time over which
the ions slide down the Lennard-Jones well from the lip,
emit heat, jostle about, and settle 
into the Posner geometry.

The environment effectively measures the ions 
with a frequency $1 / \TPos$.
We model the measurement with a
\emph{projector-valued measure} (PVM)~\cite{NielsenC10}.
Consider the Hilbert space 
$( \Hil_\nuc )^{\otimes 6 }$
of the Posner's six phosphorus nuclei.
An antisymmetric subspace $\HilPos$ consists of the states
available to the indistinguishable nuclei.
(The states are detailed in Sec.~\ref{section:How_to_enc_transf}.)
The subscript stands for ``no-colliding-nuclei'':
No two nuclei can inhabit the same Posner-cube face.

Let $\PiPos$ denote the projector onto $\HilPos$.
The PVM has the form
\begin{align}
   \label{eq:Form_Pos_PVM}
   \Set{  \PiPos  ,  \id  -  \PiPos  }  \, .
\end{align}

Suppose that one length-$(1 / \TPos )$ time interval has just passed.
The environment has measured the ions.
Suppose that, during the interval, 
the ions have emitted considerable heat.
The environment has registered 
the outcome ``Yes, a Posner has formed.''
$\PiPos$ has projected the ions' joint state.

Suppose, instead, that the ions have not emitted much heat.
The environment has registered the outcome ``No, no Posner has formed.''
$\id  -  \PiPos$ has projected the ions' joint state.\footnote{
One might try to model the environment
as measuring the ions continuously.
This model is unfaithful:
The environment would continuously project the ions
onto states inaccessible to a Posner.
No Posner could form,
due to the quantum Zeno effect~\cite{Misra_77_Zeno's}.
The Posner-creation time $\TPos$
sets the measurement's time scale.}$^,$\footnote{
F\&R suggest that,
upon forming, a molecule is entangled with 
its environment~\cite[Eq.~(7)]{Fisher_17_Quantum}.
Our PVM is consistent with F\&R's model,
by the principle of deferred measurement~\cite{NielsenC10}:
Let $S$ denote a general quantum system.
A measurement of $S$ consists of two steps:
First, $S$ is entangled with a memory $M$.
Second, $M$ is measured.
Suppose that (i) the entanglement is maximal
and (ii) the $M$ measurement is projective.
The $M$ measurement projects the system's state.
Suppose that $S$ evolves after the $M$ measurement.
This entangling, $M$ measurement, and evolution
is equivalent to
the entangling, followed by the $S$ evolution, 
followed by the $M$ measurement.
The $M$ measurement can be deferred
until after the evolution.
Deferral fails to alter the measurement statistics.
Let $S$ denote the nuclei, and let $M$ denote the environment.
The $M$ measurement is deferred
in F\&R's model, not in ours.
The models are equivalent, by the deferred-measurement principle.}

Let $\hat{\mathbf{S}}_{1 \ldots 6}$ denote 
the six phosphorus nuclei's total spin operator.
We assume that Posner creation can be modeled as
a two-stage process.
First, the independent phosphates tumble in the fluid.
They might experience magnetic fields 
generated by firing neurons.
The spins would rotate unitarily.
Second, the phosphates combine into a Posner
via an evolution that preserves $( \hat{S}^{ \zlab } )^{ \otimes 6 }$.

The assumption follows from
Fisher's claims that the spins barely decohere~\cite{Fisher15}:
The spins do not entangle with anything.
At worst, therefore, the spins rotate on the Bloch sphere
during Posner creation.
Most rotations fail to preserve $\hat{S}^{ \zlab }$.
But Posner creation that involves rotations
is mathematically equivalent to 
(i) rotations followed by 
(ii) $( \hat{S}^{ \zlab } )^{ \otimes 6 }$-conserving Posner creation.
The initial rotations can be absorbed into
the pre-Posner rotations.
We therefore will say that Posner creation 
``essentially preserves'' $( \hat{S}^{ \zlab } )^{ \otimes 6 }$.

\subsection{Encoded states and their changing physical representations}
\label{section:How_to_enc}

Phosphorus nuclear spins cleanly encode QI before Posners form.
The spins, Fisher conjectures, are decoupled from 
the nuclei's positions~\cite{Fisher15}.
Posner creation antisymmetrizes the spin-and-orbital state.
The spins become entangled with the positions,
no longer encoding QI cleanly.

But, we posit, Posner creation maps
each pre-Posner spin state
to an antisymmetric Posner state deterministically.
Posner creation preserves QI but changes
how QI is encoded physically.
Hence spin configurations can label
a computational basis for the Posner Hilbert space, e.g.,
$\uparrow \uparrow \uparrow \uparrow \uparrow \uparrow 
\equiv  000000$.

This section is organized as follows.
Section~\ref{section:How_to_enc_pre} concerns
premolecule phosphorus nuclear spins.
Section~\ref{section:Encodings} reviews the formalism of encodings.
A map between (i) physical states of pre-Posner spins
and (ii) logical states is formalized.
Logical states are mapped to Posner states 
in Sec.~\ref{section:How_to_enc_transf}.

\subsubsection{Physical encoding of quantum information 
in the phosphorus nuclei that will form a Posner molecule}
\label{section:How_to_enc_pre}

Consider six phosphates that approach each other,
soon to form (with $\Ca\Two$ ions) a Posner.
We index the phosphorus nuclei as $a = 1 , 2, \ldots, 6$.
Each nucleus has a spin DOF and an orbital DOF.
Nucleus $a$ occupies some quantum state 
$\rho_a  \in  \mathcal{D} ( \Hil_\nuc )$.
$\mathcal{D} ( \Hil )$ denotes 
the set of density operators (trace-one linear operators) 
defined on the Hilbert space $\Hil$.
$\rho_a$ may be pure (unentangled with any external DOFs)
or mixed (entangled with external DOFs, e.g., 
another phosphorus nucleus's spin).

Tracing out the orbital DOF yields the reduced spin state:
$\rho_a^\spin  :=  \Tr_\orb ( \rho_a )  \in  \mathcal{D} ( \mathbb{C}^2 )$.
The magnetic spin quantum number $m_a  =  \pm \frac{1}{2}$ 
quantifies the spin's $\zlab$-component.

Shifting focus from chemistry to information theory, 
we adopt QI notation:
We usually omit hats from operators,
and we often omit factors of $\hbar$ and $\frac{1}{2}$.
We often replace the spin operator's $\alpha$-component with 
the Pauli $\alpha$-operator, for $\alpha = x, y , z$:
$\hat{S}^\alpha  
\equiv  S^\alpha
=  \frac{ \hbar }{ 2 }  \:  \sigma^\alpha
\equiv  \sigma^\alpha$.
The $\sigma^z$ eigenstates are often labeled as
$\ket{ 0 }  :=  \ket{ \uparrow }$ and
$\ket{ 1 }  :=  \ket{ \downarrow }$.

Tracing out the spin DOF from $\rho_a$ yields 
the reduced orbital state:
$\rho_a^\orb  :=  \Tr_\spin ( \rho_a )  \in  \mathcal{D} ( \Hil_\nuc^\orb )$.
We parameterize $\Hil_\nuc^\orb$ with
the eigenstates $\ket{ \mathbf{x} }$ of the position operator, $\mathbf{x}$.
The coordinates are defined with respect to the lab frame. 
$\Set{ \ket{ \mathbf{x} } }$ forms a continuous set.

The spin and/or orbital DOFs can store QI.
But water and other molecules buffet the phosphates.
An independent phosphate's position is expected to be mixed.
The spin, in contrast, is expected to remain coherent
for long times. 
(See Sec.~\ref{section:diV} and~\cite{Fisher15}.)
The spins encode protected QI.

The nuclear spins form six qubits.
The qubits correspond to the Hilbert space
$(\Hil_\nuc^\spin)^{ \otimes 6 }  =  \mathbb{C}^{12}$, 
which has dimensionality $2^6  =  64$.
A useful basis for $\Hil_\nuc^\spin$ consists of tensor products
of $\sigma^z$ eigenstates:
$\Basis_\comp  :=
\Set{ \ket{ 0 , 0 , \ldots , 0 } , \ket{ 0 , 0 , \ldots , 0 , 1 } ,  
\ldots ,  \ket { 1 , 1 , \ldots , 1 } }$.
The notation $\ket{ A , B , \ldots, K }  \equiv  
\ket{A} \otimes \ket{B} \otimes \ldots \otimes \ket{K}$.
The set $\Basis_\comp$  the \emph{computational basis
for the physical states}.

Consider $\Sites$ hextuples of phosphates
($\Sites$ sets of six phosphates).
The phosphorus nuclei correspond to 
a spin space $\mathbb{C}^{ 6 \Sites }$. 
We suppose, without loss of generality, 
that the $6 \Sites$ spins occupy 
a pure joint state $\ket{ \psi }$.
Each hextuple could contain three singlets, for example.
Or a spin in some hextuple $A$
could form a singlet with
a spin in some hextuple $B$.

%
%
%
\subsubsection{Notation and quick review: Encodings}
\label{section:Encodings}

Imagine an agent Alice who wishes 
to send another agent, Bob, a message.
A quantum message is a quantum state $\ket{ \psi_\Log }  \in  \Hil_\Log$, 
called the \emph{logical state}.
Let $\Basis^\Log_\comp$ 
denote a preferred basis for
the Hilbert space $\Hil_\Log$.
Operations are expressed in terms of this
\emph{computational basis for the logical space}.

Alice must encode $\ket{ \psi_\Log }$
in the state of a physical system.
The agents would choose a \emph{code},
a dictionary between the computational basis $\Basis^\Log_\comp$ 
for the logical space
and the computational basis $\Basis_\comp$ for the physical space.
Alice would decompose $\ket{ \psi_\Log }$ 
in terms of $\Basis^\Log_\comp$ elements $\ket{ j_\Log }$;
replace each $\ket{ j_\Log }$
with a $\Basis_\comp$ element $\ket{ j }$;
and prepare the resultant \emph{physical state}:
$\ket{ \psi_\Log }  =  \sum_j  c_j  \ket{ j_\Log }
=  \sum_j  c_j  \ket{ j }  =  \ket{ \psi }$.

$\Hil_\Log$ cannot be arbitrarily large,
if the encoding is \emph{faithful}.
A faithful encoding can be reversed,
yielding the exact form of $\ket{ \psi_\Log }$. 
The six-qubit state $\ket{ \psi }$ can faithfully encode
a $\ket{ \psi_\Log }$ of $\leq 6$ qubits,
called \emph{logical qubits}.
The phosphorus nuclear spins---the 
physical DOFs that encode the logical qubits---are 
called \emph{physical qubits}.

Suppose that $\ket{ \psi_\Log }$ is a state of six logical qubits.
We label the logical space's computational basis as
$\Basis^\Log_\comp  =  \Set{
\ket{ 0 0 \ldots 0 },  \ket{ 0 0 \ldots 0 1 } ,  \ldots ,  \ket{ 1 1 \ldots 1 }  }$.
A simple code from $\Basis_\comp$ to $\Basis^L_\comp$
has the form 
\begin{align}
   \label{eq:Simple_code}
   \ket{ m_{1 } , \ldots ,  m_{6 } } 
   \equiv  \ket{ m_1  \ldots m_6 }  \, ,
\end{align}
for $m_1 , \ldots,  m_6  =  0,1$.
For example, all six physical qubits' pointing upward is equivalent to
all six logical qubits' pointing upward:
$\ket{ 0 , \ldots , 0 }  =  \ket{ 0 \ldots 0 }$.

\subsubsection{Transformation of the encoding during 
Posner-molecule creation}
\label{section:How_to_enc_transf}

%

Consider six phosphates that join together, forming a Posner.
The phosphorus nuclei might begin with distinguishable DOFs
(Sec.~\ref{section:When_Pos_form}).
The spins entangle with each other
and with orbital DOFs~\cite{Fisher15,Fisher_17_Quantum}.
The QI $\ket{ \psi_\Log }$ stored in the spins 
``spills'' into the orbital DOFs.

But, we posit, Posner creation maps 
each pre-Posner spin state
to an antisymmetric Posner state deterministically.
The physical qubits change from spins to
spin-and-orbital DOFs.
The physical state's form changes 
from $\ket{ \psi }  \in  \mathbb{C}^{12}$ to 
some $\ket{ \psi' }  \in  \HilPos$.
The Posner state $\ket{ \psi' }$ encodes $\ket{ \psi_\Log }$ faithfully.

Reparameterizing position will prove useful.
We labeled by $\mathbf{x}$ a pre-Posner 
phosphorus nucleus's position.
A Posner's phosphorus nuclei occupy
the centers of cube faces (Fig.~\ref{fig:Geo}).
Let $\mathbf{r}  =  ( r , \varphi, h )$ label
a nucleus's position 
relative to the cube's center.
The cube's size determines each nucleus's 
distance $r$ from the cube center.
Hence we suppress the $r$:
$\ket{ \mathbf{r} }  \equiv  \ket{ \varphi, h }$.
The angle variable is restricted to 
$\varphi  =  \phi ,  \phi + 2 \pi / 3,  \phi + 4 \pi / 3$
(Fig.~\ref{fig:Phi}).
The height variable is restricted to $h  =  h_\pm$ (Fig.~\ref{fig:h_axis}).

Which states can one phosphorus nucleus occupy
when in a Posner?
One might reason na\"ively as follows.
The basis $\Set{ \ket{0} ,  \ket{1} }$ spans 
the nuclear-spin space $\Hil_\nuc^\spin$.
The basis $\Set{ \ket{ \varphi ,  h } }$ spans 
the nuclear-position space $\Hil_\nuc^\orb$.
Hence a product basis spans the nuclear Hilbert space
$\Hil_\nuc  =  \Hil_\nuc^\spin  \otimes  \Hil_\nuc^\orb$:
\begin{align}
   \label{eq:Prod_basis}
   & \{  \ket{ 0 ;  \phi , h_+ } ,   \ket{ 0 ;  \phi ,  h_- } ,
   \ket{ 0 ;  \phi + 2 \pi / 3 ,  h_+ } ,  \ket{ 0 ; \phi + 2 \pi / 3 ,  h_- }  ,
   \nonumber \\ & \;  \;
   \ket{ 0 ;  \phi + 4 \pi / 3 ,  h_+ } ,  \ket{ 0 ; \phi + 4 \pi / 3 ,  h_- }  ,
   \ket{ 1 ;  \phi , h_+ } ,   \ket{ 1 ;  \phi ,  h_- } ,
   \nonumber \\ & \;  \;
   \ket{ 1 ;  \phi + 2 \pi / 3 ,  h_+ } ,  \ket{ 1 ; \phi + 2 \pi / 3 ,  h_- }  ,
   \ket{ 1 ;  \phi + 4 \pi / 3 ,  h_+ } ,  
   \nonumber \\ & \;  \;
   \ket{ 1 ; \phi + 4 \pi / 3 ,  h_- } \}  \, .
\end{align}
We have condensed tensor products 
$\ket{ m }  \otimes  \ket{ \varphi , h }$
into $\ket{ m ; \varphi , h }$.
One might expect the phosphorus nucleus 
to be able to occupy any state in~\eqref{eq:Prod_basis}.
The hextuple of nuclei would be able to occupy a product state
\begin{align}
   \label{eq:Prod_state}
   \ket{ m_1 ;  \varphi_1 , h_1 }  \otimes
   \ldots \otimes
   \ket{ m_6 ;  \varphi_6 , h_6  }  \, .
\end{align}

The nuclei cannot occupy such a state, 
due to their indistinguishability.
The nuclei are fermions.
Hence Posner formation antisymmetrizes the nuclei's joint state. 
We have assumed, in the spirit of~\cite{Fisher15}, 
that Posner creation essentially preserves 
each phosphorus nucleus's $S^{ \zlab }$ 
(Sec.~\ref{section:Formal_Pos_form}). 
Hence the pre-Posner nuclei's 
set $\Set{ m }$ of spin quantum numbers
equals the in-Posner nuclei's set.
But Posner creation prevents any particular $m$
from corresponding, anymore, 
to any particular nucleus.
The nuclei delocalize across the cube-face centers.

Let us mathematize this physics.
The one-nucleus states~\eqref{eq:Prod_basis}
combine into the antisymmetric six-nucleus states
\begin{align}
   \label{eq:Slater}
   & \frac{1}{ \sqrt{ 6! } }
   \sum_{ \alpha = 1}^{ 6! }   \:   
   \bigotimes_{j = 1 }^6   \:
   (-1)^{ \pi_\alpha }
   \ket{ m_{ \pi_\alpha (j) }  ,  \mathbf{r}_{ \pi_\alpha ( j ) } }  \\  \nonumber
   &  :=  \ket{ 
      ( m_{1},  \mathbf{r}_1 )  
      ( m_{2},  \mathbf{r}_2 )  
      ( m_{3},  \mathbf{r}_3 ) ;
      ( m_{4},  \mathbf{r}_4 )  
      ( m_{5},  \mathbf{r}_5 )  
      ( m_{6},  \mathbf{r}_6 )  }  \, .
\end{align}
Each term contains a tensor product of six one-nucleus kets.
Each ket is labeled by one tuple 
$( m_{ \pi_\alpha (j) }  ,  \mathbf{r}_{ \pi_\alpha ( j ) } )$.
No tuple equals any other tuple in the same term,
by Pauli's exclusion principle.
Permuting one term's six tuples yields
another term, to within a minus sign.

$\pi_\alpha$ denotes the $\alpha^\th$ term's permutation.
The permutation's sign,
$(-1)^{ \pi_\alpha }  =  (-1)^{ \text{parity of permutation} }$,
equals the term's sign.\footnote{
A permutation's parity is defined as follows.
Let $\pi_{0}$ denote the first term's permutation.
Consider beginning with $\pi_0$ and 
swapping ket labels pairwise.
Some minimal number $n_\ell$ of swaps
yields permutation $\pi_\ell$.
The parity of $n_\ell$ is the parity of $\pi_\ell$.}
The semicolon in Eq.~\eqref{eq:Slater} separates
the $h_+$ spins from the $h_-$ spins.
\eqref{eq:Slater} is equivalent to 
a Slater determinant~\cite{Ashcroft_76_Solid}.

If not for the Posner's geometry,
two tuples could contain the same position variables.
$\mathbf{r}_1$ could equal $\mathbf{r}_3$, for example,
if $m_{1}$ did not equal $m_{3}$.
But each cube face can house only one phosphate.
The phosphorus nuclei's state occupies the
\emph{no-colliding-nuclei subspace} $\HilPos$ 
of the antisymmetric subspace.

Posner creation, we posit, projects the nuclei's state
onto $\HilPos$. The projector has the form
\begin{widetext}
\begin{align}
      \label{eq:PiMinus} 
      & \PiPos  :=  {\sum}'  \:
      \Big\lvert  \substack{
      ( m_1,  \mathbf{r}_1 )  
      ( m_{2},  \mathbf{r}_2 )  
      ( m_{3},  \mathbf{r}_3 ) ;  \\
      ( m_{4},  \mathbf{r}_4 )  
      ( m_{5},  \mathbf{r}_5 )  
      ( m_{6},  \mathbf{r}_6 ) } 
      \Big\rangle 
      \Big\langle  \substack{
      ( m_{1},  \mathbf{r}_1 )  
      ( m_{2},  \mathbf{r}_2 )  
      ( m_{3},  \mathbf{r}_3 ) ;  \\
      ( m_{4},  \mathbf{r}_4 )  
      ( m_{5},  \mathbf{r}_5 )  
      ( m_{6},  \mathbf{r}_6 )  }
      \Big\rvert   \, .
   \end{align}
\end{widetext}
The sum $\sum'$ runs over values of $(m_1 , \ldots , m_6)$.
The value of $( \mathbf{r}_1 , \ldots \mathbf{r}_6 )
= \LParen ( h_+ , \phi ) , \ldots , ( h_- , \phi + 4 \pi / 3 ) \RParen$
remains invariant throughout the terms.\footnote{
Each pre-Posner spin variable $m$ pairs 
with one position $\mathbf{r}$.
What determines which spin pairs with $\mathbf{r}_1$?
Two factors: (i) the choice of coordinate system and
(ii) the phosphates' pre-Posner positions and momenta.
See App.~\ref{eq:Subtle_HilPos} for details.}
In every term,
the first spin quantum number, $m_1$,
would correspond to the position $\mathbf{r}_1 = ( h_+ , \phi )$.
Different terms correspond to different values $m_1 = 0 , 1$.

Projection by $\PiPos$ applies the map
\begin{align}
   \label{eq:Pos_form_map}
   & \ket{ m_{1} }  \otimes  \ldots  \otimes
   \ket{ m_{6 } }
   \mapsto  \\ \nonumber & 
   \ket{ 
      ( m_{1},  \mathbf{r}_1 )  
      ( m_{2},  \mathbf{r}_2 )  
      ( m_{3},  \mathbf{r}_3 ) ;
      ( m_{4},  \mathbf{r}_4 )  
      ( m_{5},  \mathbf{r}_5 )  
      ( m_{6},  \mathbf{r}_6 )  }  .
\end{align}
The left-hand side (LHS) represents
an element of the computational basis $\Basis_\comp$ for the space 
$( \Hil_\nuc )^{\otimes 6}$ of
the pre-Posner physical qubits. 
The right-hand side (RHS) represents
an element of the computational basis 
$\Basis_\comp^\Posner$ for the space 
$\HilPos$ of the in-Posner physical qubits.

Each pre-Posner state consists of a unique assignment
of $m$-values to nuclei,
a unique distribution of six fixed $m$-values
across six kets.
Similarly, each Posner state consists of a unique assignment
of $m$-values to positions,
a unique distribution of six fixed $m$-values
across six $\mathbf{r}$-values.
Sixty-four pre-Posner $\Basis_\comp$ states exist.
Hence 64 $\Basis_\comp^\Posner$ basis elements
must exist.
A counting argument in App.~\ref{section:Dim_anti}
confirms this conclusion.

Let us combine the map~\eqref{eq:Pos_form_map}
with the simple code~\eqref{eq:Simple_code}.
The result is another simple code.
This code maps between 
(i) elements of the computational basis $\Basis_\comp^\Posner$ 
for the Posner space and
(ii) elements of the computational basis $\Basis^\Log_\comp$
for the logical space:
\begin{align}
   \label{eq:Simple_code2}
   & \ket{ 
      ( m_{1},  \mathbf{r}_1 )  
      ( m_{2},  \mathbf{r}_2 )  
      ( m_{3},  \mathbf{r}_3 ) ;
      ( m_{4},  \mathbf{r}_4 )  
      ( m_{5},  \mathbf{r}_5 )  
      ( m_{6},  \mathbf{r}_6 ) }
   \nonumber \\ & 
   =  \ket{ m_1 m_2 \ldots m_6  }  \, .
\end{align}
Equation~\eqref{eq:Simple_code2} shows how
the QI, initially stored in pre-Posner spin states,
is encoded faithfully in spin-and-orbital states.
We will often replace the physical state's label
(the LHS)
with the logical state's label (the RHS),
to streamline notation.

\subsection{Charge-protected encodings for 
quantum information stored in Posner molecules}
\label{section:Code_constraints}

The computational-basis elements~\eqref{eq:Simple_code2}
are states of spin-and-orbital DOFs.
(Each computational-basis element is 
a multifermion state.
Every multifermion state is, by the spin-statistics theorem, 
an antisymmetric state of
the fermions' spins and positions.)
The Posner's dynamics conserve the spins' states for long times,
Fisher hypothesizes~\cite{Fisher15}.
The dynamics might not conserve the orbital DOFs' states.
Hence the dynamics might not conserve
the states~\eqref{eq:Simple_code2}.

But we posit, guided by~\cite{Fisher15,Fisher_17_Quantum}, 
that the Posner's dynamics conserve certain charges:
(i) the generator $\GC$ of
a permutation operator $\CThree$ (Sec.~\ref{section:C_symm}) and
(ii) the total spin operator's $\zlab$-component,
$S^{ \zlab }_{1 \ldots 6}$
(Sec.~\ref{section:S_z_tot_symm}).
Eigenstates shared by these charges
(Sec.~\ref{section:Eigenbasis})
may be conserved.

The dynamics likely will not map
an eigenstate $\ket{ \psi }$,
associated with eigenvalues $\tau_\psi$ and 
$m_{1 \ldots 6}^{ ( \psi ) }
 :=  \sum_{j = 1}^6  m_j^{ (\psi) }$
of the charges,
into an eigenstate $\ket{ \phi }$
associated with different eigenvalues 
$\tau_\phi$ and $m_{1 \ldots 6}^{ ( \phi ) }$.
These eigenstates may serve as long-lived codewords.
Charge preservation helps ``protect'' such codes. 

We identify a quantum error-detecting code 
partially protected by $\CThree$.
A repetition code is partially protected by 
$S_{ 1 \ldots 6}^{ \zlab }$.
Section~\ref{section:QEC} introduces these codes.

\subsubsection{Conserved charge 1: The generator $\GC$ of
the permutation operator $\CThree$}
\label{section:C_symm}

Consider rotating a Posner counterclockwise,
about the symmetry axis $\hat{z}_\In$,
through an angle $2 \pi / 3$.
The molecule's post-rotation structure (arrangement of atoms)
looks identical to
the original structure~\cite{Treboux_00_Posner,Kanzaki_01_Posner,Yin_03_Posner,Swift_17_Posner}.
The spins (represented loosely by the $m_a$'s) undergo
a counterclockwise cyclic permutation.

Let the operator $\CThree$ represent 
this spin permutation.
$\CThree$ cyclically permutes 
the $h_+$ spins [the first three $m$-values in~\eqref{eq:Simple_code2}] 
while identically permuting the $h_-$ spins 
(the final three $m$-values).
$\CThree$ transforms
the $\Basis_\comp^\Posner$ elements~\eqref{eq:Simple_code2} as
\begin{align}
   \label{eq:CThree_Pos_basis}
   &  \CThree \,  :  \,
   \ket{ m_1 m_2 m_3 m_4 m_5 m_6  }
   \nonumber \\ & \qquad  \equiv
   \ket{  ( m_{1},  \mathbf{r}_1 )  
      ( m_{2},  \mathbf{r}_2 )  
      ( m_{3},  \mathbf{r}_3 ) ;
      ( m_{4},  \mathbf{r}_4 )  
      ( m_{5},  \mathbf{r}_5 )  
      ( m_{6},  \mathbf{r}_6 ) }
   \nonumber \\ & \quad  \mapsto 
   \ket{  ( m_3,  \mathbf{r}_1 )  
      ( m_{1},  \mathbf{r}_2 )  
      ( m_{2},  \mathbf{r}_3 ) ;
      ( m_{6},  \mathbf{r}_4 )  
      ( m_{4},  \mathbf{r}_5 )  
      ( m_{5},  \mathbf{r}_6 ) }
    \nonumber \\ & \qquad
  \equiv  \ket{  m_3 m_1 m_2 m_6 m_4 m_5  }  \, .
\end{align}
The $\Basis_\comp^\Posner$ elements~\eqref{eq:Simple_code2}
are not $\CThree$ eigenstates.

But $\CThree$ eigenstates can be constructed.
We adopt F\&R's notation for the eigenvalues,
\begin{align}
   \label{eq:CThree_eig}
   & \omega^\tau  \, ,  \quad \text{wherein}  \quad
   \omega  :=  e^{ i 2 \pi / 3 }  
   \; \text{and} \;  \\ \nonumber &
   \tau  =  0 , 1 , 2   \quad
   \text{or, equivalently,}  \quad  \tau = 0,  \pm 1  \, .
\end{align}
F\&R call $\tau$ a three-level ``pseudospin.''
We call $\tau$ the eigenvalue of the observable $\GC$
that generates $\CThree$.\footnote{
\label{footnote:Pseudospin}
A pseudospin is a physical DOF
that transforms according to a certain rule.
$\tau$ is, rather, the eigenvalue of an observable.
Suppose that $\tau$ were a three-level quantum pseudospin.
$\tau$ would occupy a quantum state 
in some three-dimensional effective Hilbert space $\Hil_{\rm pseudo}$.
No such space can be associated uniquely with a Posner,
to our knowledge.
Rather, the Posner Hilbert space $\HilPos$ has dimensionality 64.
$\HilPos$ equals a direct sum of 
the three $\GC$ eigenspaces:
$\HilPos = \Hil_{\tau = 0}  \oplus  \Hil_{ \tau = 1 }  \oplus \Hil_{ \tau = 2 }$.
Each subspace is degenerate.
Hence no subspace can serve as one element in
a basis for any $\Hil_{\rm pseudo}$.
One could conjure up a $\Hil_{\rm pseudo}$ by
choosing one state $\ket{ \tau{=}0 } \in \Hil_{ \tau = 0}$,
one $\ket{ \tau{=}1 } \in \Hil_{ \tau = 1}$, 
and one $\ket{ \tau{=}2 } \in \Hil_{ \tau = 2}$;
then constructing
$\Hil_{\rm pseudo}  =  {\rm span}\Set{ 
\ket{ \tau{=}0 } ,  \ket{ \tau{=}1 } ,  \ket{ \tau{=}2 }  }$.
We do so in Sections~\ref{section:Qutrit_code} and~\ref{section:Stick_apps}.
But the choice of $\ket{ \tau{=}0 }$ is nonunique,
as is the choice of $\ket{ \tau{=}1 }$, as is the choice of $\ket{ \tau{=}2 }$.
Hence no unique three-level Hilbert space corresponds to a Posner,
to our knowledge.
Hence $\tau$ appears not to label a unique quantum pseudospin.}
The general form of a $\CThree$ eigenstate
appears in~\cite{Fisher_17_Quantum}.
F\&R use second quantization.

We translate into QI.
We also extend~\cite{Fisher_17_Quantum} 
by characterizing the eigenspaces of $\CThree$
and by identifying a useful basis for each eigenspace
(Sec.~\ref{section:Eigenbasis}).\footnote{
A related characterization, and alternative bases,
appeared in~\cite{Swift_17_Posner},
shortly after the present paper's initial release.}
The $\tau = 0$ eigenspace has degeneracy 24;
the $\tau = 1$ eigenspace, degeneracy 20;
and the $\tau = -1$ eigenspace, degeneracy 20.
The $\tau = 0$ eigenspace will play an important role in
Posner resource states for universal quantum computation
(Sec.~\ref{section:AKLT}).

%
%
%

\subsubsection{Conserved charge 2: 
The total-spin operator $S^{ \zlab }_{1 \ldots 6}$}
\label{section:S_z_tot_symm}

Fisher conjectures that Posners' phosphorus nuclear spins 
have long coherence times~\cite{Fisher15}.
We infer that the Posner Hamiltonian 
$H_\Posner$ conserves
$S^{ \zlab }_{1 \ldots 6}  
=  \sum_{ j = 1 }^6  
\id^{ \otimes (j - 1) }  \otimes  S^{ \zlab }_{j}  
\otimes  \id^{ \otimes (6 - j) }$.
The total magnetic spin quantum number,
$m_{1 \ldots 6 }  =  \sum_{j = 1}^6 m_j$,
remains constant.
Appendix~\ref{section:Z_Cons} supports this argument with
conjectured interactions between
a Posner's phosphorus nuclear spins.
Interactions with the environment are expected
to conserve $S^{ \zlab }_{1 \ldots 6}$ approximately (for long times).

We decompose $\HilPos$ into composite-spin subspaces
in App.~\ref{section:Group_thry}.
That appendix also reviews 
the addition of quantum angular momentum.

\subsubsection{Eigenbasis shared by the conserved charges}
\label{section:Eigenbasis}

We introduced the computational basis 
$\Basis_\comp$ for $\HilPos$ in Eq.~\eqref{eq:Slater}.
Most $\Basis_\comp$ elements transform nontrivially under $\CThree$
[Eq.~\eqref{eq:CThree_Pos_basis}].
The Posner dynamics conserve $\CThree$.
So, too, would the dynamics ideally conserve quantum codewords.
We therefore seek a useful $\CThree$ eigenbasis
from which to construct QEC codes.

The $\CThree$ eigenspaces have degeneracies.
Which basis should we choose for each eigenspace?
A basis shared with $S^{ \zlab }_{1 \ldots 6}$,
the other conserved charge.

Yet $\CThree$ and $S^{ \zlab }_{1 \ldots 6}$ do not form
a complete set of commuting observables (CSCO)~\cite{Cohen_Tannoudji_91_Quantum}.
Many eigenbases of $\CThree$ 
are eigenbases of $S^{ \zlab }_{1 \ldots 6}$.
Another operator is needed to break the degeneracy,
to complete the CSCO.
We choose the spin-squared sum\footnote{
Swift \emph{et al.} choose
the total-spin operator $\mathbf{S}_{1 \ldots 6}^2$
and the Hamiltonian.
Their choice can be used to define 
an alternative computational basis.
Our choice clarifies the preparation of 
the universal quantum-computation resource state
in Sec.~\ref{section:AKLT}.
}
\begin{align}
      \label{eq:S_tot_op}
      \mathbf{S}_{123}^2  +  \mathbf{S}_{456}^2  
      &  \equiv  ( \mathbf{S}_1  +  \mathbf{S}_2  +  \mathbf{S}_3 )^2
      +  ( \mathbf{S}_4  +  \mathbf{S}_5  +  \mathbf{S}_6 )^2  \\
      & \equiv     \left( 
      \sum_{ a = 1 }^3  \id^{ \otimes ( a - 1 ) }  \otimes
      \mathbf{S}_a  \otimes  \id^{ \otimes ( 3 - a ) }
      \right)^2 
      \nonumber \\ & \quad 
      +  \left(   \sum_{ a = 4 }^6  \id^{ \otimes ( a - 4 ) }  \otimes
      \mathbf{S}_a  \otimes  \id^{ \otimes ( 6 - a ) }
      \right)^2    \, .
\end{align}
The CSCO consists of $\mathbf{S}^2_{123}$;  $\mathbf{S}^2_{456}$;
the $\CThree$ analog that permutes just qubits 1, 2, and 3;
the $\CThree$ analog that permutes just qubits 4, 5, and 6;
$S^{ \zlab }_{1 2 3}$;  and  $S^{ \zlab }_{4 5 6}$.

Geometry and measurement-based quantum computation
(Sec.~\ref{section:Posner_AKLT})
motivate the choice of $\mathbf{S}_{123}^2  +  \mathbf{S}_{456}^2$:
A Posner contains two triangles of spins
(Fig.~\ref{fig:h_axis}).
The positions in the $h_+$ triangle are labeled
$\mathbf{r}_1$, 
$\mathbf{r}_2$, and 
$\mathbf{r}_3$. 
Hence the first three tuples in Eq.~\eqref{eq:Simple_code2}
correspond to the $h_+$ triangle.
Hence the magnetic spin quantum numbers
$m_1$, $m_2$, and $m_3$ may be viewed as
occupying the $h_+$ triangle.
These spins' joint state is equivalent to a three-qubit logical state,
$\ket{ m_1 m_2 m_3 }$.
An analogous argument concerns $h_-$.
Hence the antisymmetric state~\eqref{eq:Simple_code2}
is equivalent to
a product of two three-logical-qubit states:\footnote{
In greater detail,
\begin{align}
   \label{eq:Trio_help1}
   & \ket{ ( m_{1},  \mathbf{r}_1 )  
      ( m_{2},  \mathbf{r}_2 )  
      ( m_{3},  \mathbf{r}_3 ) ;
      ( m_{4},  \mathbf{r}_4 )  
      ( m_{5},  \mathbf{r}_5 )  
      ( m_{6},  \mathbf{r}_6 ) }
   \\ & \qquad
   \label{eq:Trio_help2}
   \equiv  \ket{ m_1 m_2 m_3 m_4 m_5 m_6  }
   \\ &  \qquad
   \label{eq:Trio_help3}
   \equiv  \ket{ m_1 } \ket{ m_2 } \ket{ m_3 } \ket{ m_4 } \ket{ m_5 } \ket{ m_6 }
   \\ &  \qquad
   \label{eq:Trio_help4}
   =  ( \ket{ m_1 } \ket{ m_2 } \ket{ m_3 } ) ( \ket{ m_4 } \ket{ m_5 } \ket{ m_6 } )
   \\ &  \qquad
   \label{eq:Trio_help5}
   \equiv  \ket{ m_1 , m_2 , m_3 }    \ket{ m_4 , m_5 , m_6 }  
   \\ &  \qquad
   \label{eq:Trio_help6}
   \equiv  \ket{ m_1 m_2 m_3 }  \ket{ m_4 m_5 m_6 }  \, .
\end{align}
Equation~\eqref{eq:Trio_help2} is equivalent to
Eq.~\eqref{eq:Simple_code2}.
Equation~\eqref{eq:Trio_help3} is equivalent to 
Eq.~\eqref{eq:Simple_code}.
Equation~\eqref{eq:Trio_help4} follows from
the tensor product's associativity.
Equation~\eqref{eq:Trio_help5} consists of a rewriting
with new notation.
Equation~\eqref{eq:Trio_help6} is analogous to Eq.~\eqref{eq:Simple_code}.
} 
\begin{align}
   \label{eq:Trios}
   & \ket{ ( m_{1},  \mathbf{r}_1 )  
      ( m_{2},  \mathbf{r}_2 )  
      ( m_{3},  \mathbf{r}_3 ) ;
      ( m_{4},  \mathbf{r}_4 )  
      ( m_{5},  \mathbf{r}_5 )  
      ( m_{6},  \mathbf{r}_6 ) }
   \nonumber \\ & \qquad
   \equiv  \ket{ m_1 m_2 m_3 }  \ket{ m_4 m_5 m_6 }  \, .
\end{align}
The trios function logically as independent units.
Such trios can be used to prepare
universal quantum-computation resource states
(Sec.~\ref{section:Posner_AKLT}).
Hence the spin-operator trios
in Eq.~\eqref{eq:S_tot_op}.

%
%
\begin{table*}[t] 
\begin{center} 
\begin{tabular}{|M{.8cm}|M{4cm}|M{.6cm}|M{.75cm}|M{.4cm}|N}
    \hline
        State
   &  Decomposition
   &   $S_{123}$
   &  $m_{{123}}$
   &  $\tau$
   &
   \\[5pt]  \hline  \hline
       $\ket{ 0 0 0 }$
   &  $\ket{ 0 0 0 }$
   &  $3/2$
   &  $3/2$
   &  0
   &
   \\[5pt]  \hline
       $\ket{W}$
   &  $\frac{1}{ \sqrt{3} } ( \ket{ 1 0 0 }  +  \ket{ 0 1 0 }  +  \ket{ 0 0 1 } )$
   &  $3/2$
   &  $1/2$
   &  0
   &
   \\[5pt]  \hline  
       $\ket{ \bar{W} }$
   &  $\frac{1}{ \sqrt{3} } ( \ket{ 0 1 1 }  +  \ket{ 1 0 1 }  +  \ket{ 1 1 0 } )$
   &  $3/2$
   &  $ - 1/2$
   &  0
   &
   \\[5pt]  \hline  
       $\ket{ 1 1 1 }$
   &  $\ket{ 1 1 1 }$
   &  $3/2$
   &  $- 3 / 2$
   &  0
   &
   \\[5pt]  \hline  \hline
       $\ket{ \omega }$
   &  $\frac{1}{ \sqrt{3} } 
         ( \ket{ 1 0 0 }  +  \omega^2 \ket{ 0 1 0 }  +  \omega \ket{ 0 0 1 } )$
   &  $1/2$
   &  $1/2$
   &  $1$
   &
   \\[5pt]  \hline  
       $\ket{ \bar{\omega} }$
   &  $\frac{1}{ \sqrt{3} } 
         ( \ket{ 0 1 1 }  +  \omega^2 \ket{ 1 0 1 }  +  \omega \ket{ 1 1 0 } )$
   &  $1/2$
   &  $- 1 /2$
   &  $1$
   &
   \\[5pt]  \hline  \hline
       $\ket{ \omega^2 }$
   &  $\frac{1}{ \sqrt{3} } ( \ket{ 1 0 0 } + \omega \ket{ 0 1 0 } + \omega^2 \ket{ 0 0 1 } )$
   &  $1/2$
   &  $1/2$
   &  $2$
   &
   \\[5pt]  \hline  
       $\ket{ \overline{ \omega^2 } }$
   &  $\frac{1}{ \sqrt{3} } ( \ket{ 0 1 1 } + \omega \ket{ 1 0 1 } + \omega^2 \ket{ 1 1 0 } )$
   &  $1/2$
   &  $-1/2$
   &  $2$
   &
   \\[5pt]  \hline  
\end{tabular}
\caption{\caphead{Symmetric basis for a trio of qubits:}
A Posner molecule consists of two triangles of spins (Fig.~\ref{fig:Geo}).
In accordance with Eq.~\eqref{eq:Trios},
each triangle functions as
a trio of logical qubits.
The three physical qubits correspond to 
an eight-dimensional Hilbert space, $\mathbb{C}^6$.
A useful basis is an eigenbasis shared by
the conserved charges,
$\CThree$ (a permutation operator) and
$S^{ \zlab }_{123}$
(the $z$-component, relative to the lab's $\hat{z}_\lab$-axis,
of the total spin). 
These operators share many bases.
The eigenbasis shared also by $\mathbf{S}_{123}^2$ proves useful
in the preparation of universal quantum-computation resource states 
(Sec.~\ref{section:AKLT}).
$\tau$ describes, here, how a triangle transforms under
the permutation represented by $\CThree$.
}
\label{table:Trio_states}
\end{center}
\end{table*}

Each qubit trio corresponds to a Hilbert space $\mathbb{C}^{6}$.
Let us focus on qubits 1-3, for concreteness.
$\CThree$, $S^{ \zlab }_{123 }$, and $\mathbf{S}_{123}^2$
share the basis in Table~\ref{table:Trio_states}.
Each basis element is symmetric with respect to 
cyclic permutations of the three logical qubits.

Tensoring together two one-triangle states
yields a state of a Posner's phosphorus nuclear spins:
$\ket{ 0 0 0 } \ket{ 0 0 0 },  \ket{ 0 0 0 }  \ket{ W }  , \ldots , 
\ket{ \overline{ \omega^2 } }  \ket{ \overline{ \omega^2 } }$.
Sixty-four such states exist.
We classify them with quantum numbers 
in App.~\ref{section:C_eigenspaces}.

We have pinpointed an eigenbasis
shared by the conserved charges.
The Posner dynamics are expected not to map
states in one charge sector to states in another.
Hence different-sector states suggest themselves as
quantum codewords.
We present partially charge-protected QECD codes next.

\subsubsection{Quantum error-detecting and -correcting codes
accessible to Posner molecules}
\label{section:QEC}


We exhibit two codes formed from states accessible to Posners.
Each codeword is an eigenstate 
of a conserved charge, 
$\CThree$ or $S^{ \zlab }_{1 \ldots 6 }$.
Each code's codewords correspond to 
distinct eigenvalues of the charge.
Hence the Posner dynamics
likely do not map
any codeword into any other.

This section demonstrates the compatibility of 
(i) protection of phosphorus nuclei by Posner dynamics with
(ii) protection of QI with quantum error correction.
Presenting codewords that (i) enjoy (partial) protection by Posner dynamics
and (ii) satisfy the quantum error-detection and -correction criteria~\cite{Knill_97_Theory,Bennett_96_Mixed,Kribs_05_Unified,Kribs_05_Operator,Nielsen_07_Algebraic,Preskill_99_QEC}
suffices.
We do not specify operations via which 
errors can be detected or corrected.
In the language of QI theory, we present an existence proof
and a partial construction.
Augmenting the construction remains an opportunity
discussed in Sec.~\ref{section:Outlook}.

Section~\ref{section:Qutrit_code} introduces a quantum error-detecting code.
One Posner, we show, can encode one logical qutrit.
The code detects one arbitrary physical-qubit error.
Section~\ref{section:Pos_code_rep} shows how to implement
a repetition code with Posner states.
The code corrects two bit flips.

More Posner codes, we expect, await discovery.
For example, one conserved charge
partially protects each of our codes.
The Posner's dynamics should prevent
any codeword from evolving into any other.
But the molecular dynamics could map one codeword
outside the codespace.
Our two codes therefore illustrate protection by 
the Posner's conserved quantities.
Full protection by conserved quantities 
is left for future research.
This opportunity and others are detailed in Sec.~\ref{section:Outlook}.

We have already discussed an encoding
of logical states in physical systems 
(Sec.~\ref{section:Encodings}).
Earlier, the logical Hilbert space $\Hil_\Log$ shared
the physical Hilbert space's dimensionality, 64.
Section~\ref{section:Encodings} concerned
a bijective, injective map between the spaces.
QECD encodes a small logical space
in a larger physical space.
Notation will reflect the distinction between 
Sec.~\ref{section:Encodings} and QECD:
Script subscripts $\Logg$ (as in $\Hil_\Logg$) will replace 
the Roman $\Log$
(as in $\Hil_\Log$).
QECD is reviewed in App.~\ref{section:Backgrnd_QEC}.

\paragraph{Qutrit error-detecting code
formed from Posner-molecule states:}
\label{section:Qutrit_code}
One Posner, we show, can encode one logical qutrit.
The code detects arbitrary single-physical-qubit errors.
The physical qubits are
the spin-and-orbital DOFs of Sec.~\ref{section:How_to_enc_transf}.

The code has the form 
\begin{align} 
   \label{eq:Qutrit_code}
   \Hil_\Logg^\qutrit  =  {\rm span} \Set{   
   \ket{ 0_\Logg } ,   \ket{ 1_\Logg } ,   \ket{ 2_\Logg }   }  \, ,
\end{align}
wherein
\begin{align}
   & \ket{ 0_\Logg }  =  \frac{1}{ \sqrt{2} } 
   ( \ket{ W }  \ket{ \bar{W} }  -  \ket{ \bar{W} }  \ket{W} ) \, , \\
   & \ket{ 1_\Logg }  =  \frac{1}{ \sqrt{2} } 
   ( \ket{ \omega^2 }  \ket{ \overline{ \omega^2 } }  
     -  \ket{ \overline{ \omega^2 } }  \ket{ \omega^2 } ) \, , 
   \quad \text{and}  \quad  \\
   & \ket{ 2_\Logg }  = \frac{1}{ \sqrt{2} } 
   ( \ket{ \omega }  \ket{ \bar{ \omega } }  
   -  \ket{ \bar{ \omega } }  \ket{ \omega } )  \, .
\end{align}
Each logical state $\ket{ j_\Logg }$ occupies the $\tau = j$ subspace.

The codewords satisfy
the two quantum error-detection criteria~\cite{Knill_97_Theory,Bennett_96_Mixed,Kribs_05_Unified,Kribs_05_Operator,Nielsen_07_Algebraic,Preskill_99_QEC}.
First, the states are locally indistinguishable:
\begin{align}
   & \langle j_\Logg | \sigma^x |  j_\Logg  \rangle
   =  \langle j_\Logg | \sigma^y |  j_\Logg  \rangle
   = 0  \, , \quad \text{and} \\
   & \langle  j_\Logg  | \sigma^z |  j_\Logg  \rangle
   =  \frac{1}{12} \, 
\end{align} 
for all $j$.
That is, the codewords satisfy the diagonal criterion.
(See App.~\ref{section:Backgrnd_QEC} for background.)
Second, the codewords satisfy the off-diagonal criterion,
\begin{align}
   \langle j_\Logg | \sigma^\alpha | k_\Logg \rangle  =  0
   \qquad  \forall j \neq k  \, ,
   \quad \forall \alpha  = x , y , z \, ,
\end{align}
by direct calculation.

\paragraph{Repetition code
formed from Posner-molecule states:}
\label{section:Pos_code_rep}
%
%
The repetition code originated in classical error correction~\cite{Watrous_06_Lecture}.
Each logical bit is cloned until $n$ copies exist:
$0  \mapsto  \underbrace{0 0 \ldots 0}_n$,
and $1  \mapsto  \underbrace{1 1 \ldots 1}_n$.
Suppose that errors flip under half the bits.
For example, $000000$ may transform into $011000$.
One decodes the bit string by counting the zeroes, counting the ones,
and following majority rule.
More physical bits end as 0s than as 1s in our example.
A logical zero, the receiver infers, was likely sent.

The repetition code can be translated into quantum states.
For example, let
$\Hil_\Logg^\rep  =  \Set{
    \ket{ 0_\Logg }  ,  \ket{ 1_\Logg } }  \, ,$
wherein
$\ket{ 0_\Logg }  =  \ket{ 000000 }$ and
$\ket{ 1_\Logg }  =  \ket{ 111111 }  \, .$
(As we are defining a new code,
we are defining $\ket{ 0_\Logg }$ and $\ket{ 1_\Logg }$ anew.)
This code corrects two $\sigma^x$ errors. 
But each codeword is unentangled.\footnote{
More precisely, the element $\ket{000000}$
of the computational basis for the logical space 
(Sec.~\ref{section:How_to_enc_transf})
is unentangled.
The spin-and-orbital state represented by $\ket{000000}$
[by Eq.~\eqref{eq:Simple_code2}] is entangled.}
Hence $\Hil_\Logg^\rep$ fails to satisfy
the off-diagonal error-detection criterion,
\begin{align}
   \langle j_\Logg | \sigma^\alpha  \sigma^\beta | k_\Logg \rangle  =  0
   \qquad  \forall j \neq k  \, ,
\end{align}
whenever $\alpha \neq x$ and/or $\beta \neq x$.

\subsection{The model of Posner quantum computation}
\label{section:Abstract_logic}


Fisher has conjectured that
several physical processes
occur in biofluids~\cite{Fisher15}.
We reverse-engineer two more.
We abstract away the physics,
identifying the computations that the processes effect.
We call the computations \emph{Posner operations}.\footnote{
We occasionally call the physical processes, too,
``Posner operations.''}
The operations form 
a model of quantum computation,
\emph{Posner quantum computation}.

The model's operations, we will show, can be used 
(i) to teleport QI incoherently
and (ii) to prepare, efficiently, 
universal resource states for 
measurement-based quantum computation
(Sections~\ref{section:AKLT}-\ref{section:AKLT}).
Whether the model's operations can realize universal quantum computation 
remains an open question (Sec.~\ref{section:diV}).

Posner quantum computation is defined
in Sec.~\ref{section:Abstract2}.
The model is analyzed in Sec.~\ref{section:Analyze_abstract}.
We discuss the model's ability to entangle qubits
and the control required to perform QI-processing tasks.
Fisher's narrative~\cite{Fisher15} is also cast
as a quantum circuit.

\subsubsection{Definition of Posner quantum computation}
\label{section:Abstract2}

Terminological notes are in order.
When discussing physical processes, we discuss 
phosphorus nuclear spins, spin-and-orbital DOFs, and Posners.
When discussing logical DOFs, we discuss qubits.
A circuit-diagram element represents each operation
(Figures~\ref{fig:operation1}-\ref{fig:operation7}):

\begin{enumerate}[leftmargin=*]

   \item  \label{item:Bell}
   \textbf{Singlet-state preparation} (Fig.~\ref{fig:operation1}):
   \emph{Arbitrarily many singlets $\ket{ \Psi^- }$ can be prepared.}
   Singlets are prepared when an enzyme hydrolyzes diphosphates into
   entangled phosphate pairs (Sec.~\ref{section:When_Pos_form}).\footnote{
   Biofluids might prepare phosphorus nuclear spins
   in nonsinglet states.
   For example, one phosphate might detach from ATP,
   leaving adenosine diphosphate (ADP).
   Identifying the phosphate's quantum state
   would require physical modeling outside this paper's scope.
   Therefore, we restrict our focus to singlets.}

   \begin{figure}[tb]
   \centering
   \includegraphics[width=.15\textwidth, clip=true]{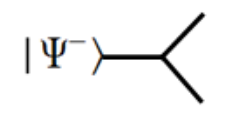}
   \caption{\caphead{Circuit-diagram element that represents
   singlet-state preparation (operation~\ref{item:Bell}).} }
   \label{fig:operation1}
   \end{figure}

These singlets are prepared differently than
in conventional quantum circuits. 
Conventionally, one prepares two qubits in the state $\ket{ 0 }^{ \otimes 2 }$;
performs a Hadamard\footnote{
The \emph{Hadamard gate} $H$ transforms one qubit~\cite{NielsenC10}.
In terms of Pauli operators, $H = \frac{1}{ \sqrt{2} }  \:  ( \sigma_x  +  \sigma_z )$.
The gate has a geometric interpretation
expressed in terms of the Bloch sphere:
The state rotates through $180^\circ$ about the axis
$\frac{1}{ \sqrt{2} }  \:  ( \hat{x}  +   \hat{z} )$.}
on the first qubit;
and performs a CNOT,\footnote{
\label{footnote:CNOT}
The \emph{CNOT}, or \emph{controlled-not}, gate
transforms two qubits~\cite{NielsenC10}.
One qubit is called the \emph{control},
and one is called the \emph{target}.
If the control occupies the state $\ket{0}$,
the CNOT preserves the target's state.
If the first qubit occupies $\ket{1}$,
the target evolves under $\sigma_x$.
The CNOT has the form
$\ketbra{0}{0}  \otimes  \id  +  \ketbra{1}{1}  \otimes  \sigma_x$.}
controlling on the first qubit: 
$\text{CNOT} ( H \otimes \id ) \ket{00}  =  \ket{ \Psi^- }$.

In contrast, Fisher posits that 
enzymes prepare singlets
by projective measurements~\cite{Fisher15}.
We formalize Fisher's statement as follows.
A diphosphate's phosphorus nuclear spins
occupy some state $\rho_{\rm diphos}$. 
The diphosphate enters a pyrophosphatase enzyme.
The enzyme measures the PVM 
$\Set{ \ketbra{ \Psi^- }{ \Psi^- } ,  \id  -  \ketbra{ \Psi^- }{ \Psi^- } }$.

Suppose that the diphosphate separates into
two disconnected phosphates.
The spins' state has been projected with 
$\ketbra{ \Psi^- }{ \Psi^- }$.

Suppose, instead, that the diphosphate leaves the enzyme uncleaved.
The second possible measurement outcome has obtained. 
The diphosphate cannot form a Posner molecule
with other ions.
Hence the diphosphate cannot participate in quantum cognition.
Hence the diphosphate plays no role in
Posner quantum computation.
Hence the PVM's $\id  -  \ketbra{ \Psi^- }{ \Psi^- }$ outcome
plays no role.
Any diphosphate that remains uncleaved
``is discarded,'' in QI language.
The $\ketbra{ \Psi^- }{ \Psi^- }$ outcome is classically postselected on.
Classical postselection provides no superquantum computational power;
see footnote~\ref{footnote:Postselect}.

In summary, Posner operations include
the preparation of $\ket{ \Psi^- }$.
Quantum-cognition systems prepare $\ket{ \Psi^- }$ 
by measuring a PVM nondestructively,
then postselecting classically on the ``yes'' outcome.
``No''-outcome ions 
do not participate in later chemical events of interest. 

   
   \begin{figure}[tb]
   \centering
   \includegraphics[width=.4\textwidth, clip=true]{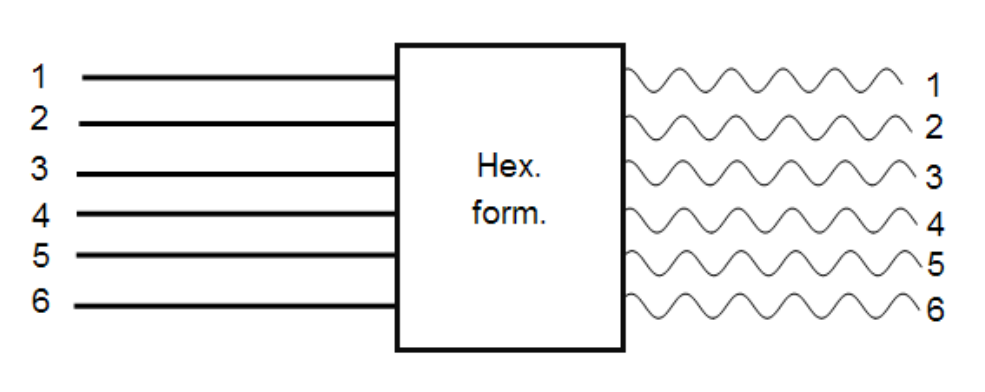}
   \caption{\caphead{Circuit-diagram element that represents
   hextuple formation (operation~\ref{item:Form_Pos}):} 
   Straight lines represent qubits not in hextuples.
   Each wavy line represents a qubit in a hextuple
   (that is not paired with any other hextuple
   as a result of operation~\ref{item:Meas}).
   As the lines' labels show,
   the circuit element is defined
   as preserving the qubits' ordering.}
   \label{fig:operation3}
   \end{figure}
   
   %
   
   %
   \item  \label{item:Form_Pos}
   \textbf{Hextuple formation} (Fig.~\ref{fig:operation3}):
   \emph{Qubits can group together in hextuples (groups of six).
   Hextuple formation evolves the logical qubits trivially, 
   under the operator $\id$.
   But hextuple formation associates the qubits with a geometry
   and with an observable $\GC$.}
   
   Logical qubits form hextuples as ions bind together, forming Posners.
   A logical qubit can occupy, at most, one hextuple.
   Section~\ref{section:How_to_enc_transf} explains why
   hextuple formation fails to change logical qubits' states:
   The spins' state changes,
   suggesting that the logical qubits' state changes. 
   But the logical information's physical encoding changes, too.
   
   Hextuple creation impacts the logical system in three ways:
   (i) Each hextuple has an observable $\GC$.
   (ii) Hextuple creation induces a geometry 
   that influences operation~\ref{item:Six_rot}.
   (iii) The six logical qubits' Hilbert space transforms
   from $\mathbb{C}^{12}$ to the isomorphic $\HilPos$.
   Let us detail these three effects.
   
   First, creating a hextuple creates an observable $\GC$
   (Sec.~\ref{section:C_symm}).
   $\GC$ has eigenvalues $\tau = 0, 1 , 2$ 
   (equivalently, $\tau = 0, \pm 1$).
   $\tau$ impacts operation~\ref{item:Meas}.
   
   Second, hextuple creation induces a geometry.
   Each logical qubit is assigned to a cube face,
   in accordance with Sec.~\ref{section:How_to_enc_transf}.
   The six qubits can be distributed across the six faces
   in any of $6!$ ways.
   Physically, different assignments follow from
   different pre-Posner orbital states (App.~\ref{eq:Subtle_HilPos}).
   The six qubits form two triangles, called \emph{trios} below, 
   in accordance with Fig.~\ref{fig:Geo}.
   This geometry limits the single-qubit unitaries
   that can evolve the six qubits (operation~\ref{item:Six_rot}).
   The geometry also influences our construction of
   universal quantum-computation resource states
   (Sec.~\ref{section:AKLT}).
   
   Third, hextuple creation changes the system's Hilbert space
   from $( \mathbb{C}^2 )^{ \otimes 6 }$ to
   $\HilPos$ (Sec.~\ref{section:How_to_enc_transf}).
   $\HilPos$ is isomorphic to $( \mathbb{C}^2 )^{ \otimes 6 }$,
   as the map~\eqref{eq:Simple_code2} is injective and bijective.
   Hence we will keep referring to the logical space as 
   $( \mathbb{C}^2 )^{ \otimes 6 }$.

   \item  \label{item:In_hex} 
   \textbf{Hextuple unitaries:}
   \emph{Any hextuple can undergo any sequence of any instances 
   of the following operations:}
   \begin{enumerate}[leftmargin=18pt]
   
   \begin{figure}[tb]
   \centering
   \includegraphics[width=.4\textwidth, clip=true]{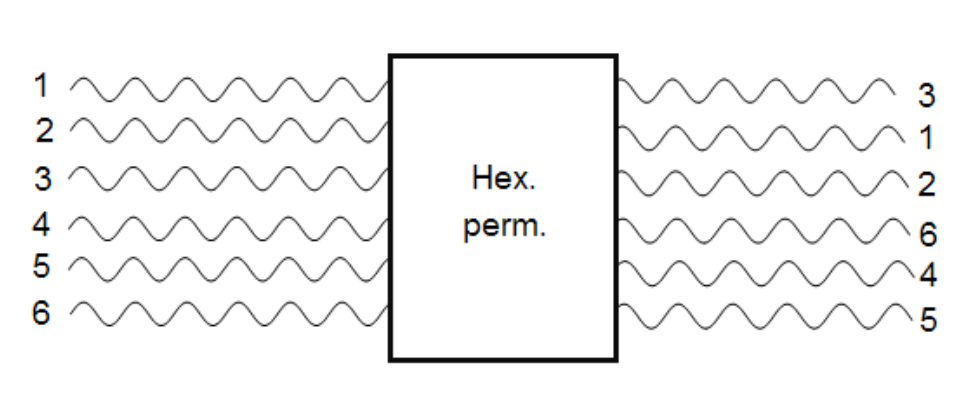}
   \caption{
   \caphead{Circuit-diagram element that represents
   hextuple permutation (operation~\ref{item:Swap}).}}
   \label{fig:operation4a}
   \end{figure}

   \item  \label{item:Swap} 
   \textbf{Hextuple permutation} (Fig.~\ref{fig:operation4a}):
   \emph{The qubits can undergo 
   the permutation $(231)(564)$
   [Eq.~\eqref{eq:CThree_Pos_basis}].} 
   The qubits are permuted as the Posner rotates
   about its symmetry axis, $\hat{z}_\In$,
   through an angle $2 \pi / 3$.
   Posners rotate while tumbling in the fluid.
   
   \begin{figure}[tb]
   \centering
   \includegraphics[width=.25\textwidth, clip=true]{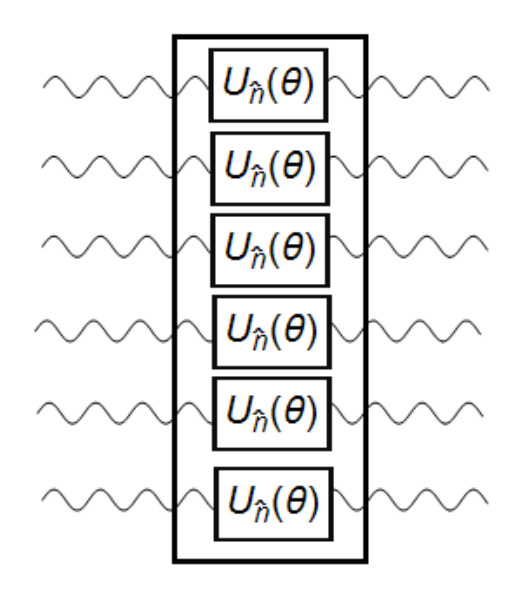}
   \caption{
   \caphead{Circuit-diagram element that represents
   hextuple-coordinated single-qubit rotations
   (operation~\ref{item:Six_rot}).}
   $\hat{n}$ denotes the axis rotated about.
   $\theta$ denotes the angle rotated through.}
   \label{fig:operation4b}
   \end{figure}

   \item   \label{item:Six_rot} 
   \textbf{Hextuple-coordinated single-qubit rotations}
   (Fig.~\ref{fig:operation4b}):
   \emph{The qubits in a hextuple can undergo 
   identical single-qubit rotations simultaneously.}
   The qubits are expected to rotate,
   consisting essentially of spin magnetic moments in  
   the magnetic field generated as neurons spike.
   
   Fisher's narrative allows for, though does not include,
   this operation.
   We approximate the magnetic field as
   roughly constant over the Posner's length scale,
   $\sim 1$ nm~\cite{Luo_11_Modeling}.
   The formalism can easily be generalized to accommodate 
   short-distance field fluctuations, however.
   The unitary
   $[ U_{ \hat{n} } (\theta ) ]^{ \otimes 6 }$ 
   rotates the qubits through an angle $\theta$
   about an axis $\hat{n}$ 
   relative to the lab frame.
   Angles of up to $\theta \approx \pi$ might be reached.
   Smaller angles are expected to be typical, however.
   Section~\ref{section:diV} and
   App.~\ref{section:One_Qubit_Gates} detail
   the rotation mechanism and scales. 
   
   \end{enumerate}
   
   \begin{figure}[tb]
   \centering
   \includegraphics[width=.3\textwidth, clip=true]{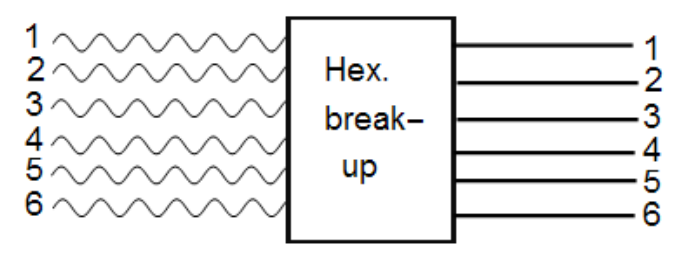}
   \caption{\caphead{Circuit-diagram element that represents
   hextuple break-up
   (operation~\ref{item:pH}):}
   Each coil represents a qubit in a dodectuple.       
   Each straight line represents a qubit 
   not grouped with any other qubits.}
   \label{fig:operation7}
   \end{figure}

   \item  \label{item:pH}
   \textbf{Hextuple break-up}
   (Fig.~\ref{fig:operation7}):
   \emph{One hextuple can break down into its constituents: 
   The qubits can cease to correspond to geometries
   or to observables $\GC$.}

   Different regions of the body have different pH's and
   different magnesium-ion (Mg$^{2+}$) concentrations.
   A Posner can migrate to a region packed with
   H$^+$ and/or with Mg$^{2+}$.
   These ions can bind to $\PO$, as $\Ca\Two$ can.
   The higher the H$^+$ and Mg$^{2+}$ concentrations,
   the more H$^+$ and Mg$^{2+}$ ions dislodge
   Posners' $\Ca\Two$ ions~\cite{Fisher_17_Personal}.
   The dislodging hydrolyzes the molecules.
   
   Fisher's narrative allows for, though does not include, hextuple break-up.
   Operation~\ref{item:pH} can be used to prepare Posners, efficiently, in
   resource states that can power 
   universal measurement-based quantum computation (Sec.~\ref{section:AKLT}).

   \begin{figure}[tb]
   \centering
   \includegraphics[width=.45\textwidth, clip=true]{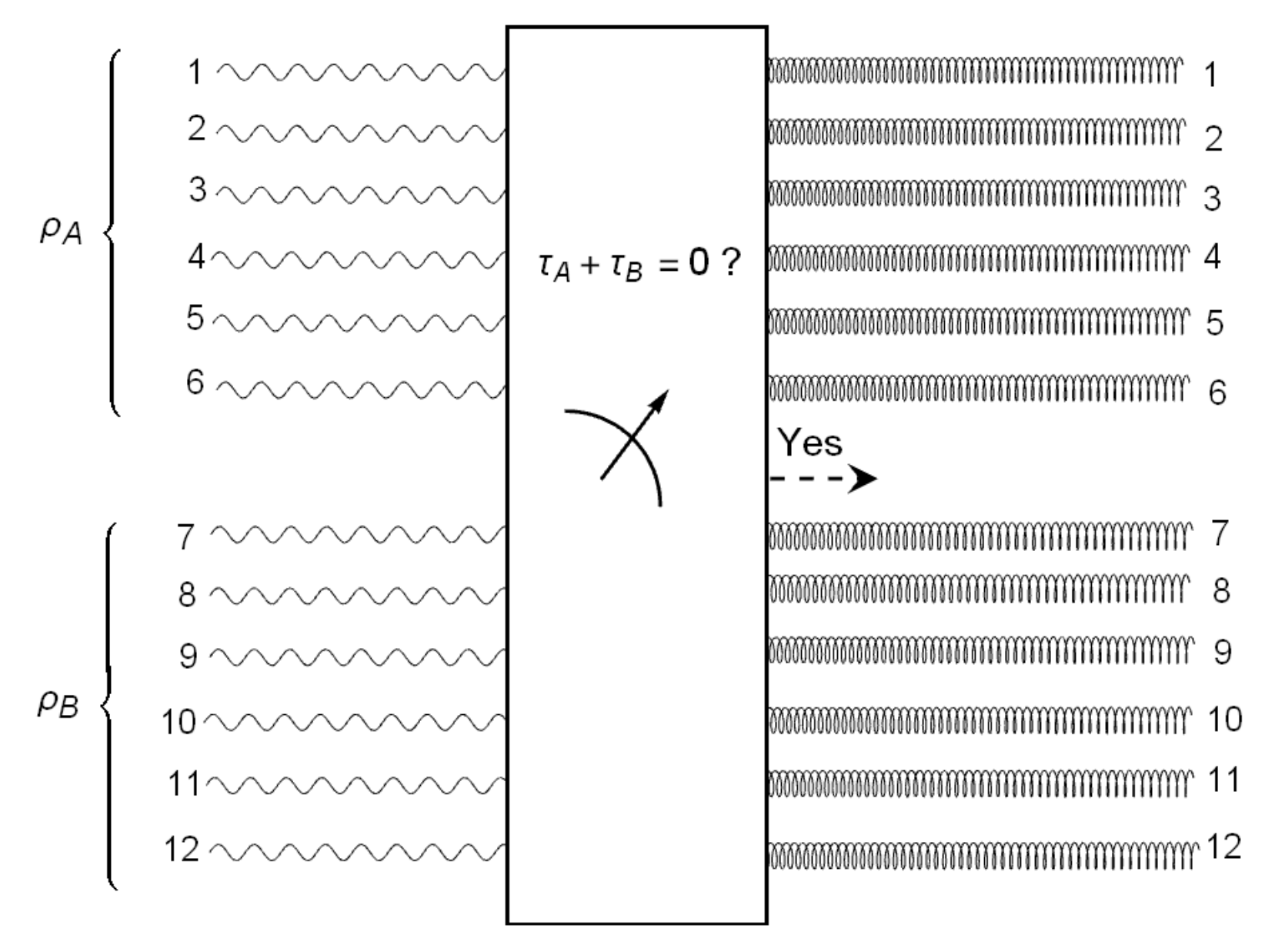}
   \caption{\caphead{Circuit-diagram element that represents
   a Posner-binding measurement (operation~\ref{item:Meas})
   that yields a positive outcome:} 
   Wavy lines represent qubits in hextuples
   that are not in a dodectuple (a pair of hextuples).
   Coils represent qubits in a dodectuple.}
   \label{fig:operation5positive}
   \end{figure}
   
   \begin{figure}[tb]
   \centering
   \includegraphics[width=.45\textwidth, clip=true]{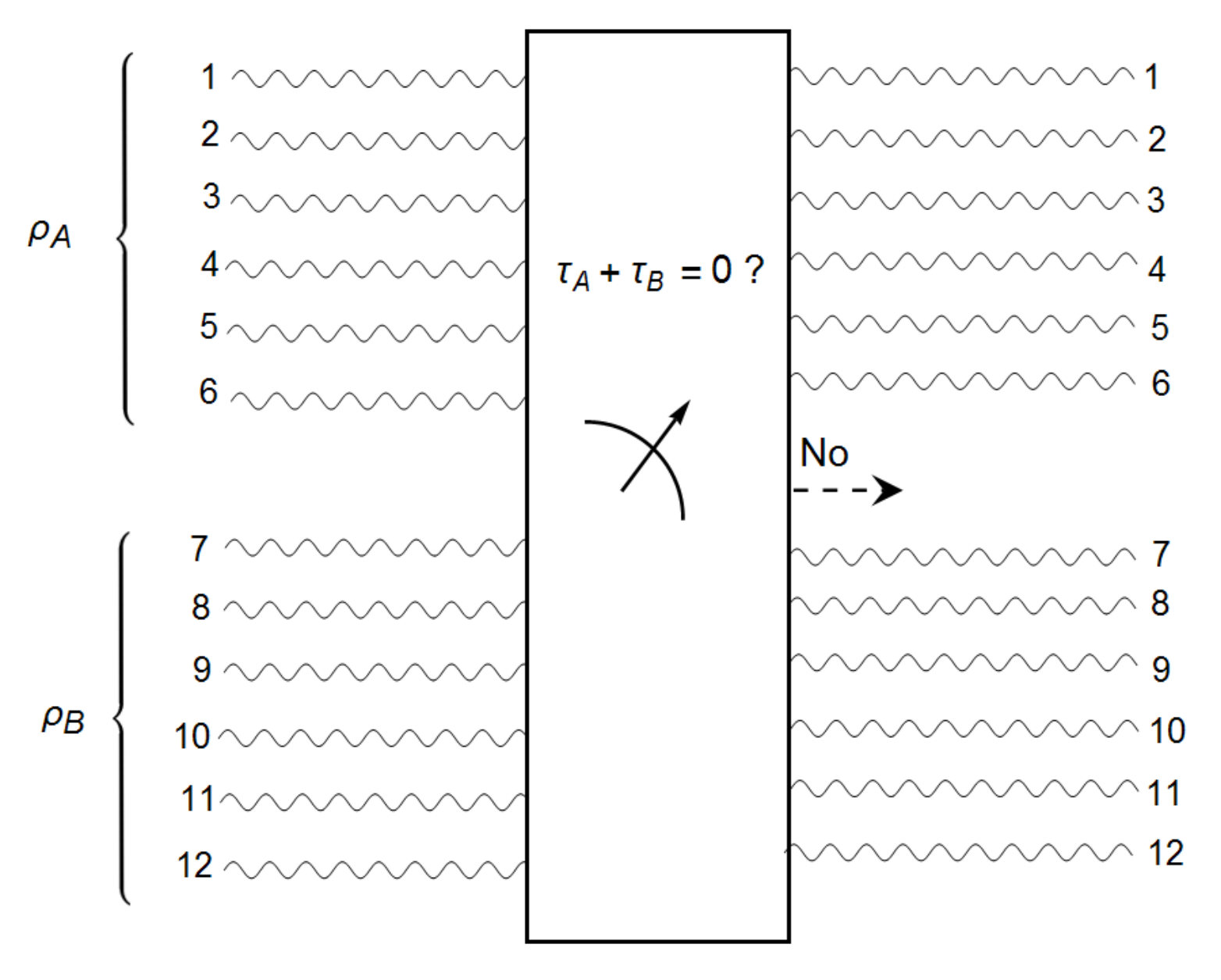}
   \caption{\caphead{Circuit-diagram element that represents
   a Posner-binding measurement (operation~\ref{item:Meas})
   that yields a negative outcome.}}
   \label{fig:operation5negative}
   \end{figure}

   \item  \label{item:Meas}
   \textbf{Posner-binding measurement}
    (Figures~\ref{fig:operation5positive} and~\ref{fig:operation5negative}):
   \emph{Let $A$ and $B$ denote two hextuples
   formed via operation~\ref{item:Form_Pos}.
   Whether the hextuples' $\GC$ eigenvalues 
   sum to zero can be measured nondestructively: 
   $\tau_A  +  \tau_B  =  0$.}
   
   First, we discuss the measurement's physical manifestation.
   Then, we mathematize the operation with a PVM.
   
   The measurement manifests in the binding, or failure to bind,
   of two Posners.
   Fisher conjectures as follows~\cite{Fisher15},
   supported by quantum-chemistry calculations~\cite{Swift_17_Posner}:
   Two Posners, $A$ and $B$, can bind together.
   They bind upon approaching each other
   such that their directed symmetry axes (Fig.~\ref{fig:h_axis})
   lie side-by-side and point oppositely.
   One Posner forms a mirror image of the other.
   
   Such Posners bind, Fisher conjectures~\cite{Fisher15},
   when and only when $\tau_A + \tau_B = 0$.
   F\&R support the conjecture with a Berry-phase argument.
   They formalize the conjecture as 
   a ``quantum dynamical selection rule''~\cite{Fisher_17_Quantum}.
   Hence if $A$ and $B$ approach with the right orientation,
   whether they bind depends entirely on
   whether $\tau_A + \tau_B = 0$.
   The molecules' bound-together-or-not status 
   serves as a classical measurement record.
   So does the environment, as in Posner creation (Sec.~\ref{section:One_Pos_set_up}):
   Posner binding releases about 1 eV of heat~\cite{Fisher15,Swift_17_Posner}.
   
   Let us formalize the measurement, using the mathematics of QI.
   We define a projector on $\left( \HilPos \right)^{ \otimes 2}$:
   \begin{align}
      & \PiStick
      \label{eq:ProjTau2} :=
      \left(  \Pi_{ \tau_A  =  0 }  \otimes  \Pi_{ \tau_B  =  0 }  \right)
       +  \left( \Pi_{ \tau_A  =  \pm 1 }  \otimes 
                   \Pi_{ \tau_B  =  \mp 1 }  \right) \, .
   \end{align}
   The PVM
   \begin{align}
      \label{eq:C_PVM}
      \Set{ \PiStick ,  \id  -  \PiStick  }
   \end{align}
   can be measured.
   Suppose that the first outcome obtains (that the Posners bind).
   The two-Posner state $\rho$ updates as 
   \begin{align} 
      \rho \mapsto   \frac{  \PiStick  \,  \rho  \,  \PiStick  }{
      \Tr \left(  \PiStick  \,  \rho  \,  \PiStick \right) }  \, .
   \end{align}
   The twelve qubits form a \emph{dodectuple}.
   Suppose, instead, that the second outcome obtains 
   (that the Posners fail to bind).
   The joint state updates as
   \begin{align}
      \rho  \mapsto  \frac{  
      \rho  -  \{  \PiStick ,  \rho  \}  +  \PiStick  \,  \rho  \,  \PiStick }{
      1  -  \Tr  \left(  \PiStick  \,  \rho  \right)  }  \, .
   \end{align}
   The anticommutator of operators $O$ and $O'$ is denoted by
   $\{ O  ,  O'  \}$.

   \item  \label{eq:Bound_Ops}
   \textbf{Dodectuple operations:}
   \emph{Suppose that hextuples $A$ and $B$ have been measured
   with the PVM~\eqref{eq:C_PVM}.
   Suppose that outcome $\PiStick$ has obtained.
   The twelve logical qubits can undergo operation~\ref{item:Twelve_rot},
   followed by~\ref{item:Separate} or~\ref{item:Disperse_bind}.}
   \begin{enumerate}[leftmargin=18pt]

   
      %
      
       \begin{figure}[tb]
       \centering
       \includegraphics[width=.18\textwidth, clip=true]{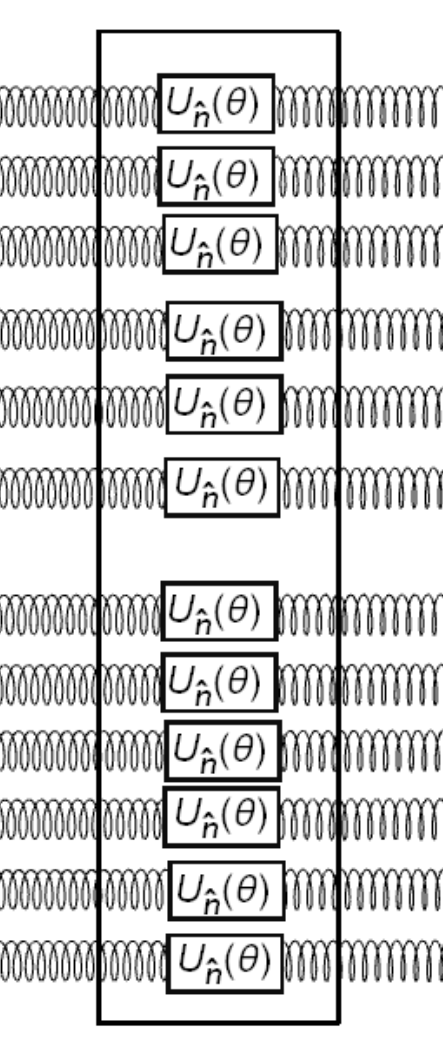}
       \caption{
       \caphead{Circuit-diagram element that represents
       dodectuple-coordinated single-qubit unitaries
       (operation~\ref{item:Twelve_rot}).}
       $\hat{n}$ denotes the axis rotated about.
       $\theta$ denotes the angle rotated through.}
       \label{fig:operation6b}
       \end{figure}

      \item  \label{item:Twelve_rot}
      \textbf{Dodectuple-coordinated single-qubit unitaries}
      (Fig.~\ref{fig:operation6b}):
      \emph{The two hextuple's qubits can undergo 
      approximately identical single-qubit rotations: 
      $[ U_{ \hat{n} }(\theta) ]^{ \otimes 12 }$.}
      The qubits rotate as neuron firings generate magnetic fields.
      See the comments about operation~\ref{item:Six_rot}.
      
       \begin{figure}[tb]
       \centering
       \includegraphics[width=.4\textwidth, clip=true]{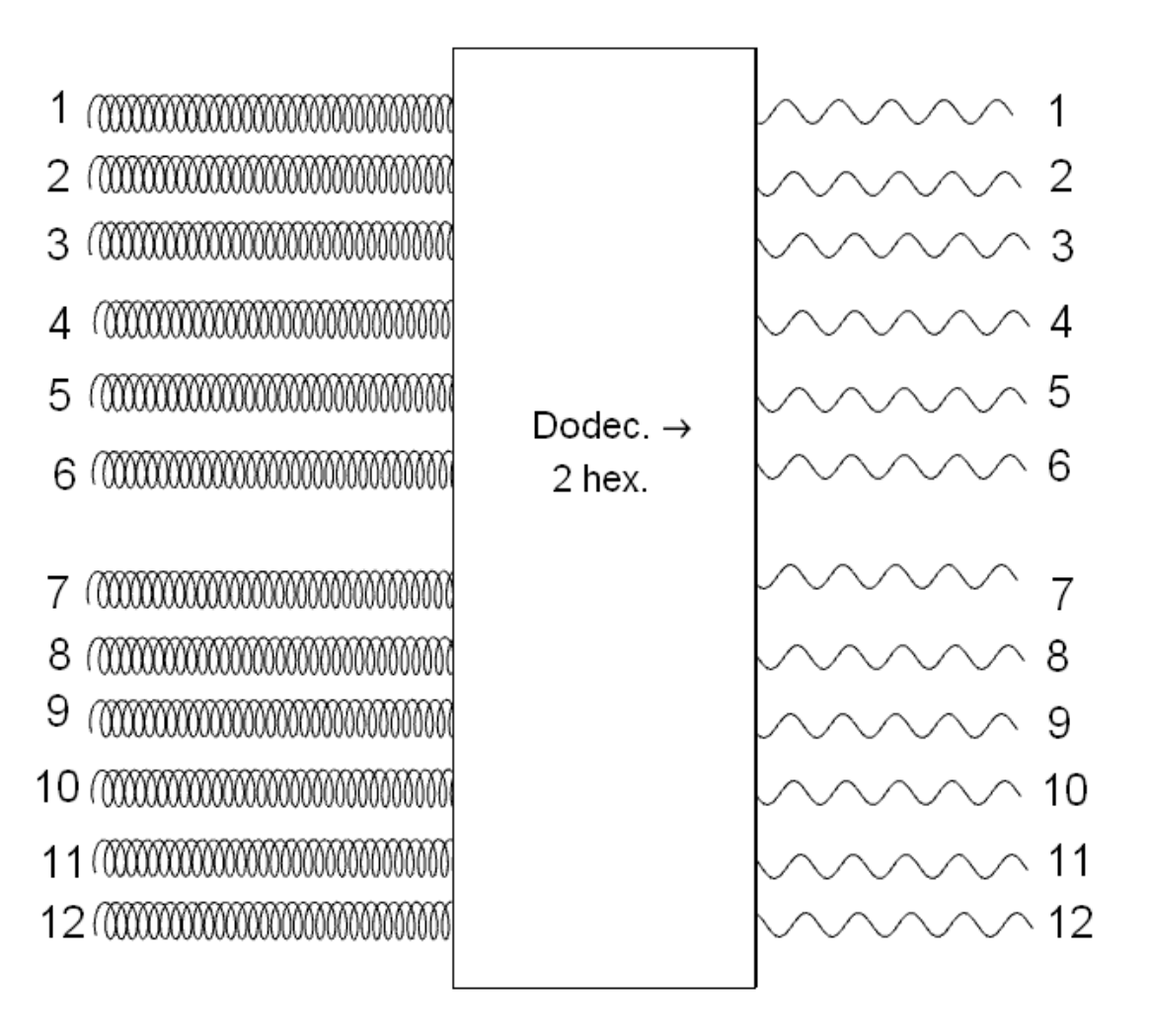}
       \caption{\caphead{Circuit-diagram element that represents
       the separation of a dodectuple into two hextuples
       (operation~\ref{item:Separate}):}
       Each coil represents a qubit in a dodectuple. 
       Each wavy line represents a qubit in a hextuple
       that is not in a dodectuple.}
       \label{fig:operation6c}
       \end{figure}

      \item  \label{item:Separate}
      \textbf{Dodectuple $\to$ 2 hextuples}
      (Fig.~\ref{fig:operation6c}):
      \emph{A dodectuple can separate into 
      independent hextuples,
      the hextuples that joined together.}
      The Posners can drift apart.
      They can return to undergoing
      operations~\ref{item:Six_rot} and ~\ref{item:pH}.
      
       \begin{figure}[tb]
       \centering
       \includegraphics[width=.4\textwidth, clip=true]{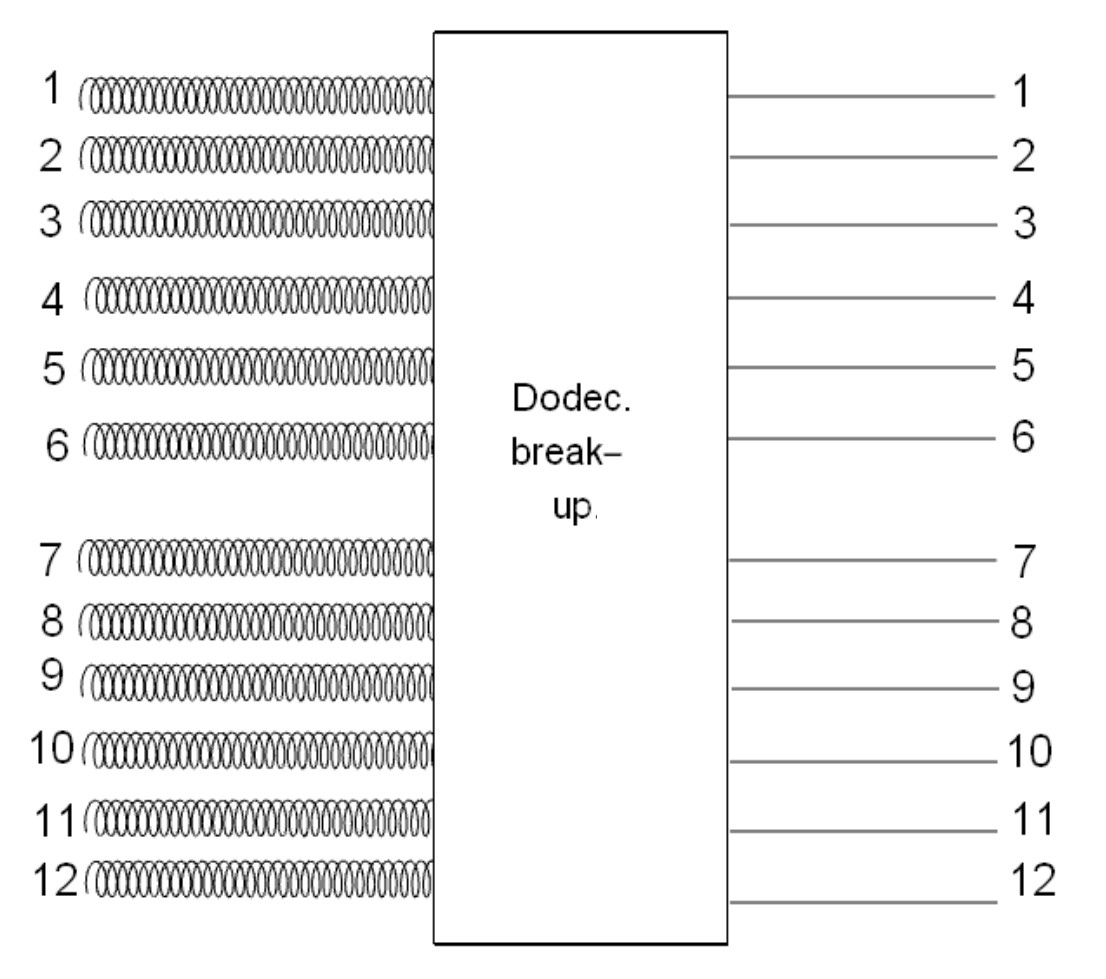}
       \caption{\caphead{Circuit-diagram element that represents
       dodectuple break-up
       (operation~\ref{item:Disperse_bind}):}
       Each coil represents a qubit in a dodectuple.       
       Each straight line represents a qubit not grouped with
       any other qubits.}
       \label{fig:operation6d}
       \end{figure}

      \item  \label{item:Disperse_bind}
      \textbf{Dodectuple break-up}
     (Fig.~\ref{fig:operation6d}):
      \emph{The hextuples can break down into their constituents:
      The qubits can cease to correspond to 
      meaningful geometries or to observables $\GC$.}
      The qubits thereafter behave independently.
      They can, again, undergo 
      operation~\ref{item:Form_Pos}.
      
      The hextuples break down as the Posners hydrolyze.
      Fisher conjectures that bound-together Posners 
      hydrolyze more often
      than separated Posners~\cite{Fisher15},
      as reviewed in this paper's Sec.~\ref{section:Backgrnd_MF}.
   
   \end{enumerate}
   
\end{enumerate}

\subsubsection{Analysis of Posner quantum computation}
\label{section:Analyze_abstract}

We have dissected Fisher's narrative into physical processes,
then abstracted out the computations that the processes effect.
Fisher's narrative~\cite{Fisher15}
can now be cast as a quantum circuit,
depicted in Fig.~\ref{fig:Circuit_MF_story}.
As shown there, the qubits spatially near a qubit $A$
can change from time step to time step.
The qubits approach and separate as
the phosphorus nuclei traverse the fluid.
Spatially local interactions (such as molecular binding) 
implement the circuit elements.

Four features of Posner operations merit analysis.
Two operations entangle logical qubits.
The entanglement generated is discussed in 
Sec.~\ref{section:Ent_capacity}.
Section~\ref{section:Control} concerns control:
To perform the QI-processing tasks introduced 
in Sections~\ref{section:Stick_apps}-\ref{section:AKLT}, 
one might need fine control over Posners.
Biofluids might not exert such control.
But assuming control
facilitates first-step QI analyses.
One operation merits its own section:
The measurement~\eqref{eq:C_PVM} is compared
with a Bell measurement,
and applied in QI-processing tasks, in Sec.~\ref{section:Stick_apps}.

\begin{figure*}[p] 
\centering
\includegraphics[width=0.76\textwidth, clip=true]{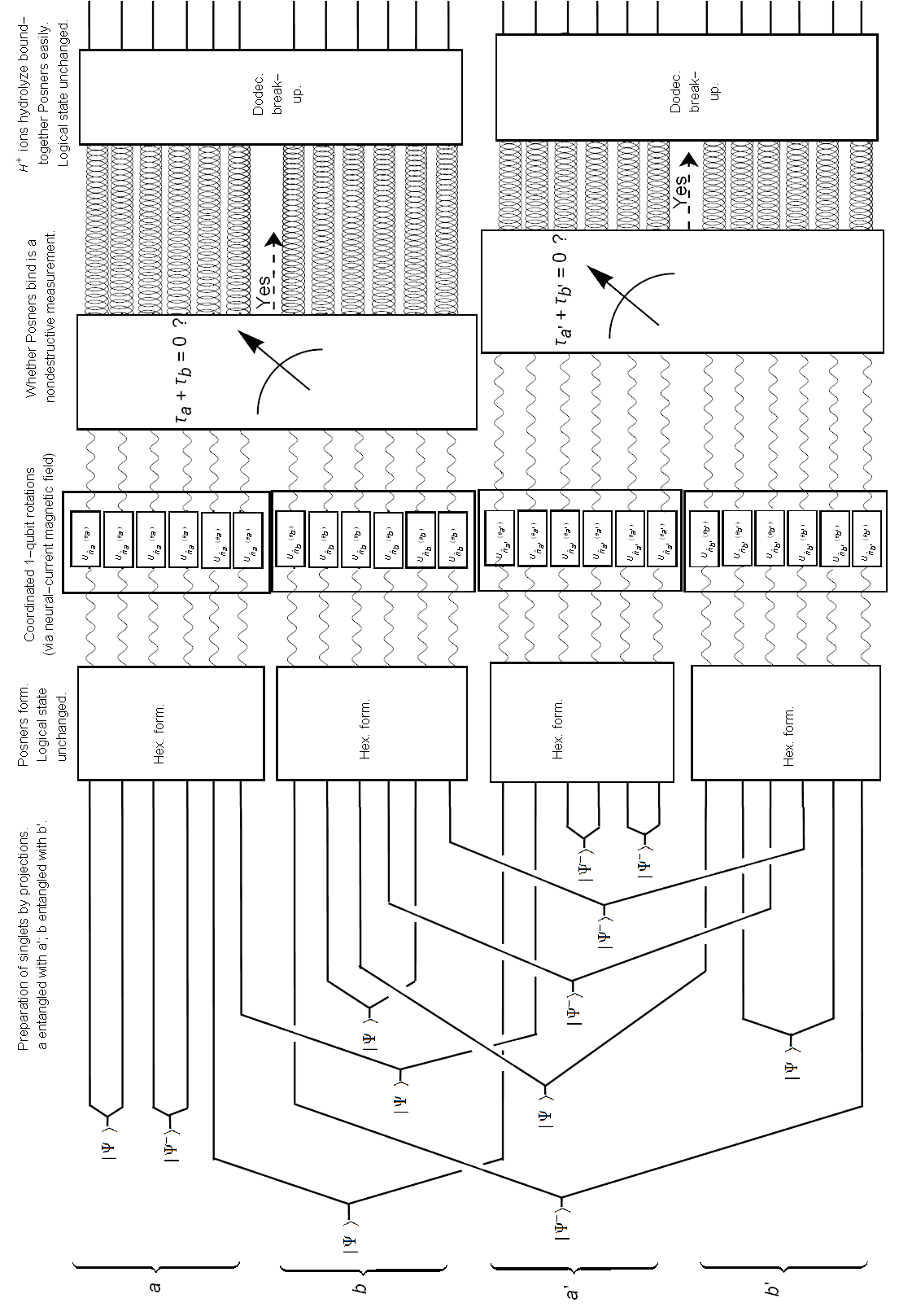}
\caption{\caphead{Circuit representation of Fisher's quantum-cognition narrative:} 
Fisher conjectures that certain chemical processes occur,
in a certain sequence, in the body~\cite{Fisher15}.
The sequence is reviewed in this paper's introduction.
We abstracted out the computations effected by the chemical processes,
in Sec.~\ref{section:Abstract2}.
The abstraction enables us to recast Fisher's narrative 
as a quantum circuit.
Time progresses from left to right in the figure 
(from the bottom to the top of the page).
The sets of six qubits are labeled $a$, $b$, $a'$, and $b'$, 
as in~\cite[p. 5, Fig. 3]{Fisher15}.
The circuit elements are defined in Sec.~\ref{section:Abstract2}.
}
\label{fig:Circuit_MF_story}
\end{figure*}
\paragraph{Entanglement generation:}
\label{section:Ent_capacity}
Entanglement enables quantum computers
to solve certain problems quickly.\footnote{
More precisely, contextuality does~\cite{Spekkens_08_Negativity,Howard_14_Contextuality}.}
Two Posner operations create entanglement:
Bell-pair creation (operation~\ref{item:Bell}) and
the Posner-binding measurement
(operation~\ref{item:Meas}).

Bell pairs serve as units of entanglement
in QI~\cite{NielsenC10}.
We present two implications of Bell-pair creation
for Posners.
First, Bell-pair creation (operation~\ref{item:Bell}),
with the Posners' geometry,
can efficiently prepare a state
that fuels universal measurement-based quantum computation
(Sec.~\ref{section:AKLT}).
Second, distributing Bell pairs across Posners
can affect their binding probabilities (Sec.~\ref{section:Bio_Bell}).

The role played by Bell pairs in QI processing is well-known.
Less obvious is how much,
and which kinds of, 
entanglement $\PiStick$ creates and destroys.
We characterize this entanglement in two ways
(Sec.~\ref{section:Analyze_Stick}).
The PVM~\eqref{eq:C_PVM}, we show,
transforms a subspace as 
a coarse-grained Bell measurement.
Bell measurements facilitate
quantum teleportation~\cite{Bennett_93_Teleporting}.
The PVM~\eqref{eq:C_PVM} facilitates \emph{incoherent teleportation}:
A state's weights are teleported;
the coherences are not.

One might expect Posner binding to render
Posner quantum computation universal:
Conventional wisdom says,
nearly any entangling gate, plus all single-qubit unitaries,
form a universal gate set~\cite{Deutsch_95_Universality,Lloyd_95_Almost,Brylinski_01_Universal,Freedman_02_Diameters,Harrow_08_Exact}.
Posner operations include entangle qubits
(via operation~\ref{item:Meas})
and rotate qubits
(via operations~\ref{item:Six_rot} and~\ref{item:Twelve_rot}).
(Arbitrary rotations through angles of up to $\pi$ might be realized.
Typical angles are expected to be smaller.
See App.~\ref{section:One_Qubit_Gates}.)

But the conventional wisdom appears inapplicable to Posner operations,
for three reasons.
First, conventional-wisdom gates
evolve the system unitarily.
The Posner-binding measurement~\eqref{eq:C_PVM} does not.
(Hence our shift to measurement-based quantum computation
in Sec.~\ref{section:AKLT}.)
Second, many universality proofs decompose
a desired entangling gate 
into implementable gates.
The Posner-binding measurement seems unlikely to decompose.

Third, conventional-wisdom entangling gates
are defined in terms of qubits' states.
The Posner-binding measurement 
is defined in terms of $\tau$.
$\tau$ is an eigenvalue of an observable $\GC$
of a hextuple of qubits.
One must deduce how the measurement transforms 
any given qubit. 
Does this indirect entangler of qubit states,
with single-qubit rotations 
(operations~\ref{item:Six_rot} and~\ref{item:Twelve_rot}),
form a universal set?
The answer merits further study.

Finally, one might wonder how the permutations~\ref{item:Swap}
alter entanglement.
Each permutation decomposes into two-party swaps, as
$(231) = (321)(132)$.
Swaps shift entanglement amongst subsystems,
rather than creating entanglement.

%
%
%




%
%
%
\paragraph{Control required to perform 
quantum-information-processing tasks with Posner molecules:}
\label{section:Control}
In Sections~\ref{section:Stick_apps}--\ref{section:Bio_Bell},
we concatenate Posner operations
to form QI-processing protocols.
Implementing the protocols may require 
fine control over the chemical processes
that effect the computations.
The body might seem unlikely to realize fine control.
We illustrate with two examples.
Then, we justify the assumption of fine control.

Consider, as a first example, 
running an arbitrary quantum circuit.
Arbitrary qubits must rotate through arbitrary angles $\theta$,
about arbitrary axes $\hat{n}$, arbitrarily precisely.
In the quantum-cognition setting, 
logical qubits rotate as Posners experience
magnetic fields generated by neural currents
(via operations~\ref{item:Six_rot} and~\ref{item:Twelve_rot}).
The field experienced depends on the Posner's location,
which depends on the Posner's collisions with other particles.
Fluid particles collide randomly.
Random collisions appear unlikely to facilitate
the precise rotations required for a given circuit.

The $\tau_A + \tau_B = 0$ measurement (operation~\ref{item:Meas})
provides a second example.
Consider Posners $A$ and $B$ 
that approach each other.
One might wish to infer, from the Posners' binding or lack thereof,
whether $\tau_A + \tau_B = 0$.
But the inference is justified only if
$A$ and $B$ were oriented such that
whether they would bind 
depended only on whether $\tau_A + \tau_B$ vanished.

Suppose that one Posner's $\hat{z}_\In$-axis 
stood tip-to-tail with
the other Posner's $\hat{z}_\In$-axis,
rather than side-by-side.
The Posners would fail to bind.
But one could not infer that $\tau_A + \tau_B  \neq  0$.
Only finely tuned two-Posner encounters
reflect whether $\tau_A + \tau_B = 0$.
Only finely tuned encounters constitute measurements.\footnote{
\label{footnote:Postselect}
If two Posners bind, then $\tau_A + \tau_B = 0$;
the inference is justified.
But binding does not constitute merely a measurement.
Binding constitutes a measurement followed by classical postprocessing.
By \emph{classical postprocessing}, we mean the following. 
Consider performing some protocol in each of several trials.
Let the protocol involve a measurement.
Consider the data collected throughout trials.
Consider discarding some of the data, keeping only the data
collected during the trials in which
the measurement yielded some outcome $x$.
One has classically postselected on $x$.
If two Posners bind, then 
(i) whether $\tau_A + \tau_B = 0$ is measured and
(ii) the ``yes'' outcome is classically postselected on.
If two Posners bind, step (i) alone is not implemented;
a measurement alone is not performed.

Classical postprocessing differs from 
the postselection in, e.g.,~\cite{Aaronson_05_Quantum}.
The latter postselection affords computational power
unlikely to grace quantum systems.
In contrast, classical postprocessing happens in today's laboratories.
}

But assuming perfect control can facilitate QI-theoretic analyses.
Many QI protocols are phrased in the language of ``agents.''
One imagines intelligent agents, Alice and Bob, 
who wish to process QI.
One specifies and analyzes protocols in terms of the agents' intents and actions.
Alice and Bob are often assumed to perform 
certain operations with perfect control.
Examples of such ``allowed operations'' include 
local operations and classical communications~\cite{Horodecki_09_Quantum}.

Control partitions 
(i) what can be achieved in principle from
(ii) what can be achieved easily 
with today's knowledge and techniques.
Item (ii) shifts with our understanding and technology.
Item (i) is permanent and is the focus of much QI theory.

A few decades ago, for example, experimentalists had trouble 
performing CNOT gates. 
Many groups have mastered the gate by now.
These groups implement protocols devised before
CNOTs appeared practical.

Similarly, precise phosphate rotations appear impractical.
But some precise-rotation mechanism could be discovered.
Also, by assuming perfect control, we derive a limit on 
what Posners can achieve without perfect control.
We ascertain what QI processing is possible in principle.

\subsection{The Posner-binding measurement and
applications thereof to quantum information processing}
\label{section:Stick_apps}

The measurement~\eqref{eq:C_PVM} entangles two Posners' states.
Yet the measurement projectors,
$\Pi_{AB}$ and $\id - \Pi_{AB}$,
entangle states in different ways.
How much either projector entangles is not obvious.
Neither is the PVM's potential for processing QI. 

This section sheds light on these unknowns.
We compare the PVM to a \emph{Bell measurement},
a standard QI operation (Sec.~\ref{section:Analyze_Stick}).
The next two sections detail applications of the PVM:
The PVM facilitates incoherent teleportation
(Sec.~\ref{section:Teleportation}).
Also, the PVM can be used to
project Posners
onto their $\tau = 0$ eigenspaces
(Sec.~\ref{section:Tau_zero_3}).

\subsubsection{Comparison of the Posner-binding measurement
with a Bell measurement}
\label{section:Analyze_Stick}

First, we review Bell states and measurements~\cite{NielsenC10}.
A Bell measurement prepares an entangled state
of two qubits.
Four maximally entangled states
span the two-qubit Hilbert space, $\mathbb{C}^4$.
The orthonormal \emph{Bell basis} is 
\begin{align}
   \label{eq:Bell_basis}
   \{  & \ket{ \Phi^+ }  :=  \frac{1}{ \sqrt{2} }  
   (  \ket{ 0 0 }  +  \ket{ 1 1 }  )  \\
   & \ket{ \Phi^- }  :=  \frac{1}{ \sqrt{2} }  
   (  \ket{ 0 0 }  -  \ket{ 1 1 }  )  \\
   & \ket{ \Psi^+ }  :=  \frac{1}{ \sqrt{2} }
   (  \ket{ 0 1 }  +  \ket{ 1 0 } )  \\
   & \ket{ \Psi^- }  :=  \frac{1}{ \sqrt{2} }
   (  \ket{ 0 1 }  -  \ket{ 1 0 } )  \}  \, .
\end{align}

A \emph{Bell measurement} is represented by the PVM 
\begin{align}
   \label{eq:Bell_meas}
   \{  \ketbra{ \Phi^+ }{ \Phi^+ }  ,  \ketbra{ \Phi^- }{ \Phi^- }  ,
   \ketbra{ \Psi^+ }{ \Psi^+ }  ,  \ketbra{ \Psi^- }{ \Psi^- }  \}  \, .
\end{align}
Many QI protocols involve Bell measurements.
Examples include quantum teleportation~\cite{Bennett_93_Teleporting}, 
superdense coding (the effective transmission of two bits 
via the direct transmission of just one bit, with help from entanglement)~\cite{Bennett_92_Communication},
and teleportation-based quantum computation~\cite{Gottesman_99_Demonstrating,Knill_01_Scheme,Nielsen_03_Quantum,Zhou_00_Methodology}.

Posner binding simulates a coarse-grained Bell measurement.
The Bell-state projectors~\eqref{eq:Bell_basis}
are defined on $\mathbb{C}^2$.
In contrast, the Posner Hilbert space $\HilPos$
is isomorphic to $\mathbb{C}^6$.
We therefore define an effective qubit subspace.
Let $\ket{ \onet }$ denote an arbitrary 
$\tau = 1$ eigenstate of $\CThree$;
and $\ket{ \twot }$, an arbitrary
$\tau = 2$ eigenstate.
$\ket{ \onet }$ and $\ket{ \twot }$ serve analogously
to $\ket{ 0 }$ and $\ket{ 1 }$ in
${\rm span} \Set{ \ket{ \onet } ,  \ket{ \twot } }$.

%
\begin{proposition}
\label{prop:Bell_meas}
Let $A$ and $B$ denote two Posners.
The measurement~\eqref{eq:C_PVM}
transforms the effective two-qubit space
\begin{align}
   {\rm span} \Set{ \ket{ \onet , \onet }  ,  \ket{ \onet , \twot },
\ket{ \twot , \onet },  \ket{ \twot , \twot } }
\end{align}
identically to the coarse-grained Bell measurement
\begin{align}
   \label{eq:Coarse_Bell}
   \Set{ \ketbra{ \Phi^+ }{ \Phi^+ }  +  \ketbra{ \Phi^- }{ \Phi^- } ,
            \ketbra{ \Psi^+ }{ \Psi^+ }  +  \ketbra{ \Psi^- }{ \Psi^- } }  \, .
\end{align}
\end{proposition}

\begin{proof}
The projector~\eqref{eq:ProjTau2} transforms
the two-qubit space as
\begin{align}
   \PiStick  & =  \ketbra{ \onet , \twot }{ \onet , \twot }
   +  \ketbra{ \twot , \onet }{ \twot , \onet }  \, .
\end{align}
Let us relabel $\onet$ as 0 and $\twot$ as 1.
The projector becomes
\begin{align}
   \label{eq:Bell_sim1}
   \PiStick    & =  
   \ketbra{ \Psi^+ }{ \Psi^+ }  +  \ketbra{ \Psi^- }{ \Psi^- }  \, .
\end{align}
Direct substitution into the RHS yields the LHS.

Consider the complementary projector in the measurement~\eqref{eq:C_PVM}.
$\id - \PiStick$ transforms the effective two-qubit space as
\begin{align}
   \id  -  \PiStick  =  \ketbra{ \onet , \onet }{ \onet , \onet }
   +  \ketbra{ \twot , \twot }{ \twot , \twot }  \, .
\end{align}
Relabeling and direct substitution show that
\begin{align}
    \id  -  \PiStick  =
    \ketbra{ \Phi^+ }{ \Phi^+ }  +  \ketbra{ \Phi^- }{ \Phi^- }  \, .
\end{align}

\end{proof}

Let us quantify the coarse-graining in Proposition~\ref{prop:Bell_meas}.
Let $\ket{ \chi }$ denote
an arbitrary two-qubit state.
Consider measuring $\ket{ \chi }$ in the Bell basis.
One of four possible outcomes obtains.
The outcome can be encoded in
$\log_2 (4)  =  2$ bits.
You could encode, in one bit,
whether a $\Phi$ outcome or a $\Psi$ outcome obtained.
You could encode, in the second bit,
whether a $+$ outcome or a $-$ outcome obtained.

Imagine knowing the first bit's value and forgetting the second bit's.
The state most reasonably attributable to the system
would be
$( \ketbra{ \Phi^+ }{ \Phi^+ }  +  \ketbra{ \Phi^- }{ \Phi^- } )  
\ket{ \chi }$ or
$( \ketbra{ \Psi^+ }{ \Psi^+ }  +  \ketbra{ \Psi^- }{ \Psi^- } )
\ket{ \chi}$,
depending on the first bit.
This state would be the state most reasonably attributable to the system
if, instead,~\eqref{eq:Coarse_Bell} were measured 
and the outcome were known.

The information in the measurement outcome
can be quantified differently.
Appendix~\ref{section:Quant_Bind_Outcome} contains details.

\subsubsection{Application 1 of binding Posner molecules:
Incoherent teleportation}
\label{section:Teleportation}

Quantum teleportation transmits a state $\ket{ \psi }$
from one system to another~\cite{Bennett_93_Teleporting}.
Consider agents Alice and Bob who 
live in the same town.
Suppose that Bob moves to another country.

Let Alice hold a qubit $A$
that occupies a state $\ket{ \psi }  =  c_0 \ket{ 0 }  +  c_1 \ket{1}$.
Alice may wish to send Bob $\ket{ \psi }$.
Mailing $A$ would damage the state.
Alice should not measure $A$,
call Bob on the telephone, and tell him the outcome.
Bob would receive too little information
to reconstruct $\ket{ \psi }$ in his lab.

Suppose that, before Bob moved away,
he and Alice created a Bell state, e.g., $\ket{ \Psi^-}$.
Let $B$ and $C$ denote the entangled qubits.
Suppose that Bob takes $C$ across the world.
Alice should perform 
a Bell measurement~\eqref{eq:Bell_meas} of $AB$.
One of four possible outcomes will obtain.
Alice should tell Bob which, via telephone.
Her call communicates $\log_2(4) = 2$ bits.
Bob should transform $C$ with a unitary
whose form depends on the news.
$C$ will come to occupy the state $\ket{ \psi }$.
$A$ will occupy a different state.
Alice will have teleported $\ket{ \psi }$ to Bob.

We introduce a variation on quantum teleportation,
\emph{incoherent teleportation}.
The protocol illustrates the power of Posner binding.
The protocol relies on
entanglement, classical information, and Posner binding.

Posner binding resembles a coarse-grained Bell measurement,
as shown in Sec.~\ref{section:Analyze_Stick}.
Hence Posner binding fails to teleport
all the information teleportable with 
a Bell measurement.
The coherences in $\ket{ \psi }$ are not sent.
A classical random variable,
which results from decohering $\ket{ \psi }$, is.

The set-up and notation are introduced in Sec.~\ref{section:Tele_setup}.
The protocol is introduced in Sec.~\ref{section:Tele_protocol}
and analyzed in Sec.~\ref{section:Teleport_analysis}.

\paragraph{Set-up and notation:}
\label{section:Tele_setup}

Let $\ket{ \jt }$ denote an arbitrary $\tau = j$ eigenstate of $\CThree$,
for $j = 0 , 1 , 2$.
The $\ket{ \jt }$'s form the computational basis
for the qutrit space
${\rm span}\{ \ket{ \zerot } , \ket{ \onet },  \ket{ \twot }  \}$.
This basis serves, in incoherent teleportation, 
similarly to the $\sigma_z$ eigenbasis
in conventional teleportation.

Consider restricting the projector~\eqref{eq:ProjTau2}
to the space of two qutrits:
\begin{align}
   \label{eq:Proj_Qutrit}
   \PiStick'  & :=  \ketbra{ \zerot ,  \zerot }{ \zerot , \zerot }
   +  \ketbra{ \onet , \twot }{ \onet , \twot }
   \nonumber \\ & \qquad
   +  \ketbra{ \twot , \onet }{ \twot , \onet }  \, .
\end{align}  
Let $\ket{ \plust }  :=  \frac{1}{ \sqrt{3} }  
(  \ket{ \zerot }  +  \ket{ \onet }  +  \ket{ \twot }  )$.

\paragraph{Incoherent-teleportation protocol:}
\label{section:Tele_protocol}

Let $A$, $B$, and $C$ denote three Posners.
Suppose that $B$ and $C$ begin in $\ket{ \plust , \plust  }$,
then bind together.\footnote{
One might worry that the spin state would decohere
before the Posners bound.
But chemical binding consists of electronic dynamics.
$\ket{ \plust , \plust  }$ is a state of nuclear spins.
Nuclear dynamics tend to unfold
much more slowly than electronic dynamics.
The Born-Oppenheimer approximation reflects
this separation of time scales.
Hence the nuclear state appears unlikely to decohere
before the Posners bind.}
The joint state becomes
\begin{align}
   \label{eq:Qutrit_Bell}
   \PiStickBC  \ket{ \plust , \plust }
   =  \frac{1}{ \sqrt{3} }  \: (  \ket{ \zerot , \zerot }  +  \ket{ \onet , \twot }
   +  \ket{ \twot ,  \onet } )  \, .
\end{align}
In the first term's absence,~\eqref{eq:Qutrit_Bell}
would be a triplet.
A triplet is a Bell pair, a maximally entangled state
that can fuel quantum teleportation.
\eqref{eq:Qutrit_Bell}, we will show, fuels incoherent teleportation.

Suppose that, after~\eqref{eq:Qutrit_Bell} is prepared,
Posners $B$ and $C$ drift apart.
(In quantum-computation language, Alice and Bob share a Bell pair.)
Let $B$ approach $A$.
Let $A$ occupy an arbitrary state
\begin{align}
   \label{eq:ToTeleport}
   \ket{ \psi }  =  c_0 \ket{ \zerot }  +  c_1  \ket{ \onet }  +  c_2  \ket{ \twot }  \, .
\end{align}
The complex coefficients satisfy the normalization condition
$\sum_{j = 0}^2  | c_j |^2  =  1$.
(In quantum-computation language, $\ket{ \psi }$ is the unknown state
that contains information that Alice will teleport to Bob.)
The three Posners occupy the joint state
\begin{align}
   \label{eq:Tel_3_state}
   \ket{ \chi }  :=  
   \ket{ \psi }  \left(  \PiStickBC    \ket{ \plust , \plust }  \right)  \, .
\end{align} 

Suppose that Posners $A$ and $B$ bind together.
(During the analogous quantum-teleportation step,
Alice performs a Bell measurement of her qubits.)
The three-Posner state becomes
\begin{align}
   \label{eq:Tel_3_state2}
   &  \PiStick   \ket{ \chi }
   /  \bra{ \chi }  \PiStick   \ket{ \chi }
   \\ & 
   \label{eq:Tel_3_state2a}
   =  c_0  \ket{ \zerot , \zerot , \zerot }
   +  c_1  \ket{ \onet ,  \twot ,  \onet }
   +  c_2  \ket{ \twot , \onet , \twot }  \\
   & =:  \ket{ \chi' }  \, .
\end{align}

Posner $C$ occupies (Bob holds) the reduced state
\begin{align}
   \rho_C  & :=
   \Tr_{AB} ( \ketbra{ \chi' }{ \chi' } )  \\
   \label{eq:C_state}
   & =  | c_0 |^2  \ketbra{ \zerot }{ \zerot }
   +  | c_1 |^2  \ketbra{ \onet }{ \onet }
   +  | c_2 |^2  \ketbra{ \twot }{ \twot }  \, .
\end{align}
Posner $C$'s state encodes information about $\ket{ \psi }$,
the square moduli of the coefficients in Eq.~\eqref{eq:ToTeleport}.
Yet $C$ has never interacted with $A$ directly.
Information has teleported from $A$ to $C$,
with help from $\ket{ \plust , \plust }$
and from Posner binding.

Posners $A$ and $B$ had a probability
\begin{align}
   p_{\Pi}  =  \Tr  \LParen  \PiStick  \,
   \Tr_C  (  \ketbra{ \chi }{ \chi }  )  \RParen
\end{align}
of binding together.
(An analogous probability can be introduced
into quantum teleportation:
Let Alice have a nonzero probability of
failing to perform her Bell measurement.)

Suppose, instead, that $A$ and $B$ fail to bind together.
The projector
\begin{align}
   \label{eq:OneMinusPi}
   & \id  -  \PiStick
   =  \left(  \Pi_{ \tau_A = 0 }  \otimes  \Pi_{ \tau_B  =  1 }  \right)
   +  \left(  \Pi_{ \tau_A = 0 }  \otimes  \Pi_{ \tau_B  =  2 }  \right)
   \nonumber \\ & \qquad \qquad
   +  \left(  \Pi_{ \tau_A = 1 }  \otimes  \Pi_{ \tau_B  =  0 }  \right)
   +  \left(  \Pi_{ \tau_A = 1 }  \otimes  \Pi_{ \tau_B  =  1 }  \right)
   \nonumber \\ & \qquad \qquad
   +  \left(  \Pi_{ \tau_A = 2 }  \otimes  \Pi_{ \tau_B  =  0 }  \right)
   +  \left(  \Pi_{ \tau_A = 2 }  \otimes  \Pi_{ \tau_B  =  2 }  \right)
\end{align}
projects the state of $AB$.
The three-Posner state $\ket{ \chi }$
[Eq.~\eqref{eq:Tel_3_state}] updates to
\begin{align}
   \left[  \left( \id  -  \PiStick  \right)  \otimes  \id  \right]  \ket{ \chi }
   & =  \frac{1}{2}  [
   c_0  (  \ket{ \zerot , \onet , \twot }  
                  +  \ket{ \zerot , \twot , \onet } )
   \nonumber \\ & 
   +  c_1  (  \ket{ \onet , \zerot , \zerot }
                  +  \ket{ \onet , \onet , \twot } )
   \nonumber \\ & 
   +  c_2  ( \ket{ \twot , \zerot , \zerot }
                  +  \ket{ \twot , \twot ,  \onet } )  ]
   \nonumber  \\
   & =:  \ket{ \chi'' }  \, .
\end{align}
Posner $C$ occupies (Bob holds) the reduced state
\begin{align}
   \label{eq:C_state2}
   & \Tr_{AB}  ( \ketbra{ \chi'' }{ \chi'' } )
   =  \frac{1}{2} [  
   ( | c_1 |^2  +  | c_2 |^2 )  \ketbra{ \zerot }{ \zerot }
   \\ \nonumber & \qquad  \quad
   +  ( | c_2 |^2  +  | c_0 |^2 )  \ketbra{ \onet }{ \onet }
   +  ( | c_0 |^2  +  | c_1 |^2 )  \ketbra{ \twot }{ \twot }  ]  \, .
\end{align}
Again, $C$ contains information about $\ket{ \psi }$,
despite never having interacted directly with $A$.

Suppose that Bob measures $\GC$,
the observable that generates the unitary $\CThree$.
Bob samples from a random variable
whose values 0, 1, and 2 are distributed
according to $( p'_0  =  | c_1 |^2  +  | c_2 |^2 ,  
p'_1  = | c_2 |^2  +  | c_0 |^2 ,  
p'_2  = | c_0 |^2  +  | c_1 |^2 )$.

%
%
%
\paragraph{Analysis of incoherent teleportation:}
\label{section:Teleport_analysis}

Five points merit analysis.
First, we quantify the classical information teleported.
Second, we characterize the QI not teleported.
Third, we compare the resources 
required for incoherent teleportation
to the resources required for quantum teleportation.
Incoherent teleportation, we show fourth, 
implements superdense coding---the
effective sending of much classical information
via the direct sending of little classical information,
with help from entanglement.
Fifth, we explain how to prepare $\ket{ \plust , \plust }$ and $\ket{ \psi }$
with Posner operations.

\paragraph{Quantification of the information teleported:}
Posners $A$ and $B$ teleport a trit to $C$.
A \emph{trit} is classical random variable
that can assume one of three possible values.
Imagine preparing a Posner in the state $\ket{ \psi }$ [Eq.~\eqref{eq:ToTeleport}]
and measuring $\GC$.
The measurement has a probability $p_0  =  | c_0 |^2$
of yielding 0, 
a probability $p_1  =  | c_1 |^2$ of yielding 1,  and
a probability $p_2  =  | c_2 |^2$ of yielding 2.
So does a $\GC$ measurement of $C$, 
if $A$ binds to $B$ [Eq.~\eqref{eq:C_state}].
The distribution has been teleported from $A$ to $C$.

Suppose that $A$ fails to bind to $B$.
Measuring Posner $C$ has a probability 
$p'_0  =  \frac{1}{2}  (  | c_1 |^2  +  | c_2 |^2  )$ of yielding 0,
a probability $p'_1  =  \frac{1}{2}  (  | c_2 |^2  +  | c_0 |^2 )$ 
of yielding 1, and
a probability $p'_2  =  \frac{1}{2}  (  | c_0 |^2  +  | c_1 |^2 )$ 
of yielding 2  [Eq.~\eqref{eq:C_state2}]. 
The measurement of Posner $C$ is equivalent to
an encoded generalized measurement of $\ket{ \psi }$.

A \emph{positive-operator-valued measure} (POVM)
$\Set{ M_1 , M_2 ,  \ldots  ,  M_\ell }$
represents a generalized quantum measurement~\cite{NielsenC10}.
The measurement elements are positive operators
$M_k > 0$.
They satisfy the completeness condition
$\sum_k  M_k^\dag  M_k  =  \id$.
The $M_k$'s need not be projectors,
unlike PVM elements.

Consider the POVM 
\begin{align}
   & \Bigg\{  \ketbra{ \overline{ \zerot } }{ \overline{ \zerot } }
   =  \frac{1}{ \sqrt{2} }  ( \ketbra{ \onet }{ \onet }  +  \ketbra{ \twot }{ \twot } ),
   \\ & \; \;
   \ketbra{ \overline{ \onet } }{ \overline{ \onet }  }
   =  \frac{1}{ \sqrt{2} }  ( \ketbra{ \twot }{ \twot }  +  \ketbra{ \zerot }{ \zerot } ) ,
   \\ & \; \;
   \ketbra{ \overline{ \twot } }{ \overline{ \twot }  }
   =  \frac{1}{ \sqrt{2} }  ( \ketbra{ \zerot }{ \zerot }  +  \ketbra{ \onet }{ \onet } )   
   \Bigg\}  \, .
\end{align}
Measuring this POVM is equivalent to measuring
the encoded observable
$\overline{ \mathcal{G} }_{ \CThree }  :=
\sum_j  j_\tau  \ketbra{ \overline{ j_\tau } }{ \overline{ j_\tau } }$.
Measuring the $\overline{ \mathcal{G} }_{ \CThree }$ of $\ket{ \psi }$
has a probability $p'_j$ of yielding
the encoded outcome $\overline{\jt}$.

Suppose that Posners $A$ and $B$ fail to bind.
A measurement of 
the $\GC$ of $C$ simulates
an encoded measurement of 
the $\GC$ of $\ket{ \psi }$.

\paragraph{Classicality of the teleported information:}
Only the square moduli $| c_j |^2$ are teleported.
The coefficients' phases are not.
Hence incoherent teleportation achieves less than 
quantum teleportation does.

Section~\ref{section:Analyze_Stick} clarifies why: 
Quantum teleportation involves Bell measurements.
Incoherent teleportation involves measurements of
whether $\tau_A + \tau_B = 0$.
The $\tau_A + \tau_B = 0$ measurement simulates
a coarse-grained Bell measurement.

\paragraph{Comparison of resources required for incoherent teleportation
with resources required for quantum teleportation:}
In quantum teleportation,
qubit $C$ undergoes a local unitary
conditioned on the Bell measurement's outcome.
Our Posner $C$ needs no such conditional correcting.

Yet part of our story depends on 
the Posner-binding measurement's outcome:
the interpretation of the outcome of a $\GC$ measurement of Posner $C$.
Suppose that $A$ binds to $B$.
A $\GC$ measurement of Posner $C$ simulates
a measurement of the $\GC$ of $\ket{ \psi }$.
Suppose, instead, that $A$ fails to bind to $B$.
A $\GC$ measurement of Posner $C$ simulates
a measurement of the 
$\overline{ \mathcal{G} }_{ \CThree }$ of $\ket{ \psi }$.

\paragraph{Incoherent teleportation as superdense coding:}
Incoherent teleportation offers less power,
we have seen, than quantum teleportation.
Yet incoherent teleportation offers more power
than classical communication.
Suppose that Alice has incoherently teleported $\ket{ \psi }$.
Bob may wish to know which probability distribution he holds,
$\Set{ p_0 , p_1 , p_2 }$ or $\Set{ p'_0 , p'_1 , p'_2 }$.
Alice should send Bob a bit directly:
a zero if $A$ bound to $B$
and a one otherwise.\footnote{
How could such classical communication 
manifest in biological systems?
In ordinary QI protocols,
classical communication manifests as telephone calls.
Today's phones do not fit in human brains.
But one can envision classical channels in a biofluid.
For example, if $A$ and $B$ bind, 
they shove water molecules away together.
If $A$ and $B$ fail to bind, 
water propagates away from them differently.
The patterns in the fluid's motion
may be distinguished.
The fluid-motion pattern would encode the bit.}

Alice would directly send Bob a bit,
while effectively sending a trit,
with help from entanglement and Posner binding.
A trit is equivalent to $\log_2(3) > 1$ bits.
Hence Alice packs much classical information (a trit)
into a small classical system (a bit).

Much classical information packs into a small classical system,
with help from a Bell pair and a Bell measurement,
in \emph{superdense coding}~\cite{Bennett_92_Communication}.
Conventional superdense coding packs
two bits into one.
Our protocol packs information less densely.

\paragraph{Preparing $\ket{ \psi }$ and $\ket{ \plust , \plust }$:}
Incoherent teleportation involves two coherent quantum states,
$\ket{ \psi }$ and $\ket{ \plust , \plust }$.
Instances of these states can be prepared with Posner operations.
We illustrate with an example in App.~\ref{section:Tele_state_app}.
To construct each state,
one arranges singlets in each Posner.
One then rotates one spin per Posner
through an angle $\frac{\pi}{4}$ about the $y_\lab$-axis.\footnote{
A single-qubit rotation can occur as follows:
Suppose that qubits $A$ and $B$ form a singlet
(having occupied the same diphosphate).
$A$ enters a Posner $P$, while $B$ enters a Posner $P'$.
The qubits in $P$ undergo a rotation;
the qubits in $P'$ do not.
$P$ then hydrolyzes.
Meanwhile, the pyrophosphatase enzyme
that split $A$ from $B$ 
has been splitting other diphosphates.
$A$ forms a Posner $P''$ with newly split diphosphates.
}
Alternative preparation protocols might exist.


%
%
%
\subsubsection{Application 2 of binding Posner molecules:
Projecting Posner molecules onto 
their $\tau = 0$ subspaces}
\label{section:Tau_zero_3}

The AKLT state can be prepared via
projections onto subspaces associated with
the spin quantum number $s = \frac{3}{2}$.
Posners can come to occupy a variation AKLT$'$ on
the AKLT state.
The Posners must be projected
onto their $\tau = 0$ subspaces
(Sec.~\ref{section:AKLT}).
Posner-binding measurements can effect these projections.

\begin{proposition}
\label{proposition:Tau_zero_3}
Let $A, B, C, \ldots, M$ label $m$ Posner molecules.
The following sequence of events projects each Posner's state
onto the $\tau = 0$ eigenspace:
\begin{enumerate}
   \item    \label{item:AB0}
   $A$ and $B$ bind together, then drift apart.
   
   \item     \label{item:BC0}
   $B$ and $C$ bind together, then drift apart.
   
   \item     \label{item:CA0}
   $C$ and $A$ bind together, then drift apart.
   $A$, $B$, and $C$ have been projected onto their $\tau = 0$ subspaces.
   
   \item   \label{item:Rest}
   Each remaining Posner ($D, \ldots, M$) binds to 
   a projected Posner, then drifts away.

\end{enumerate}
\end{proposition}

\begin{proof}
First, we prove that steps~\ref{item:AB0}-\ref{item:CA0}
project $A$, $B$, and $C$ onto
their $\tau = 0$ subspaces.
Then, we address step~\ref{item:Rest}.

A projector of the form~\eqref{eq:ProjTau2}
represents each binding.
A product $\Pi_{123}$ of projectors represents
the sequence~\ref{item:AB0}-\ref{item:CA0} of bindings:
\begin{align}
   \label{eq:Proj_prod1}
   \Pi_{123}  & =
   \Big[  \left(  \Pi_{ \tau_A = 0 }  \otimes  \Pi_{ \tau_B  =  0  }  
             \otimes  \id^{ \otimes (m - 2) }  \right)
             \\ \nonumber & \qquad
            +  \left(  \Pi_{ \tau_A = \pm 1 }  \otimes  \Pi_{ \tau_B  =  \mp 1  }  
                \otimes  \id^{ \otimes (m - 2) }  \right)  \Big]
   \\ & \nonumber \quad \times
   \Big[  
   \left( \id  \otimes  \Pi_{ \tau_B  =  0 }  \otimes  \Pi_{ \tau_C  =  0 }
           \otimes  \id^{ \otimes ( m - 3 ) }  \right)
   \\ & \nonumber \qquad  \quad
   +   \left(  \id  \otimes  \Pi_{ \tau_B  =  \pm 1 }  \otimes  \Pi_{ \tau_C  =  \mp 1 }
                 \otimes  \id^{ \otimes ( m - 3 ) }  \right)
   \Big]
   \nonumber \\ \nonumber  & \quad \times
   \Big[  
   \left(  \Pi_{ \tau_A  =  0 }  \otimes  \id  \otimes  \Pi_{ \tau_C  =  0 }
            \otimes  \id^{ \otimes ( m - 3 ) }  \right)
   \\ \nonumber & \qquad  \quad
   +  \left(  \Pi_{ \tau_A  =  \pm 1 }  \otimes  \id  \otimes  \Pi_{ \tau_C  =  \mp 1 }
                \otimes  \id^{ \otimes ( m - 3 ) }  \right)
   \Big]
   \nonumber \\
   \label{eq:Proj_prod2} 
   & =  \Pi_{ \tau_A = 0 }  \otimes   \Pi_{ \tau_B = 0 }  \otimes
   \Pi_{ \tau_C = 0 }  \otimes
   \id^{ \otimes ( m - 3 ) }   \, .
\end{align}
Equation~\eqref{eq:Proj_prod1} can be understood
in terms of a frustrated lattice,
as explained in App.~\ref{section:Frust_int}.

Step~\ref{item:Rest} of Proposition~\ref{proposition:Tau_zero_3} 
is proved as follows.
Suppose that Posners $A$ and $D$ bind together, then drift apart.
The joint state of $AD$ is acted on by
\begin{align}
   \label{eq:Proj_AD}
    \left(  \Pi_{ \tau_A = 0 }  \otimes  \Pi_{ \tau_D  =  0  }  \right)
    +  \left(  \Pi_{ \tau_A = \pm 1 }  \otimes  \Pi_{ \tau_D  =  \mp 1  }  \right)  \, .
\end{align}
The state of $A$ was projected onto the $\tau_A = 0$ subspace
during steps~\ref{item:AB0}-\ref{item:CA0}. 
Hence the final term in Eq.~\eqref{eq:Proj_AD}
annihilates the $AD$ state.
Hence $\Pi_{ \tau_D  =  0  }$ projects the state of $D$.
\end{proof}
\noindent
Proposition~\ref{proposition:Tau_zero_3} 
will provide a subroutine in the following section's protocol.

\subsection{Efficient preparation of Posner molecules in 
universal quantum-computation resource states}
\label{section:AKLT}


\begin{figure}[tb] 
\centering
\includegraphics[width=.35\textwidth, clip=true]{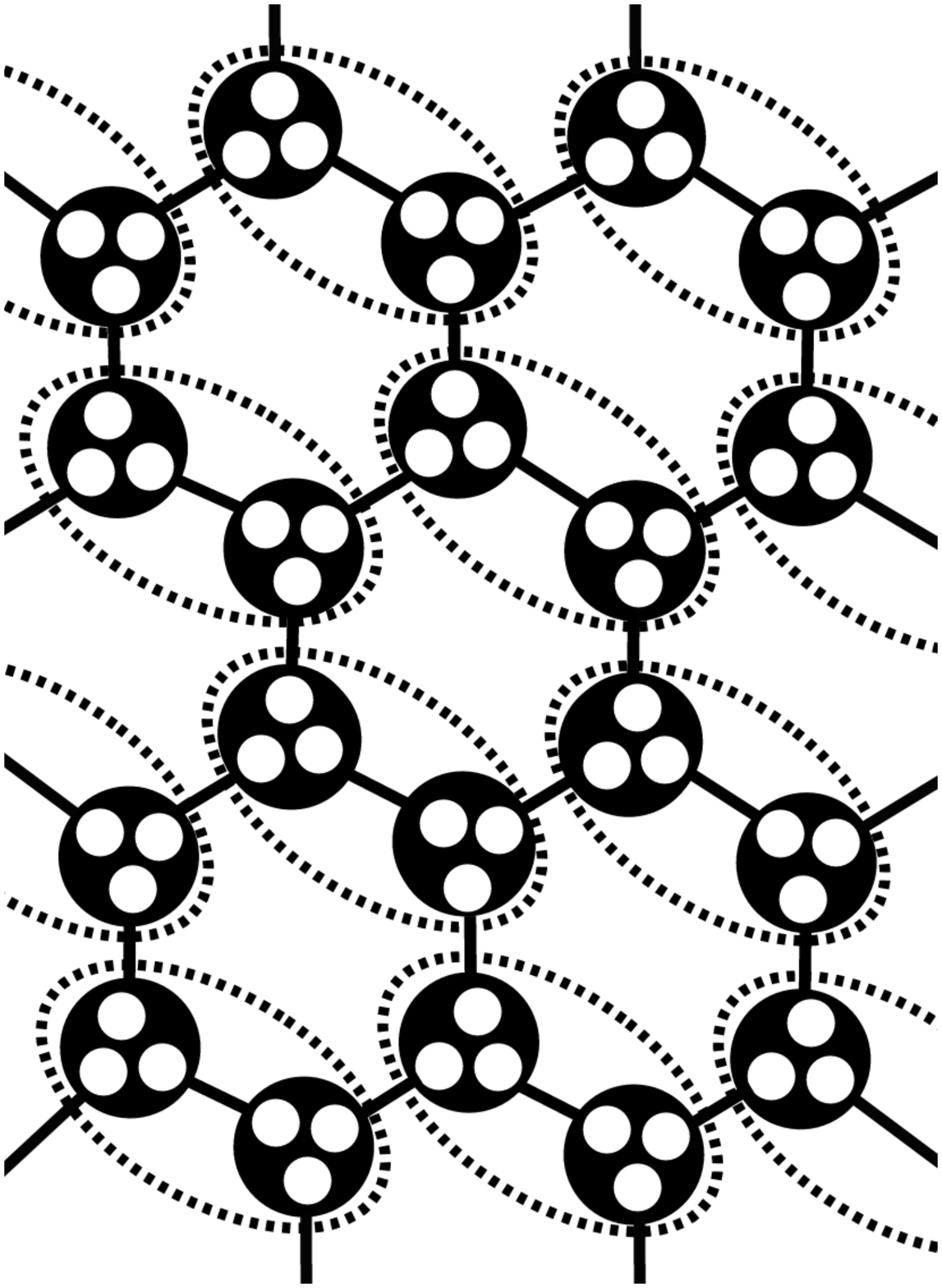}
\caption{\caphead{AKLT$'$ state:}
Spins on a honeycomb lattice 
can occupy the  Affleck-Lieb-Kennedy-Taski (AKLT) state
$\AKLTHon$~\cite{Affleck_87_Rigorous}.
$\AKLTHon$ serves as a universal resource in
measurement-based quantum computation (MBQC)~\cite{Miyake_11_Quantum,Wei_12_Two}.
So does the similar state $\AKLTPrime$,
which Posner operations can prepare efficiently.
Each dashed oval encloses 
the spins in a Posner molecule.
Each molecule consists of two trios of phosphorus nuclear spins.
Each large black dot represents a trio.
Each small white dot represents a spin.
Each thin black line connects the two spins in a singlet.
This figure resembles Fig.~3a of~\cite{Wei_12_Two},
as $\AKLTPrime$ resembles $\AKLTHon$.
This figure does not illustrate
the spatial arrangement of Posners
in $\AKLTPrime$.
Rather, the figure illustrates 
the entanglement in $\AKLTPrime$.
}
\label{fig:Posner_AKLT_Lattice}
\end{figure}

How complicated an entangled state can 
Posner operations (Sec.~\ref{section:Abstract_logic})
prepare efficiently?
Many measures quantify multipartite entanglement.
We study computational resourcefulness.
Posners operations, we show, can efficiently prepare a state 
that fuels universal MBQC:
By operating on the state locally,
one can efficiently simulate
a universal quantum computer.

The Posner state is a variation on 
an \emph{Affleck-Lieb-Kennedy-Tasaki (AKLT) state}.
AKLT first studied a one-dimensional (1D) chain
of spin-1 particles.
They constructed a nearest-neighbor antiferromagnetic Hamiltonian~\cite{Affleck_87_Rigorous,Affleck_88_Valence,Kennedy_88_Existence}.
The ground state, $\AKLTOne$, has a known form.
A constant gap, independent of the system size,
separates the lowest two energies.

$\AKLTOne$ has many applications in quantum computation~\cite{Gross_07_Novel,Gross_07_Measurement,Verstraete_04_Valence,Ostlund_95_Thermo,Fannes_92_Communications,Klumper_91_Equivalence,Verstraete_04_Renormalization,Brennen_08_Measurement,Miyake_10_Quantum,Bartlett_10_Quantum,Miyake_11_Quantum,Wei_12_Two}.
For example, $\AKLTOne$ was the first state recognized
as a matrix product state (MPS)~\cite{Ostlund_95_Thermo,Fannes_92_Communications,Klumper_91_Equivalence}.
MPSs can efficiently be represented approximately
by classical computers. 
Also, using $\AKLTOne$, one can simulate
arbitrary single-qubit rotations.
One performs local operations,
including adaptive single-qubit measurements,\footnote{
Measurements are \emph{adaptive} if 
earlier measurements' outcomes dictate
which measurements are performed later.}
on the state~\cite{Gross_07_Novel,Gross_07_Measurement,Brennen_08_Measurement}.

Two-dimensional (2D) analogs of $\AKLTOne$ have been defined.
Spin-$\frac{3}{2}$ particles
on a honeycomb lattice
can occupy the state $\AKLTHon$~\cite{Affleck_87_Rigorous}.
Local operations on $\AKLTHon$ can efficiently simulate
universal quantum computation~\cite{Miyake_11_Quantum,Wei_12_Two}.
We will draw on the proof by Wei \emph{et al.}~\cite{Wei_12_Two}.
Reference~\cite{Miyake_11_Quantum} contains results related to
the results in~\cite{Wei_12_Two}.

Wei \emph{et al.} prove the universality of $\AKLTHon$ as follows.
Local POVMs, they show,
reduce $\AKLTHon$ to 
an encoded 2D graph state $\ket{ \overline{ \Graph ( \mathcal{A} ) } }$.\footnote{
A \emph{graph state} is defined in terms of a graph $\Graph$.
Each vertex corresponds to a spin.
Consider the Hamiltonian
\begin{align}
   H_\Graph  =  \sum_{i  \in  \Graph}  
   \left(  \sigma_i^x  
   \bigotimes_{k  \in  \text{NB}(i) }  \sigma_k^z  \right)  \, .
\end{align}
$i$ indexes the vertices in $\Graph$.
The nearest neighbors of $i$ are indexed by
$k \in  \text{NB}(i)$.
$H_\Graph$ has a unique ground state,
called a \emph{graph state}~\cite{Briegel_01_Persistent,Hein_06_Entanglement}.}
The graph $\Graph$ is random,
depending on the set $\mathcal{A}$ of measurement outcomes.
Also the encoding depends on $\mathcal{A}$.
The overline in $\ket{ \overline{ \Graph ( \mathcal{A} ) } }$
represents the encoding.
Wei \emph{et al.} prescribe local measurements of a few qubits.
The measurements convert
$\ket{ \overline{ \Graph ( \mathcal{A} ) } }$ into 
a cluster state on a 2D square lattice
(if $\mathcal{A}$ is a typical set).\footnote{
A \emph{cluster state} is a graph state
associated with a regular lattice $\Graph$~\cite{Briegel_01_Persistent,Raussendorf_01_One,Hein_06_Entanglement}.}
Such cluster states serve as resources in universal MBQC~\cite{Briegel_01_Persistent,Raussendorf_03_Measurement,VandenNest_06_Universal,Miyake_11_Quantum}: 
By measuring single qubits adaptively,
one can efficiently simulate a universal quantum computer.

We introduce a variation on $\AKLTHon$.
We call the variation the \emph{AKLT$'$ state}
and denote the state by $\AKLTPrime$.
Figure~\ref{fig:Posner_AKLT_Lattice} illustrates the state.
$\AKLTPrime$ is prepared similarly to $\AKLTHon$,
resembles $\AKLTHon$ locally,
and fuels universal MBQC similarly.

This section is organized as follows.
Section~\ref{section:Intro_AKLT} reviews
the set-up and the state construction of Wei \emph{et al.}
$\AKLTPrime$ is defined in Sec.~\ref{section:Posner_AKLT}.
How to construct $\AKLTPrime$ efficiently
from phosphorus nuclear spins, 
using Posner operations, is detailed.
Section~\ref{section:Bob_AKLT} describes
the reduction of $\AKLTPrime$ to a 2D cluster state,
known to fuel universal MBQC.
The protocol is analyzed in Sec.~\ref{section:Analyze_AKLT_prep}.

$\AKLTPrime$ holds interest not only as a computational resource,
but also in its own right.
The state is analyzed in Sec.~\ref{section:AKLT_analysis}.
For instance, AKLT$'$ is shown to be a PEPS.


%
%
%
\subsubsection{Set-up by Wei \emph{et al.}}
\label{section:Intro_AKLT}

Wei \emph{et al.} consider a 2D honeycomb lattice,
illustrated in Fig.~3a of~\cite{Wei_12_Two}.
(Figure~\ref{fig:Posner_AKLT_Lattice} has nearly the same form.)
A black dot represents each site.
At each site sit three spin-$\frac{1}{2}$ DOFs.
White dots represent these DOFs, called \emph{virtual spins}.

Let $s_{123}$ and $m_{123}$ denote a site's 
total spin quantum number and
total magnetic spin quantum number.
These numbers can assume the values
$(s_{123},  m_{123} )  =  ( \frac{1}{2} ,  \pm  \frac{1}{2} ),
( \frac{3}{2} ,  \pm  \frac{ 3 }{ 2 } )$,  and
$( \frac{3}{2} ,  \pm  \frac{ 1 }{ 2 } )$.
The qubit trio can behave as one \emph{physical spin}
of $s_{123} = \frac{1}{2}$ or $\frac{3}{2}$.

$\AKLTHon$ may be prepared as follows~\cite{Affleck_87_Rigorous}:
\begin{enumerate}[leftmargin=*]

   \item
   Consider two nearest-neighbor sites.
   Choose a virtual spin in each site.
Form a singlet $\ket{ \Psi^- }$ between these spins.
Perform this process on every pair of nearest-neighbor sites.

   \item  \label{item:Wei_Proj}
   Project each physical spin (each site) 
   onto its $s_{123}  =  \frac{3}{2}$ subspace.
   
\end{enumerate}
\noindent
$\AKLTHon$ is \emph{trivalent}:
Each site links, via singlets, to
three other sites.

\subsubsection{Preparing Posner molecules in $\AKLTPrime$}
\label{section:Posner_AKLT}

\begin{figure*}[tb]
\centering
\includegraphics[width=.99\textwidth, clip=true]{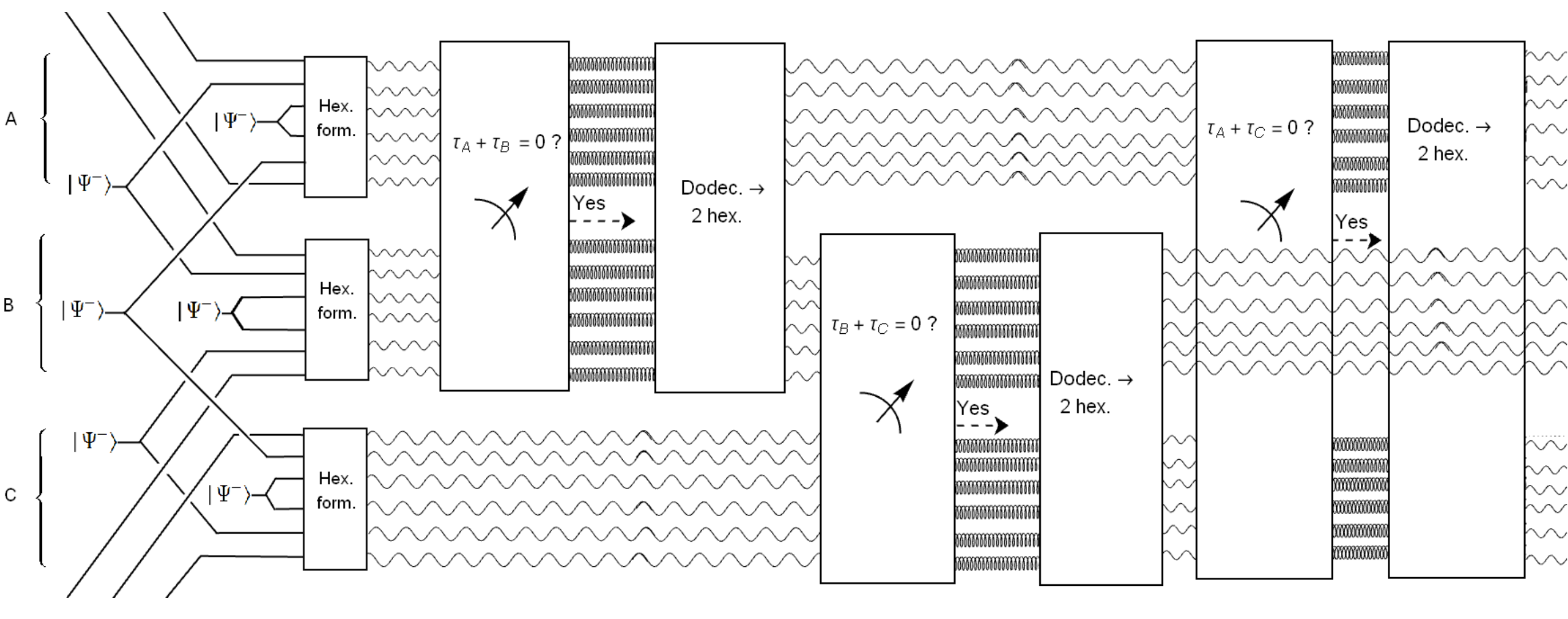}
\caption{\caphead{Part of a circuit that efficiently prepares  
the AKLT$'$ state $\AKLTPrime$:}
The circuit elements represent Posner operations,
as explained in Sec.~\ref{section:Abstract2}.
Some thin black lines extend off the diagram.
These lines represent singlets that terminate in Posners not drawn here.
}
\label{fig:Posner_AKLT_Circuit}
\end{figure*}

Posner operations (Sec.~\ref{section:Abstract_logic})
can nearly prepare $\AKLTHon$.
Whether Posner operations can project trios
onto their $s_{123} = \frac{3}{2}$ subspaces
remains unknown.
But Posner operations can project onto
a molecule's $\tau = 0$ subspace. 

The $\tau = 0$ subspace decomposes into a direct sum
of tensor products 
of two three-qubit subspaces.
The first three-qubit subspace is labeled
by $s_{123}$, the total spin quantum number of 
the qubit triangle at $\zinter = h_+$.
The second three-qubit subspace is labeled
by $s_{456}$.
The $\tau = 0$ subspace has the form
$\left( \frac{3}{2}  \otimes  \frac{3}{2}  \right)  \oplus
\left(  \frac{1}{2}  \otimes  \frac{1}{2}  \right)^{ \oplus 2 }$.
(See Appendix~\ref{section:Group_thry} 
and~\ref{section:C_eigenspaces}
for a derivation.
See~\cite{Shankar94} for background and notation.)
The first term represents
the space $s_{123} \otimes s_{456}  =  \frac{3}{2}  \otimes  \frac{3}{2}$
of two spin-$\frac{3}{2}$ particles.
Wei \emph{et al.}  project onto this space,
in step~\ref{item:Wei_Proj}.

Projecting onto the larger $\tau = 0$ space
yields $\AKLTPrime$.
We now detail how Posners can come to occupy $\AKLTPrime$.
The steps are explained in physical terms
(of molecules, binding, etc.).
Figure~\ref{fig:Posner_AKLT_Circuit} recasts the protocol
in computational terms, as a quantum circuit:
\begin{enumerate}[leftmargin=*]

   \item
   Pyrophosphatase enzymes
   cleave some number $\Sites$ of diphosphates.
   $\Sites$ singlets $\ket{ \Psi^- }$ are prepared
   (via operation~\ref{item:Bell}).
   
   \item
   The phosphates group together in trios.
   Singlets connect the trios as 
   thin black lines connect
   the white dots in Fig.~\ref{fig:Posner_AKLT_Lattice}.
   
   \item  \label{item:Proj_AKLTPrime}
   Each trio, with a nearest-neighbor trio, 
   forms a Posner molecule
   (via operation~\ref{item:Form_Pos}).
    
    \item \label{item:RefreshPos}
    Posners approach each other 
    with the prebinding orientation (see operation~\ref{item:Meas}),
    then drift apart,
    as described in Proposition~\ref{proposition:Tau_zero_3}.
    For simplicity, we focus on the case in which 
    each Posner $P$ approaches only Posners $P'$
    that are nearest neighbors of $P$ in the hexagonal lattice
    (Fig.~\ref{fig:Posner_AKLT_Lattice}).
    But this assumption is unnecessary.
    
    Suppose that each approach leads to binding.
    Proposition~\ref{proposition:Tau_zero_3} is realized.
    Every Posner's state is projected onto the
    $\tau = 0$ subspace.
    
    But two approaching Posners might fail to bind.
    The success probability\footnote{
    \label{footnote:Calc_prob}
    This probability is calculated as follows.
    The $\Sites$ Posners occupy some pure state $\ket{ \psi }$. 
    Consider the two Posners' joint reduced state, $\rho_{P P'}$. 
    The Posners share one singlet.
    The Posner pair contains ten other phosphorus nuclear spins.
    Let $a$ denote an arbitrary one of these other spins.
    $a$ forms a singlet with a spin in some other Posner, $P''$.
    $P''$ is traced out from $\ket{ \psi }$
    in this calculation of $\rho_{P P'}$.
    Hence $\rho_{P P'}$ equals a tensor product of
    ten maximally mixed qubit states 
    $\frac{ \id_2 }{ 2 }$ and $\ket{ \Psi^- }$:
    $\rho_{P P'}  =  \left( \frac{ \id_2 }{ 2 }  \right)^{ \otimes 5 }
    \otimes  \ketbra{ \Psi^- }{ \Psi^- }  \otimes
    \left( \frac{ \id_2 }{ 2 }  \right)^{ \otimes 5 }$.
    Posners $P$ and $P'$ have a probability
    $\approx \Tr \left(  \Pi_{ P P' }  \rho_{P P'}  \right)  
    =   \frac{ 43 }{128 }  \approx 0.336$
    of binding.
    [$\Pi_{ P P' }$ is defined as in Eq.~\eqref{eq:ProjTau2}.]
    This approximation does not account for
    all correlations amongst sites~\cite{Wei_12_Two}
    but is expected to capture the greatest contribution to the probability.
    Exact calculations are left as an opportunity for future study.}
    $\frac{ 43 }{128 }  \approx 0.336$.
    Suppose that Posners $P$ and $P'$ fail to bind.
    Suppose that $P$ has already been projected onto
    its $\tau_P = 0$ subspace.
    $P'$ can be ``refreshed'':
    Let $P'$ drift into a region of
    high pH and/or high $\Mg\Two$ concentration.
    $P'$ likely hydrolyzes (undergoes operation~\ref{item:pH}).
    Two of the phosphorus nuclear spins used to form 
    a singlet internal to the Posner.
    These spins form a singlet no longer,
    due to the binding failure. 
    These two spins can drift away; a fresh singlet can replace them.
    Four other phosphorus nuclear spins remain.
    They continue to form singlets with spins in other Posners.\footnote{
    \label{footnote:Reduced_state}
    This claim can be checked via direct calculation.
    Computational resources limited our calculation to
    the reduced state of 13 spins.
    Whether longer-range correlations affect the results
    is left for future study.
    See footnote~\ref{footnote:Calc_prob}.}
    The two new, and four old, phosphates can form a Posner $\tilde{P}'$,
    via operation~\ref{item:Form_Pos}.
    
    $\tilde{P}'$ occupies the state that $P$ occupied
    before the binding failure.
    $\tilde{P}'$ can approach $P$ 
    with the prebinding orientation.
    If the binding fails, $\tilde{P}'$ can be refreshed again.

\end{enumerate}

\subsubsection{Reduction of $\AKLTPrime$ to a cluster state
known to fuel universal MBQC}
\label{section:Bob_AKLT}

Local operations can reduce $\AKLTHon$
to a cluster state on a 2D square lattice~\cite{Wei_12_Two}.
Such cluster states serve as
universal resources in MBQC~\cite{Briegel_01_Persistent,Raussendorf_03_Measurement,VandenNest_06_Universal}.
Section~\ref{section:Wei_to_cluster} reviews 
the reduction in~\cite{Wei_12_Two}.
Section~\ref{section:Toward_cluster} explains
the need to deviate from this reduction.
Section~\ref{section:Cluster} details the deviation.

\paragraph{Model: Reduction of $\AKLTHon$ to a cluster state:}
\label{section:Wei_to_cluster}
Wei \emph{et al.} prescribe two steps.
First, each site is measured with a POVM.
The measurements yield a state equivalent,
under local unitaries, to a random graph state.
Second, a few qubits are measured with local POVMs.

Let us detail the initial measurements.
Site $v$ is measured with the POVM
\begin{align}
   \label{eq:Fs}
   & \Bigg\{ F_{v, x}  =  
   \sqrt{ \frac{2}{3} }  ( \ketbra{ \pm \pm \pm }{ \pm \pm \pm } )  \, ,  
   \\ \nonumber & \quad
   F_{v, y}  =   \sqrt{ \frac{2}{3} }  
   ( \ketbra{ i, i, i }{ i, i, i }  +  \ketbra{ -i, -i, -i }{ -i, -i, -i } )  \, , 
   \\ \nonumber & \quad
   F_{v, z}  =  \sqrt{ \frac{2}{3} }  
   ( \ketbra{ 0 0 0 }{ 0 0 0 }  +  \ketbra{ 1 1 1 }{ 1 1 1 } )
   \Bigg\}  \, .
\end{align}
The $F_{v, z}$ projects onto the subspace
spanned by the $S^{ \zlab }_{123}$ eigenstates associated with
the magnetic spin quantum numbers $m_{123}  =  \pm  \frac{3}{2}$.
$F_{v, \alpha}$ projects onto
the subspace spanned by 
the analogous $S^{ \alpha_\lab }_{123}$ eigenstates,
for $\alpha = x, y$.

Each subspace has dimensionality two.
Hence the POVM reduces each site's Hilbert space
to a qubit space.
The $\sqrt{ \frac{2}{3} }$ leads to 
the completeness relation 
$\sum_{ \alpha = x , y , z }  F_{v, \alpha}^\dag  F_{v , \alpha}$.

$\mathcal{A}$ denotes
the set of POVM outcomes.
The POVMs yield a state $\ket{ \overline{ \Graph ( \mathcal{A} ) } }$.
$\Graph ( \mathcal{A} )$ denotes a random graph
whose form depends on $\mathcal{A}$.
$\ket{ \Graph ( \mathcal{A} ) }$ denotes 
the graph state associated with $\Graph ( \mathcal{A} )$.
The overline denotes an encoding dependent on $\mathcal{A}$.
The system occupies a state equivalent,
via the encoding,
to $\ket{ \Graph ( \mathcal{A} ) }$.

A random graph $\Graph ( \mathcal{A} )$
defines $\ket{ \Graph ( \mathcal{A} ) }$.
In contrast, a regular graph
defines a cluster state.
The cluster state on a 2D square lattice
serves as a universal resource in MBQC.
This cluster state can be distilled 
from $\ket{ \Graph ( \mathcal{A} ) }$,
if $\mathcal{A}$ is typical.
The distillation consists of 
a few single-qubit Pauli measurements~\cite{Wei_12_Two}.

\paragraph{Toward a reduction of $\AKLTPrime$ to a cluster state:}
\label{section:Toward_cluster}
The Wei \emph{et al.} system differs from
the Posner system in two ways.
First, Posner operations cannot necessarily simulate
(i) the POVM~\eqref{eq:Fs} or 
(ii) the Pauli measurements
in Sec.~\ref{section:Wei_to_cluster}.
Whether Posner operations can
remains an open question.
Second, Wei \emph{et al.} invoke $\AKLTHon$.
Posner operations can prepare $\AKLTPrime$.

Posners therefore require a step absent from~\cite{Wei_12_Two}.
To facilitate the explanation,
we invoke the agent framework of QI
(Sec.~\ref{section:Control}).
Different experimentalists can perform different operations easily.
An agent Alice might run a biochemistry lab.
She might be able to effect Posner operations.
An agent Bob might be able to perform local POVMs
but not to create and arrange singlets.\footnote{
How Bob performs the POVMs 
and falls outside this paper's scope. 
This section's purpose is to demonstrate that
Posner operations can prepare a state 
that can fuel universal MBQC.
We have shown how to prepare the state.
Now, we need prove only that AKLT$'$
can fuel universal MBQC in principle.
Specifying the procedure in terms of POVMs,
classical communication, etc. suffices.
We leave experimental implementations as an opportunity for further study,
which might facilitate the use of Posners
in engineered QI processing.}

Together, Alice and Bob could produce cluster states.
Alice would create $\AKLTPrime$ 
and pass the state to Bob.
Bob would perform local POVMs.
(He might ask Alice to refresh a few Posners.)
Together, the agents would form cluster states
that fuel universal MBQC.

\paragraph{Reduction of $\AKLTPrime$ to a cluster state:}
\label{section:Cluster}
Bob will perform the protocol in Sec.~\ref{section:Wei_to_cluster}.
But first, he measures each Posner's
$\mathbf{S}_{123}^2  \otimes  \mathbf{S}_{456}^2$.
Suppose that Posner $P$ yields the outcome
labeled by $\frac{3}{2}$
[yields the outcome $\hbar^2  \left(  \frac{3}{2}  \right)^2  \times 2
=  \frac{9}{2}  \,  \hbar^2$].
The measurement has succeeded.

Now, suppose that Posner $P$ yields the outcome
labeled by $\frac{1}{2}$.
The measurement has failed.
Bob returns $P$ to Alice.
Alice hydrolyzes $P$ via operation~\ref{item:pH}.
She refreshes the internal singlet,
as in step~\ref{item:RefreshPos} in Sec.~\ref{section:Posner_AKLT}.
Let $\tilde{P}$ denote the refreshed Posner.\footnote{
$P$ contained four spins apart from the internal singlet.
Each of these spins remains in a singlet
with a spin in another Posner.
See the calculational comments in footnote~\ref{footnote:Reduced_state}.}
Bob measures the 
$\mathbf{S}_{123}^2  \otimes  \mathbf{S}_{456}^2$
of $\tilde{P}$.
He and Alice ``repeat until success''
(until obtaining the $\frac{3}{2}$ outcome).

Bob holds an $\AKLTHon$ state.
He now follows the prescription
of Wei \emph{et al.} (Sec.~\ref{section:Wei_to_cluster}).

\subsubsection{Analysis of $\AKLTPrime$ preparation}
\label{section:Analyze_AKLT_prep}

Posner operations, we have shown, 
can prepare $\AKLTPrime$ efficiently.
A few local measurements reduce
$\AKLTPrime$ to a cluster state on a 2D square lattice.
This cluster state can be used directly 
in universal MBQC.
Posners' universality is remarkable:
Most quantum states cannot power universal MBQC~\cite{Gross_09_Most}.
The Posners' singlets and their geometry
(the decomposition of Posners into triangles,
and the triangles' trivalence),
underlie the state's universality.\footnote{
$\AKLTPrime$ is not the simplest universal Posner resource state.
Singlets on a trivalent lattice would suffice.
The $\tau = 0$ projections are unnecessary,
by the second equality in Eq.~(31) of~\cite{Wei_12_Two}.
Yet $\AKLTPrime$ merits defining, for three reasons.
First, preparing $\AKLTPrime$ is natural:
In $\AKLTHon$, pairs of sites are projected onto
their $s_{123} \otimes s_{456} 
= \frac{3}{2}  \otimes  \frac{3}{2}$ subspaces.
In $\AKLTPrime$, pairs of sites are projected onto
slightly larger subspaces.
Hence $\AKLTPrime$ resembles $\AKLTHon$ locally.
Second, suppose that Alice did not project the Posners onto
their $\tau = 0$ subspaces.
Bob would obtain more ``error'' outcomes, labeled by $\frac{1}{2}$'s,
in Sec.~\ref{section:Cluster}.
Third, $\AKLTPrime$ holds interest in its own right 
(Sec.~\ref{section:AKLT_analysis}).}

Opportunities for enhancing and simplifying our protocol exist:
\begin{enumerate}[leftmargin=*]

   \item 
   A precise structure---a honeycomb lattice---underlies $\AKLTPrime$.
In contrast, biomolecules drift randomly.
Biological singlets might not form 
honeycombs of their own accord.
Random graphs might likely arise.\footnote{
\label{footnote:Hydroxy_hex}
Onuma and Ito offer hope that
Posners form regular lattices, however.
Hydroxyapatite is the mineral 
$\Ca_5 ( {\rm PO}_4 )_3  ( {\rm OH} )$.
A variation on hydroxyapatite forms much of bone.
Onuma and Ito proposed that hydroxyapatite crystals grow as
Posners form hexagonal lattices~\cite{Onuma98Posner}.}
Such graphs underlie states that might power MBQC.

Such graphs might have two or three dimensions.
Three-dimensional (3D) graph states
offer particular promise.
First, they have substantial connectivity,
needed for universality~\cite{Wei_12_Two}.
Second, 3D cluster states fuel fault-tolerant universal MBQC.
The scheme relies on 
toplogical quantum error correction~\cite{Raussendorf_06_Fault}.

\begin{figure}[tb]
\centering
\includegraphics[width=.37\textwidth, clip=true]{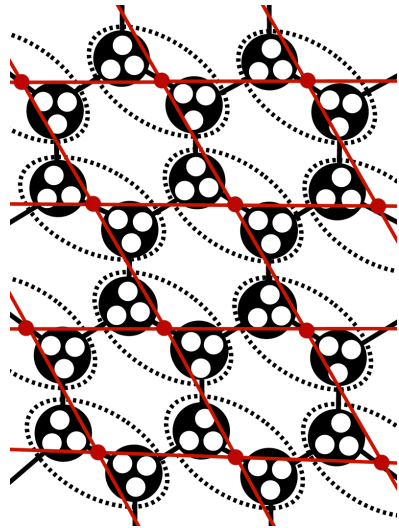}
\caption{\caphead{Coarse-graining the hexagonal lattice 
into a square lattice:}
The black lines form a hexagonal lattice
(see Fig.~\ref{fig:Posner_AKLT_Lattice}).
Three qubits (small white dots) occupy each site (large black dot).
Two neighboring sites form a Posner molecule
(encircled with a dashed hoop).
The lattice can be coarse-grained:
Each Posner's two sites can be
lumped together (into a red dot).
The coarse-grained lattice is square 
(as shown by the long, red lines).
This coarse-graining might facilitate a universality protocol
simpler than the one in Sec.~\ref{section:Cluster}.}
\label{fig:Hex_square}
\end{figure}
   \item
   Bob might avoid returning Posners to Alice.
The Posners' triangles (Fig.~\ref{fig:Geo}) form 
the sites in a hexagonal lattice (Fig.~\ref{fig:Posner_AKLT_Lattice}).
Consider coarse-graining two sites into one.
Triangle pairs are coarse-grained into Posners.
Each Posner forms a site in a square lattice
(Fig.~\ref{fig:Hex_square}).

Imagine Bob measuring the
$\mathbf{S}_{123}^2  \otimes  \mathbf{S}_{456}^2$
of a Posner $P$.
Suppose that the ``error'' outcome $\frac{1}{2}$ obtains.
Bob might discard $P$.
Alternatively, he might measure the 
$S^{ \zlab }_{1 \ldots 6}$ of $P$.
The measurement would ``terminate the lattice,''
forming a boundary.

On average, $\frac{2}{3}$ of the Posners
yield the ``good'' outcome (the $\frac{3}{2}$ outcome).\footnote{
This probability is estimated via 
the technique described in footnote~\ref{footnote:Calc_prob}.}
The 2D square lattice has a site-percolation threshold
of $p_* \approx  0.59$.\footnote{
\emph{Site percolation} is a topic in graph theory and 
statistical mechanics.
Let $\Graph$ denote a graph of $\Sites$ sites.
Consider deleting each site $v$
with probability $1 - p$.
If $v$ is deleted, so are the edges that terminate on $v$.
Let $\Graph'$ denote the remaining graph.
Does a path of edges traverse $\Graph'$ from top to bottom?
If so, $\Graph'$ \emph{percolates}.
$p_*$ denotes the \emph{percolation threshold}.
If $p \geq p_*$, $\Graph'$ percolates
in the limit as $\Sites \to \infty$.
$\Graph'$ does not if $p < p_*$.
A phase transition occurs at $p = p_*$.}
Hence Bob's site-deletion probability exceeds the threshold:
$p  \approx  \frac{2}{3} >  0.59  \approx  p_*$.
A large, richly connected component
spans Bob's graph.
Such components underlie universality~\cite{Wei_12_Two}.

Bob returns to regarding triangles, rather than Posners, as vertices.
The lattice looks hexagonal but contains holes.
A large connected component spans also this graph.
Hence the Wei \emph{et al.} prescription
(Sec.~\ref{section:Wei_to_cluster}),
or a related prescription,
appears likely to transform the state into 
a universal cluster state.

To check, one might refer to~\cite{Browne_08_Phase,Wei_14_Hybrid,Wei_15_Universal}.
The authors consider faulty lattices:
Sites might be deleted, as by measurement.


   %
   \item
   Universal quantum computation is unnecessary
   for achieving quantum supremacy~\cite{Preskill_12_Quantum}.
   Suppose that $\AKLTPrime$ has been converted into
   a cluster state on a 2D square lattice.
   Consider measuring single qubits nonadaptively.
   A random distribution $\mathcal{P}$ is sampled.
   Classical computers are expected not
   to be able to sample from $\mathcal{P}$ efficiently~\cite{Farhi_16_Quantum}.
   

\end{enumerate}
\noindent
Aside from these opportunities,
the calculational technique in footnote~\ref{footnote:Calc_prob}
may be rendered more precise.

\subsubsection{Analysis of the AKLT$'$ state}
\label{section:AKLT_analysis}


\begin{figure}[tb]
\centering
\includegraphics[width=.45\textwidth, clip=true]{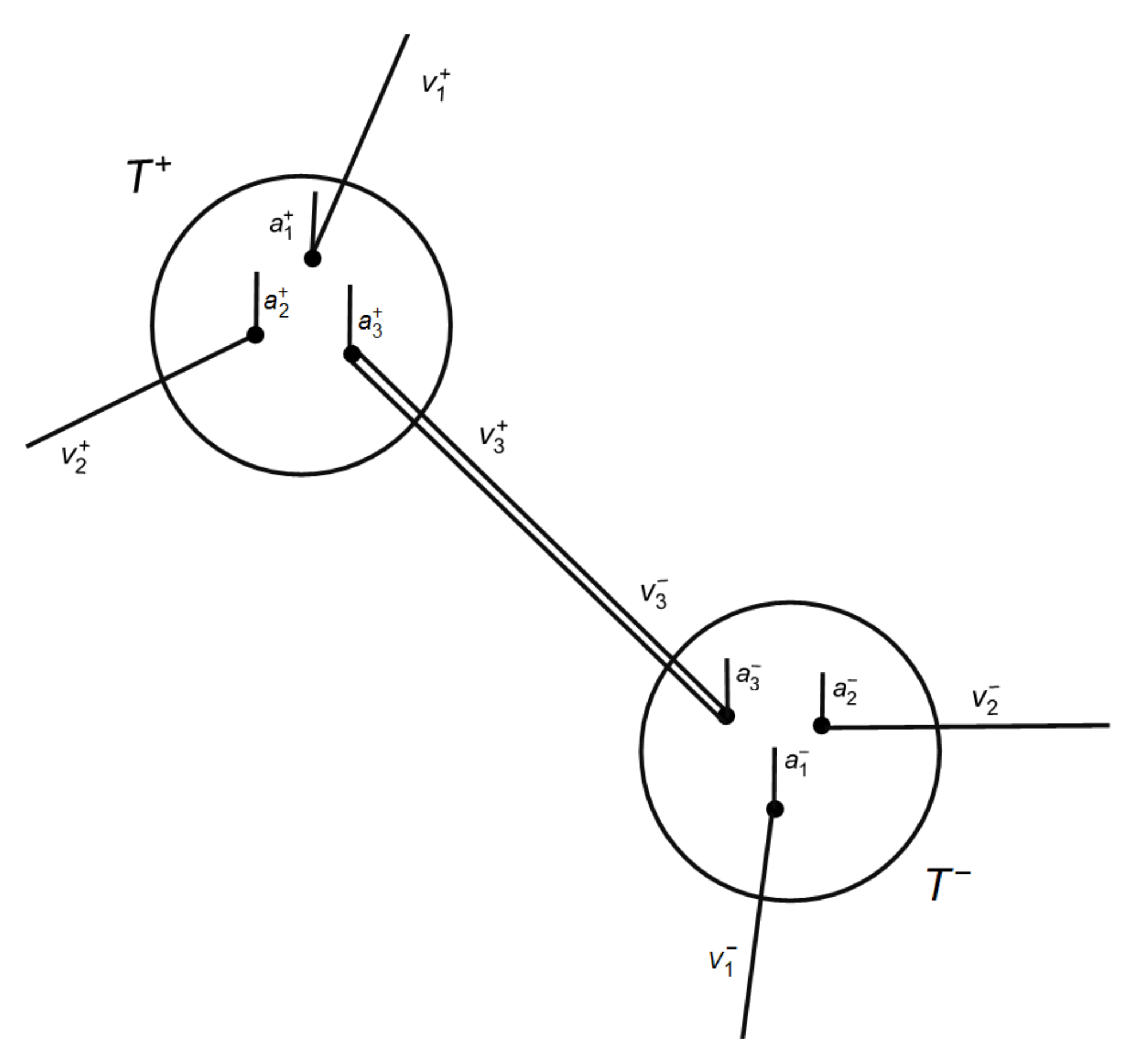}
\caption{\caphead{The AKLT$'$ state $\AKLTPrime$ as a 
projected entangled-pair state (PEPS):}
The two tensors, $T^+$ and $T^-$, 
are repeated to form the PEPS.
$T^+$ represents the state of one triangle in a Posner
(Fig.~\ref{fig:Geo});
$T'$ represents the other triangle's state.
Each tensor has three physical qubits,
labeled $a^+_j$ or $a^-_j$, wherein $j = 1 , 2 , 3$.
Each tensor has three virtual legs,
labeled $v^+_j$ or $v^-_j$,
wherein $j = 1, 2, 3$.
Each of $v^+_1$ and $v^+_2$,
and each of $v^-_1$ and $v^-_2$,
has bond dimension two.
$v^+_3$ has bond dimension six,
as does $v^-_3$.
An implicit Kronecker delta $\delta_{ v^+_3  \,  v^-_3 }$
constrains the virtual indices.}
\label{fig:PEPS}
\end{figure}

MBQC motivated the definition of $\AKLTPrime$.
Yet $\AKLTPrime$ holds interest in its own right.
$\AKLTPrime$ resembles the AKLT state $\AKLTHon$
on a honeycomb lattice.
AKLT states have remarkable properties.
We discuss analogous properties,
and opportunities to seek more analogous properties,
of $\AKLTPrime$.

First, classical resources can compactly represent AKLT states approximately.
The 1D AKLT state $\AKLTOne$ is an MPS~\cite{Ostlund_95_Thermo,Fannes_92_Communications,Klumper_91_Equivalence}.
$\AKLTHon$ is a PEPS~\cite{Verstraete_04_Valence,Verstraete_04_Renormalization,Nishino_01_Two}. 
$\AKLTPrime$ is a PEPS, illustrated in Fig.~\ref{fig:PEPS}
and detailed in App.~\ref{section:PEPS}.

Hence $\AKLTPrime$ is the ground state of
some local, frustration-free Hamiltonian $H_{\AKLT'}$~\cite{Perez_07_PEPS}.
The ground state is unique~\cite{Molnar_17_Generalization}.
The relationship between $H_{\AKLT'}$ and
the Posner Hamiltonian $H_\Posner$ merits study.
So does whether $H_{\AKLT'}$ has a constant-size gap~\cite{Affleck_87_Rigorous,Affleck_88_Valence}.
If $H_{\AKLT'}$ has,
$\AKLTPrime$ can be prepared efficiently via cooling.

Third, $\AKLTPrime$ results from deforming $\AKLTHon$.
AKLT states have been deformed 
via another strategy~\cite{Niggemann_97_Quantum,Verstraete_04_Diverging,Darmawan_12_Measurement,Wei_15_Universal}:
Let $H_\AKLT$ denote the Hamiltonian
whose ground state is the AKLT state of interest.
$H_\AKLT$ is transformed with
a deformation operator $\mathcal{D}(a)$~\cite{Wei_15_Universal}.
The parameter $a$ is tuned, changing the ground state.

$\AKLTPrime$ follows from a different deformation.
We start not from a Hamiltonian, but from
the Hilbert space $( \mathbb{C}^6 )^{ \otimes 2 }$.
Singlets are arranged; then the state is projected
onto the $\tau = 0$ eigenspace $\Hil_{\tau = 0}$.
$\Hil_{\tau = 0}$ contains
the $\frac{3}{2}  \otimes  \frac{3}{2}$ subspace.
Projecting onto the latter subspace
would yield an AKLT state.
Enlarging the projector deforms the state.

Wei \emph{et al.} study an AKLT state's computational power 
as a function of $a$~\cite{Wei_15_Universal}.
Our state's computational power might be studied
as a function of the projected-onto space.


%
%
%
\subsection{Entanglement's effect on molecular-binding rates}
\label{section:Bio_Bell}

Consider two Posners approaching each other 
with the prebinding orientation
described below operation~\ref{item:Meas}.
The Posners might bind together.
They could form subsystems in
a many-body entangled system.
Entanglement affects the binding probability,
Fisher proposes~\cite{Fisher15}.

Fisher illustrates with an example~\cite[p. 5, Fig. 3]{Fisher15}.
Let $a$, $a'$, $b$, and $b'$ denote Posners.
Let $a$ be entangled with $a'$,
and let $b$ be entangled with $b'$.
Suppose that $a$ has bound to $b$.
Suppose that $a'$ approaches $b'$ with the prebinding orientation.
$a'$ and $b'$ have a higher probability of binding, 
Fisher argues, than in the absence of entanglement.
We recast this narrative as a quantum circuit
in Fig.~\ref{fig:Circuit_MF_story}.

Fisher supports his proposal by
analyzing position-and-spin states.
We use, instead, the Posner-binding PVM~\eqref{eq:C_PVM}.
Checking Fisher's example lies beyond
our (classical) computational power.
But we check the principle behind his example quantitatively:
Entanglement, we show, can affect the probability
that two Posners bind
(Fig.~\ref{fig:Binding_rates} and Sec.~\ref{section:Bio_Bell_Ex}).
We identify a percent increase of $198\%$.
This technique can be scaled up to analyze
arbitrarily many Posners.
Random rotations, we find in Sec.~\ref{section:Rot_bind_prob}, 
can eliminate entanglement's effect on average binding rates.

\subsubsection{Illustration: Entanglement's effect on binding rates}
\label{section:Bio_Bell_Ex}

\begin{figure}[tb]
\centering
\begin{subfigure}{0.45\textwidth}
\centering
\includegraphics[width=.7\textwidth]{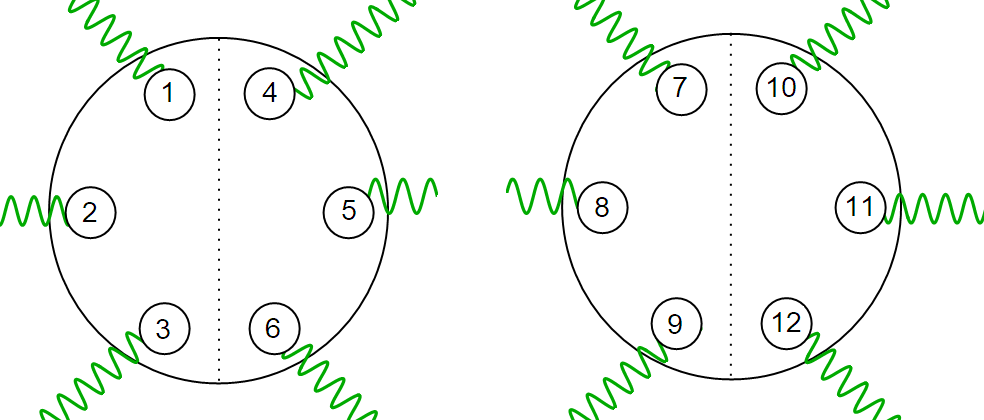}
\caption{}
\label{fig:entangledBindingBaseline}
\end{subfigure}
\begin{subfigure}{0.45\textwidth}
\centering
\includegraphics[width=.7\textwidth]{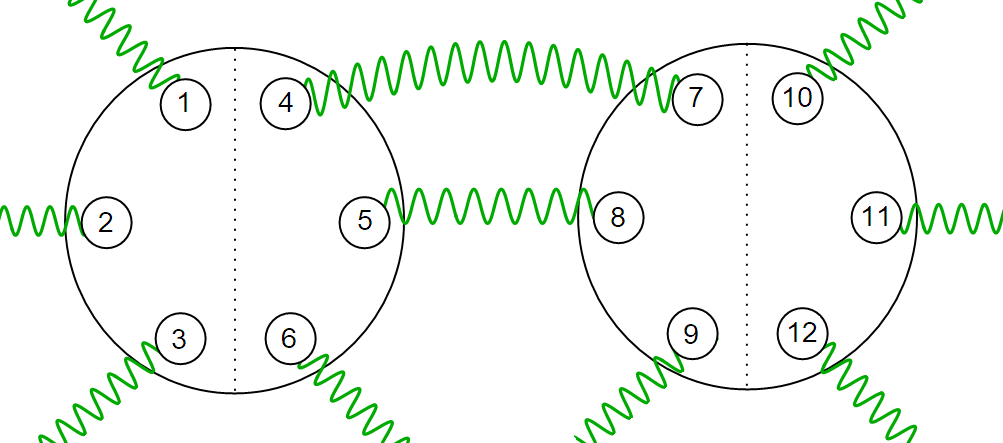}
\caption{}
\label{fig:entangledBindingTwoSinglets}
\end{subfigure}
\begin{subfigure}{0.45\textwidth}
\centering
\includegraphics[width=.7\textwidth]{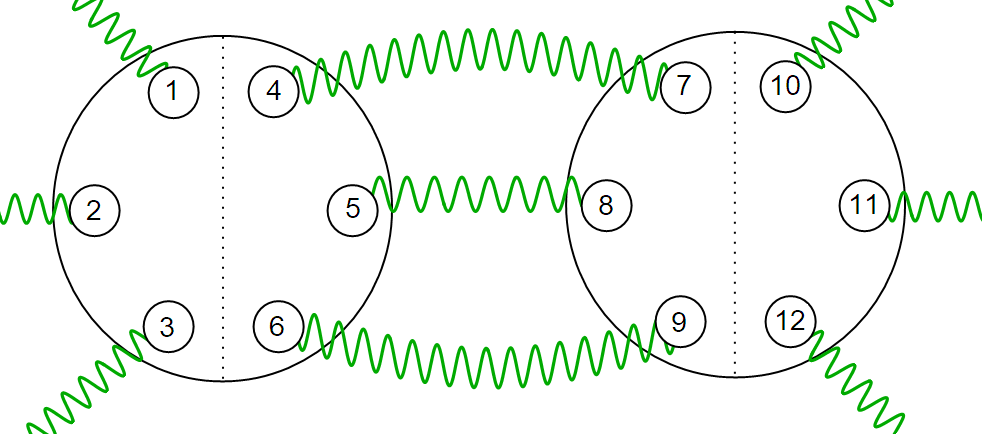}
\caption{}
\label{fig:entangledBindingThreeSinglets}
\end{subfigure}
\begin{subfigure}{0.45\textwidth}
\centering
\includegraphics[width=.7\textwidth]{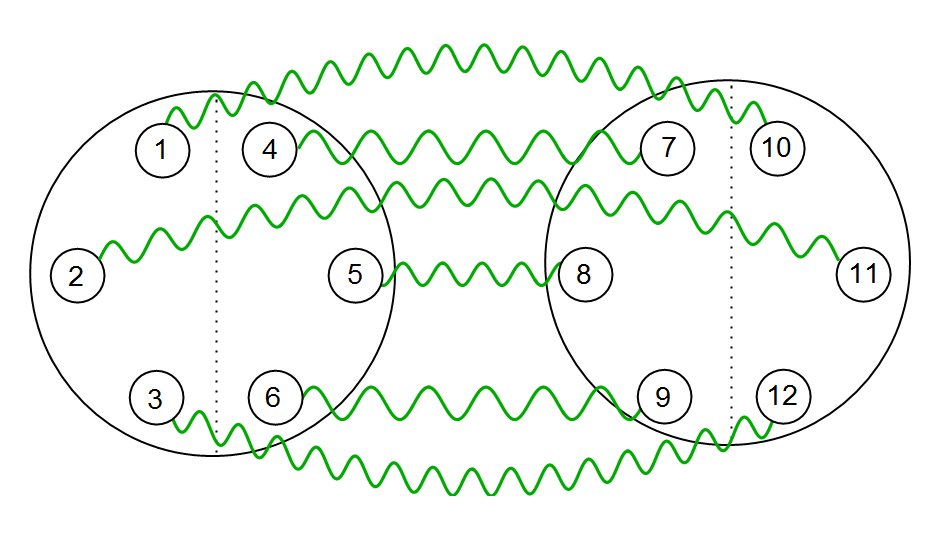}
\caption{}
\label{fig:entangledBindingSixSinglets}
\end{subfigure}
\caption{\caphead{Illustration of entanglement's effect
on molecular-binding rates:} 
The large black circles represent Posner molecules $A$ and $B$. 
Each molecule contains six phosphorus nuclear spins,
represented by small black circles.
The circles group together in trios
(see Fig.~\ref{fig:Geo}).
A dashed line represents the separation between
a Posner's trios.
Two small circles connected by a green, wavy line
form a singlet $\ket{\Psi^-}$ [Eq.~\eqref{eq:Singlet}].
Some spins form singlets with spins in external molecules.
Hence some green, wavy lines have ends that
do not terminate on spins in the diagram.
Figure~\ref{fig:entangledBindingBaseline} shows the base case,
a Posner pair that contains no singlets.
These Posners have a probability 
$\approx  0.336$ of binding together.
Figures~\ref{fig:entangledBindingTwoSinglets},~\ref{fig:entangledBindingThreeSinglets},
and~\ref{fig:entangledBindingSixSinglets}
show Posners that share 
two singlets, three singlets, and six singlets.
These pairs have binding probabilities 
$\approx  0.344$, $0.375$, and 1.
The entanglement patterns raise the probabilities
over the base case by
$\approx 2.38$\%, $\approx 11.6$\%, and $\approx 198$\%.}
\label{fig:Binding_rates}
\end{figure}
%

%
%

Let $A$ and $B$ denote two Posners.
We illustrate with four cases.
\begin{enumerate}

   %
   %
   \item  \textbf{Base case: no singlets:}
Suppose that the $AB$ system contains no singlets
(Fig.~\ref{fig:entangledBindingBaseline}).
Every phosphorus nuclear spin forms a singlet with
an external molecule.
The reduced state of $AB$ equals a product of maximally mixed states:
$\rho_{AB}  =  \frac{ \id_{12} }{ 64 }  \otimes  \frac{ \id_{12} }{ 64 }$.
The identity operator defined on $\mathbb{C}^k$
is denoted by $\id_k$.

Suppose that the Posners approach each other
with the prebinding orientation.
The Posners have a probability
\begin{align}
   \label{eq:Binding_rate_1}
   p_{AB} 
   =  \Tr \left(  \Pi_{ A B }  \,  \rho_{AB}  \right)
   =  \frac{ 43 }{ 128 } 
   \approx 0.336
\end{align}
of binding.

   %
   %
   \item  \textbf{2 singlets shared:}
   Consider any two qubits in the same triangle of $A$,
   e.g., qubits 4 and 5 in Fig.~\ref{fig:entangledBindingTwoSinglets}.
   Let each of these qubits share a singlet with
   the corresponding qubit in
   either triangle of $B$.
   For example, let 4 share a singlet with 7,
   and let 5 share a singlet with 8.
   The Posners occupy the state
   \begin{align}
      \label{eq:rhoABPrime}
      \rho'_{AB}  & =  \left(  \frac{ \id_2 }{ 2 }  \right)^{ \otimes 4 }
      \otimes  \:     _{4,7}\ketbra{ \Psi^- }{ \Psi^- }_{4,7}
      \nonumber \\ & \qquad 
      \otimes  \:     _{5,8}\ketbra{ \Psi^- }{ \Psi^- }_{5,8}  \otimes    \,
      \left(  \frac{ \id_2 }{ 2 }  \right)^{ \otimes 4 }  \, .
   \end{align}
   
   Suppose that the Posners approach each other 
   with the prebinding orientation.
   They have a probability 
   \begin{align}
      \label{eq:Binding_rate_2}
      p'_{AB}  
      &  =  \Tr ( \Pi_{AB}  \,  \rho'_{AB} ) \\
      &  =  0.34375  \approx  0.344
   \end{align} 
   of binding.
   Entanglement raises the binding probability by a fraction
   $\frac{ p'_{AB} }{ p_{AB} }  -  1
      \approx  0.238  \, ,$
or by $\approx 2.38\%$.

   %
   %
    \item    \textbf{3 singlets shared:}
    Let one qubit trio in $A$
    be maximally entangled with
    a qubit trio in $B$.
    Let one singlet link each $A$ qubit to
    the geometrically analogous qubit in the $B$ trio.
    For example, qubits 4, 5, and 6 can form singlets with
    qubits 7, 8, and 9, respectively 
    (Fig.~\ref{fig:entangledBindingThreeSinglets}). 
    The Posners would occupy the state
   \begin{align}
      \rho''_{AB}  & =  
      \left(  \frac{ \id_2 }{ 2 }  \right)^{ \otimes 3 }
      \otimes  \:     _{4,7}\ketbra{ \Psi^- }{ \Psi^- }_{4,7}  \,
      \otimes  \:     _{5,8}\ketbra{ \Psi^- }{ \Psi^- }_{5,8}  
      \nonumber \\ & \qquad
      \otimes  \:     _{6,9}\ketbra{ \Psi^- }{ \Psi^- }_{6,9}  \,
      \otimes
      \left(  \frac{ \id_2 }{ 2 }  \right)^{ \otimes 3 }  \, .
   \end{align}
   
   Suppose that the Posners approached each other 
   with the prebinding orientation.
   They would have a probability
   \begin{align}
      \rho''_{AB}  
      & =  \Tr ( \Pi_{AB}  \,  \rho''_{AB}  )
      =  0.375
   \end{align}
   of binding. The probability has risen by
   $\frac{ p''_{AB} }{ p_{AB} }  -  1
      \approx  0.116  \, ,$
   or by 11.6\%, from the base case.
   
    %
   %
   \item  \textbf{Maximal entanglement: 6 singlets shared:}
   Suppose that each $A$ qubit forms a singlet with
   the corresponding $B$ qubit (Fig.~\ref{fig:entangledBindingSixSinglets}).
   The qubits occupy the state
   \begin{align}
      \rho'''_{AB}  =  \left( \ketbra{ \Psi^- }{ \Psi^- }  \right)^{ \otimes 6 }  \, .
   \end{align}
   The binding probability rises to
   \begin{align}
      p'''_{AB}  =  \Tr ( \Pi_{AB}  \,  \rho'''_{AB} )
      =  1  \, ,
   \end{align}
   in an approximately 198\% increase over the base case.
   Such maximal entanglement ensures that
   the Posners always bind.
   

\end{enumerate}
\noindent 
Linking the Posners with just one singlet
appears not to raising the binding probability
above the baseline $p_{AB}$.

The technique illustrated here can be scaled up.
With more classical computational power, 
one can store larger quantum states.
The four-Posner conjecture illustrated in Fig.~\ref{fig:Circuit_MF_story}
can be checked.
So can entanglement's effect on the binding probabilities of
swarms of Posners.

\subsubsection{Random rotations can eliminate
entanglement's effect on average molecular-binding rates.}
\label{section:Rot_bind_prob}


Consider twelve phosphorus nuclear spins in
a joint state $\rho'_{AB}$ [Eq.~\eqref{eq:rhoABPrime}]. 
Suppose that independent currents
randomly rotate the nuclei.
Each qubit $a$ evolves under some unitary 
$U_{ \hat{n}_a } ( \theta_a )$.
The rotation axis $\hat{n}_a$ and
the rotation angle $\theta_a$ 
may be distributed uniformly.

Suppose that six nuclei form Posner $A$,
while the other nuclei form Posner $B$.
The molecules occupy the joint state
\begin{align}
   \rho''_{AB}  \left( \Set{ \hat{n}_a }  ,  \Set{ \theta_a }  \right)
   & =  \left[ \bigotimes_{a = 1}^{12}
   U_{ \hat{n}_a } ( \theta_a )  \right]
   \rho'_{AB}
   \left[ \bigotimes_{a = 1}^{12}
   U_{ \hat{n}_a } ( \theta_a )^\dag  \right]  \, .
\end{align}
Suppose that $A$ and $B$ approach each other
with the pre-binding orientation.
The Posners have a probability
\begin{align}
   p''_{AB} \left( \Set{ \hat{n}_a }  ,  \Set{ \theta_a }  \right)
   & =  \Tr  \LParen  \Pi_{AB}  \,  
  \rho''_{AB}  \left( \Set{ \hat{n}_a }  ,  \Set{ \theta_a }  \right)  
  \RParen
\end{align}
of binding.

On average over rotations,
the Posners have a probability
\begin{align}
   & \overline{p''}_{AB} 
   =  \int_{ \Set{ \hat{n}_a }  ,  \Set{ \theta_a }  }
   \Tr  \Big(  \Pi_{AB}  \,  
   \rho''_{AB}  \left( \Set{ \hat{n}_a }  ,  \Set{ \theta_a }  \right)  
   \Big)  \\
   & =  \Tr  \Bigg(  \Pi_{AB}  \Bigg[
   \left(  \frac{ \id_2 }{ 2 }  \right)^{ \otimes 3 }  \otimes
   \int_{ \Set{ \hat{n}_a }_{a=3}^{10}  ,  \Set{ \theta_a }_{a=3}^{10} } 
   \\ \nonumber &    \times
   \left[ \bigotimes_{a=3}^{10}
   U_{ \hat{n}_a } ( \theta_a )  \right]
   \ketbra{ \Psi^- }{ \Psi^- }  
   \left[ \bigotimes_{a=3}^{10}
   U_{ \hat{n}_a }^\dag ( \theta_a )  \right]  \otimes
   \left(  \frac{ \id_2 }{ 2 }  \right)^{ \otimes 3 }
    \Bigg]  \Bigg)  \\
    & = \Tr  \left(  \Pi_{AB}  \left(  \frac{ \id_2 }{ 2 }  \right)^{ \otimes 12 }  \right)  \\
    & =  p_{AB}
\end{align}
of binding.
Uniformly random rotations effectively decohere 
the internal singlets, on average.
The binding probability reduces to 
its non-singlet-enhanced value.

We emphasize that this result concerns an average over 
rotations by many \emph{different} pairs of molecules.
Each individual molecule pair can experience 
a binding-probability boost due to entanglement.

\subsection{Measuring quantum cognition against 
DiVincenzo's criteria for quantum computation and communication}
\label{section:diV}

Consider attempting to realize universal quantum computation 
and quantum communication
with any physical platform.
Which requirements must 
the physical components and processes satisfy?
DiVincenzo catalogued these requirements~\cite{diVincenzo_00_Physical}.
Five of \emph{diVincenzo's criteria} underpin quantum computation.
Quantum communication requires another two criteria.
Quantum cognition, we find,
satisfies DiVincenzo's criteria,
except perhaps universality.

We continue assuming that 
Posner operations can be performed with fine control.
This assumption is discussed in Sec.~\ref{section:Control}.
The assumption may appear questionable here:
Practicalities partially concerned DiVincenzo.
Yet the fine-control assumption facilitates a first-step analysis
of what the model can achieve in principle.
Incorporating randomness 
forms an opportunity for future research.

\textbf{1. ``A scalable physical system with well-characterized qubits'':}
Different physical DOFs encode QI
at different stages of a quantum-cognition computation
(Sec.~\ref{section:How_to_enc}).
Initially, phosphorus-31 ($\Thirtyone$) atoms' nuclear spins
serve as qubits.
These atoms occupy free-floating phosphate ions.

Six phosphates can join together, forming a Posner molecule.
Each state of six independent spins transforms into one 
antisymmetrized state~\eqref{eq:Slater}.
Hence antisymmetrized spin-and-position states store QI.

Third, each Posner has an observable $\GC$
(Sec.~\ref{section:C_symm}).
The eigenvalue $\tau$ 
assumes one of three possible values: $\tau = 0 , \pm 1$.
The eigenspaces $\Hil_\tau$ 
have 24-, 20-, and 20-fold degeneracies.
A state $\ket{ \tau{=}j }$ can be chosen
from each $\Hil_{ \tau{=}j }$ eigenspace.
${\rm span} \Set{ 
\ket{ \tau{=}0 } ,  \ket{ \tau{=}1 } ,  \ket{ \tau{=}2 } }$
forms an effective qutrit space 
(Sections~\ref{section:Qutrit_code} and~\ref{section:Stick_apps}).
But Posners are not associated uniquely with effective qutrits,
to our knowledge (footnote~\ref{footnote:Pseudospin}).

QI-storing Posner systems can be scaled up spatially.
$\Thirtyone$ nuclear spins can form singlets
distributed across Posners.
Entangled Posners can form lattices
that can power universal MBQC (Sec.~\ref{section:AKLT}).


%
\textbf{2. ``The ability to initialize the state of the qubits to a simple fiducial state, such as $\ket{000 \ldots}$'':}
Phosphorus nuclear spins can be prepared in singlets $\ket{ \Psi^- }$
(operation~\ref{item:Bell}).
A singlet forms as
the enzyme pyrophosphatase cleaves a diphosphate ion.
The resultant two phosphates are projected onto $\ket{ \Psi^- }$.

\textbf{3. ``Long relevant decoherence times, 
much longer than the gate-operation time'':}
Fisher has catalogued sources of decoherence
and has estimated coherence times~\cite{Fisher15,Fisher_17_Are,Fisher_17_Personal,Swift_17_Posner}.
Nearby spins and electrons generate electric and magnetic fields.
Protons $\H^+$ threaten $\Thirtyone$ spins the most.
The $\Thirtyone$ spins can entangle with external DOFs via
magnetic dipole-dipole coupling.

But phosphates, Posners, and other small particles tumble in solution.
As the particles move around each other,
the fields experienced by a $\Thirtyone$ changes.
The fields are expected to vanish
on average over tumbles.

$\Thirtyone$ nuclear spins in free phosphates, Fisher estimates,
remain coherent for about a second.
$\Thirtyone$ nuclear spins in Posners
remain coherent for $\sim 10^5-10^6$ s.

Fisher writes also that $\tau$ labels a pseudospin
``very isolated from the environment,
with potentially extremely long (days, weeks, months, \ldots) 
decoherence times''~\cite{Fisher_17_Are}.
As explained in Sec.~\ref{section:C_symm},
to our knowledge,
$\tau$ does not label any unique quantum three-level pseudospin. 
We translate Fisher's statement as follows:
Consider preparing a $\GC$ eigenstate
associated with the eigenvalue $\tau = j$.
Consider waiting, then measuring $\GC$.
How would long would you have to wait
to have a high probability of obtaining an outcome $\tau = k \neq j$?
At least days.

These coherence times, we expect,
exceed entangling-gate times.
Posner binding (operation~\ref{item:Meas}) consists of electronic dynamics.
Electrons have much shorter time scales than nuclei, in general.
Posner binding releases about 1 eV of energy
to the environment~\cite{Fisher15,Swift_17_Posner}.
A rough estimate for the binding's time scale is
$t_{\rm bind}  \sim  \frac{ \hbar}{ 1 \; \text{eV} }  
\sim  10^{ -15 }$ s $\ll$ $10^5$ s.
As many as $10^{20}$ entangling gates might be performed
before the spins decohere.

Could many single-qubit rotations be performed?
We estimate that a qubit can rotate through 
an angle $\sim  \pi$ in a time
$t'_\rot  \sim  10 \text{ s}
\ll 10^5 \text{ s}$, in the best case
(App.~\ref{section:One_Qubit_Gates}).
Hence the best-case single-qubit-rotation time scale 
is much less than
the out-of-Posner decoherence time.

Qubits should be able not only to rotate and entangle,
but also to leave and enter Posners
(via Operations~\ref{item:Form_Pos},~\ref{item:pH},
and~\ref{item:Disperse_bind})
before decohering.
Molecular formation and hydrolyzation transfers
about 1 eV between molecules and their environments.
This energy scale suggests a time scale of
$\frac{ \hbar }{ 1 \text{ eV} }
\approx 1 \text{ fs}
\ll 1$ s.
We therefore expect phosphates to be able 
to leave and enter Posners 
before the spins decohere.

\textbf{4. ``A `universal' set of quantum gates'':}
The biological qubits can undergo 
Posner operations, introduced in Sec.~\ref{section:Abstract2}.
The operations include single-qubit rotations
and entangling operations.
Analyses appear in Sections~\ref{section:Analyze_abstract} and~\ref{section:Stick_apps}.

We detail the rotation mechanism in App.~\ref{section:One_Qubit_Gates}.
Neuron firing generates magnetic fields $\mathbf{B}$.
We model these fields as external and classical.
(These fields contrast with 
the influence, mentioned earlier, of $\H^+$ spins.
$\H^+$ spins can entangle with $\Thirtyone$ nuclear spins,
decohering the qubits.
External classical fields cannot, to a good approximation.)
Let $\bm{\mu}$ denote
a $\Thirtyone$ nucleus's spin magnetic moment.
Neuronal-current fields could rotate the spins via the Hamiltonian
$H_{\rm mag}  =  -  \bm{\mu}  \cdot  \mathbf{B}$.
In the best case, spins could rotate through
angles of up to $\pi$.
We expect typical rotations to be
through much smaller angles, though.

Qubit gates induce gates on
the effective qutrits.
The induced gates depend on 
how the qutrits are defined.

Whether the gates form a universal set---in
transforming the qubits or the qutrits---remains
an open question.
Posner operations can efficiently prepare
a state $\AKLTPrime$ that can fuel
universal MBQC (Sec.~\ref{section:AKLT}).
Whether Posner operations can implement MBQC
remains unknown.

\textbf{5. ``A qubit-specific measurement capability'':}
Posner binding (operation~\ref{item:Meas}) measures 
whether $\tau_A + \tau_B = 0$.
Whether two molecules have bound together
is a classical property.
Each Posner's center of mass is a classical DOF:
Water and other molecules bounce off the Posner frequently.
The bounces measure the Posner's position,
precluding coherence.
Hence two Posners' closeness is a classical DOF.
Closeness serves as a proxy for binding.
Hence whether two Posners have bound
is a classical DOF.
This classical memory records whether $\tau_A + \tau_B = 0$.

Fisher proposed that the readout 
might be amplified further~\cite{Fisher15}.
Posner binding could impact later Posner binding,
then the hydrolyzation of Posners,
then neurons' $\Ca\Two$ concentrations,
and then neuron firing.
This process is overviewed in this paper's Sec.~\ref{section:Backgrnd_MF}.

\textbf{6. ``The ability to interconvert stationary and flying qubits'':}
Computing is often easiest with unmoving hardware.
\emph{Stationary qubits} remain approximately fixed.
They undergo computation.
Then, their state can be transferred to \emph{flying qubits}.
Flying qubits move easily.
They can bring states together
for joint processing.

Consider, for example, a quantum algorithm
that contains subroutines.
Different labs' quantum dots
could implement different subroutines.
The quantum dots' states could be converted into photonic states.
The photons could travel down optical fibers to a central lab.
There, the algorithm's final steps could be implemented.

Phosphorus nuclear spins could serve as stationary qubits and as flying qubits.
Phosphorus atoms occupy lone phosphates and Posners.
In each setting, the nuclear spins undergo computations
(Sec.~\ref{section:Abstract2}).
But Posners protect the spins from decoherence 
better than lone phosphates do~\cite{Fisher15}.
Hence Posners form better flying qubits.

The projector $\PiPos$ [Eq.~\eqref{eq:PiMinus}]
transforms phosphate states into a Posner state.
The projection forms a one-to-one map.
Hence ``stationary'' phosphates' states are converted into
a ``flying'' Posner's state faithfully.

\textbf{7. ``The ability faithfully to transmit flying qubits between specified locations'':}
Posners diffuse through intracellular and extracellular fluid.
A protein could transport Posners into neurons~\cite{Fisher15,Fisher_17_Personal}:
the \emph{vesicular glutamate transporter} (VGLUT)~\cite{VGLUT1_2000,VGLUT_review_2007,SLC17_2010,SLC17_2013},
alias the \emph{brain-specific (B) sodium-dependent (${\rm Na}^+$) 
inorganic-phosphate (Pi) cotransporter} (BNPI)~\cite{BNPI_1994}.
VGLUT sits in cell membranes, 
through which the protein could ferry Posners.
Posners protect in-transit $\Thirtyone$ nuclear spins for
$\sim 10^5 - 10^6$ s~\cite{Fisher15,Fisher_17_Quantum},
as discussed above.

For how long do Posners diffuse
between neurons?
We estimate by dimensional analysis.
The diffusion constant $D$ has dimensions of $\text{distance}^2 / \text{time}$:
\begin{align}
   \label{eq:Diffuse1}
   D  \sim  \frac{ \ell^2 }{ t_\diffuse }  \, .
\end{align}
The time scale over which 
a Posner diffuses between neurons
is denoted by $t_\diffuse$.
A typical synapse has an area of $\ell^2  \sim  10^{-2}  \;  \mu\text{m}^2$~\cite[Fig. 2]{Milo_15_Cell_Synapse}.

We estimate $D$ via the Einstein-Stokes relation,
\begin{align}
   \label{eq:Ein_Stokes}
   D  =  \frac{ \kB T }{ 6 \pi  \eta r }  \, .
\end{align}
Equation~\eqref{eq:Ein_Stokes} describes a radius-$r$ sphere in
a viscosity-$\eta$ fluid.
Water has a viscosity $\eta  \sim  10^{-3} \; \text{N} \cdot \text{s} / \text{m}^2$.
A Posner molecule has a radius $r \sim 10$ \AA~\cite{Yin_03_Posner}.

We substitute these numbers, with $\kB  \sim 10^{ - 23}$ J$/$K,
$T  \sim  10^2$ K, and $6 \pi  \approx 10$, into Eq.~\eqref{eq:Ein_Stokes}:
$D  \sim  10^{ -10 }  \;  \text{m}^2 / \text{s}$.
We substitute into Eq.~\eqref{eq:Diffuse1}, 
upon solving for $t_\diffuse$:
\begin{align}
   \label{eq:Diffuse2}
   t_\diffuse  \sim  \frac{ \ell^2 }{ D }
   & \sim  \frac{ 10^{ - 2 }   ( 10^{ -6 }  \;  \text{m} )^2 }{
                           10^{ - 10 }  \;  \text{m}^2 / \text{s} }  \\
   &  =  0.1  \;  \text{ms}  
   \ll  10^5  \;  \text{s}  \, .
\end{align} 
Hence Posners are expected to be able to traverse a synapse
before their phosphorus nuclear spins decohere.

\section{Discussion and Conclusions}
\label{section:Outlook}

This paper establishes a framework
for the QI-theoretic analysis of Posner chemistry.
The paper also presents applications of Posners
to QI processing:
to QI storage and protection, to quantum communication,
and to quantum computation.
Many QI applications of Posners await discovery, we expect.
In turn, QI motivates quantum-chemistry questions.
Opportunities are discussed below.

\textbf{Quantum error-correcting and -detecting codes:}
We presented one quantum error-detecting code
and one error-correcting code
accessible to Posners.
Other accessible codes might 
protect more information against more errors.
Ideally, one would show how Posner operations,
or a biochemically reasonable extension thereof,
could (i) prepare states in the codespace and 
(ii) detect and correct errors.

Furthermore, one conserved charge
``protects'' each of our codes.
In the error-detecting code, for example,
the codewords $\ket{ \jt }$ correspond to
distinct eigenvalues of $\GC$.
The natural dynamics protect $\GC$.
Hence the dynamics should not map
any codeword $\ket{ \jt }$ into
any other $\ket{ k_\Logg }$.
But the dynamics could map $\ket{ \jt }$
to another state $\ket{ \jt'}$
in the $\tau = j$ eigenspace.

Imagine a more robust code:
A complete set of quantum numbers
(e.g., $\Set{ \tau , m_{ 1 \ldots 6 } , \ldots }$)
would label each codeword.
The dynamics could not map any codeword 
$\ket{ \tau , m_{ 1 \ldots 6 } , \ldots }$
into any other codeword 
$\ket{ \tau' , m'_{ 1 \ldots 6 } , \ldots }$.
Such a code would enjoy considerable protection
by charge preservation.

Relatedly, quantum codes have been cast as
the ground spaces of Hamiltonians.
Every code's states, $\ket{ \bar{\psi} }$, occupy
a Hilbert space $\bar{ \Hil }$.
Suppose that $\bar{ \Hil }$ is the ground space of a Hamiltonian $H$.
Suppose that the system is in thermal equilibrium
at a low temperature $T = \frac{1}{ \kB T }$.
The system has a high probability of remaining in $\bar{ \Hil }$.
Entropy suppresses errors.
Equivalently, the code detects errors.
The Posner Hamiltonian $H_\Posner$ was characterized
shortly after this paper's initial release~\cite{Swift_17_Posner}.
The ground space might point to 
an entropically preserved a code.

\textbf{Quantum algorithms:}
Posners might perform quantum algorithms of two types:
(i) Known algorithms~\cite{Jordan_Zoo} might decompose into Posner operations.
(ii) Posner operations could inspire hitherto-unknown quantum algorithms.

\textbf{Reverse-engineering:}
QI processing could guide conjectures about quantum chemistry.
Fisher reverse-engineered physical mechanisms by which
entanglement could impact cognition~\cite{Fisher15}.
Similarly, one might reverse-engineer physical mechanisms by which
Posners could process QI.
This paper motivates reverse-engineering opportunities:
\begin{enumerate}[leftmargin=*]

   \item
   Section~\ref{section:Posner_AKLT} details 
   how Posner operations can efficiently prepare states
   that can fuel universal MBQC.
   To use the states, one performs the operations
   in Sec.~\ref{section:Bob_AKLT}.
   Example operations include 
   (i) measurements of the POVM $\Set{ F_x, F_y, F_z, }$
   [Eq.~\eqref{eq:Fs}] and
   (ii) adaptive single-qubit measurements.
   Could biological systems implement these operations?
   
   %
   
   %
   \item
   Reverse-engineer a measurement of
   the generator $\GC$ of
   the permutation operator $\CThree$.
   If $\GC$ can be measured,
   incoherently teleported random variables 
   can be used easily (Sec.~\ref{section:Tele_protocol}).

\end{enumerate}

\textbf{Quantum computational complexity and universality:}
Posner operations (Sec.~\ref{section:Abstract_logic})
constitute a model of quantum computation.
Which set of problems can this model solve efficiently?
Let PosQP denote the class of computational problems
solvable efficiently with Posner quantum computation.

Whether Posner quantum computation is universal
remains an open question.
(See Sec.~\ref{section:Ent_capacity} for an elaboration.)
Suppose that the model were universal.
PosQP would equal BQP
(the class of problems that a quantum computer
can solve in polynomial time~\cite{NielsenC10}).
But perhaps $\text{PosQP} \subset \text{BQP}$.
PosQP merits characterization.

\textbf{AKLT$'$ state and MBQC protocol:}
Posner operations can efficiently prepare
a state $\AKLTPrime$ that fuels universal MBQC (Sec.~\ref{section:AKLT}).
The state preparation may be simplified.
Opportunities are detailed in Sec.~\ref{section:Analyze_AKLT_prep}.
Also, $\AKLTPrime$ holds interest
outside of MBQC.
Properties to explore are discussed in Sec.~\ref{section:AKLT_analysis}.

\textbf{Entanglement's effect on binding rates 
and biological Bell tests:}
Entanglement between Posners affects binding rates.
So Fisher conjectured in~\cite{Fisher15}.
The conjecture grew from analyses of spin-and-orbital states.
We supported the conjecture with a two-Posner example,
using a PVM
(Sec.~\ref{section:Bio_Bell}).
The example illustrates
how to check Fisher's conjecture with 
the formalism of QI.
Larger-scale calculations could test 
(i) Fisher's four-Posner conjecture and
(ii) entanglement's effects on the binding probabilities of
swarms of Posners.

Moreover, the QI formalism
could lead to a framework for \emph{biological Bell tests}.
Such tests might be cast as nonlocal games~\cite{Palazuelos_16_Survey}.
The Clauser-Holt-Shimony-Hauser (CHSH) game, which illustrates Bell's theorem~\cite{Clauser_69_Proposed,Preskill_01_Entanglement},
can serve as a model.

\textbf{Quantum chemistry:}
Physical conjectures populate
Sections~\ref{section:When_Pos_form},~\ref{section:Formal_Pos_form}, 
and~\ref{section:How_to_enc_transf}.
These conjectures merit testing and refinement.
First, Posner creation was modeled with
a Lennard-Jones potential.
Second, pre-Posner spin states were assumed 
to transform deterministically
into antisymmetric Posner states.
The pre-Posner orbital state was assumed
to determine the map.
Third, Posner creation was assumed to preserve
each spin's $S^{ \zlab }$ essentially.
Fourth, Posner dynamics were assumed to preserve
$\CThree$ and $S^{ \zlab }_{1 \ldots 6}$.

\textbf{Randomness:}
Our QI-processing protocols involve
perfect executions of Posner operations.
But Posners suffer magnetic fields somewhat randomly.
Randomness could hinder some, and improve some, QI processing.

For example, Sec.~\ref{section:AKLT} features a honeycomb lattice.
Singlets might not form a honeycomb in solution.
(See footnote~\ref{footnote:Hydroxy_hex} for
a reason why regular graphs might form.)
They might have a greater probability of 
forming a random graph.
Randomness could improve the state's connectivity.
Improved connectivity might lower the bar 
for fueling universal MBQC (Sec.~\ref{section:Analyze_AKLT_prep}).
What randomness helps, and what randomness hinders,
merits investigation.

%
%
\begin{acknowledgments}
The authors thank Ning Bao, Philippe Faist, Matthew Fisher, Steve Flammia, Yaodong Li, Leo Radzihovsky, and Tzu-Chieh Wei for discussions.
We thank Fernando Pastawski for help with constructing the quantum error-detecting code.
NYH thanks John Preskill for nudges toward this paper's topic
and for feedback about drafts.
We are grateful for funding from the Institute for Quantum Information and Matter, an NSF Physics Frontiers Center (NSF Grant PHY-1125565) with support from the Gordon and Betty Moore Foundation (GBMF-2644).
This research was partially supported by the NSF also 
under Grant No. NSF PHY-1125915.
NYH is grateful for partial support from the Walter Burke Institute for Theoretical Physics at Caltech, for a Graduate Fellowship from the Kavli Institute for Theoretical Physics, for a Barbara Groce Graduate Fellowship,
and for an NSF grant for the Institute for Theoretical Atomic, Molecular, and Optical Physics at Harvard University and the Smithsonian Astrophysical Observatory. 
\end{acknowledgments}

%
%
\onecolumngrid
\begin{appendices}


\renewcommand{\thesection}{\Alph{section}}
\renewcommand{\thesubsection}{\Alph{section} \arabic{subsection}}
\renewcommand{\thesubsection}{\Alph{section} \arabic{subsection}}
\renewcommand{\thesubsubsection}{\Alph{section} \arabic{subsection} \roman{subsubsection}}

\makeatletter\@addtoreset{equation}{section}
\def\theequation{\thesection\arabic{equation}}

\section{Background: Quantum information theory}
\label{section:QI_backgrnd}

Quantum systems can process information more efficiently,
transmit information more compactly,
and secure information more reliably
than classical systems can.
Consider a system of $\Sites$ qubits,
e.g., $\Sites$ phosphorus nuclear spins.
The system corresponds to a Hilbert space $\Hil$ 
of dimensionality $2^\Sites$.
Let $\Set{ \ket{ \phi_j } }$ denote 
an orthonormal basis for $\Hil$.
The system can occupy a quantum state
$\ket{ \psi }   = \sum_j  c_j  \ket{ \phi_j }  \in \Hil$.
The $2^\Sites$ coefficients $c_j  \in  \mathbb{C}$ satisfy 
the normalization condition $\sum_j  | c_j |^2  =  1$.
Consider specifying one of the $2^\Sites$ basis elements $\ket{ \phi_j }$.
One must use $2^\Sites$ bits (two-level units of classical information).
The specification requires only $\Sites$ qubits.
One can leverage this discrepancy to process information quickly, 
using quantum systems.
The state $\ket{ \psi }$ constitutes QI.

QI can be processed
with help from entanglement~\cite{NielsenC10,Preskill_99_QEC}. 
\emph{Entanglement} manifests in
correlations stronger than any shareable 
by classical systems.
Entanglement facilitates quantum computation, communication, and cryptography.
We briefly review efficiency, quantum computational models and universality,
and quantum error correction.
Readers seeking more background
are referred to~\cite{NielsenC10,Preskill_99_QEC}.

\subsection{Efficiency}
\label{section:Backgrnd_efficiency}

Quantum computers can efficiently solve certain problems 
that, according to widespread belief, classical computers cannot.
\emph{Efficiently} loosely means the following.
Consider a family $F$ of computational problems.
For example, consider receiving a number $\mathcal{N}$
whose prime factors you must identify.
An instance of $F$ consists of, e.g., the number $\mathcal{N}$ to be factored.
Let $n$ quantify the resources 
required to specify an instance of $F$.
For example, $n$ might equal the number of bits
needed to represent $\mathcal{N}$.
Let $t$ denote the time required to solve the instance.
Suppose that the time grows, at most, 
polynomially in the amount of resources:
$t \sim (\const) n^k$, for some $k \geq 0$.
The problems in $F$ can be solved efficiently.

Quantum computers can factor arbitrary numbers
more quickly than classical computers can~\cite{Shor_97_Polynomial}.
Imagine using a quantum computer 
to solve a problem more quickly
than any classical computer.
One would achieve \emph{quantum speedup},
or \emph{quantum supremacy}~\cite{Preskill_12_Quantum}.\footnote{
Preskill coined the term ``quantum supremacy'' in~\cite{Preskill_12_Quantum}. 
The paper concludes with quantum computing's potential:
``How might quantum computers change the world? Predictions are never easy, but it would be especially presumptuous to believe that our limited classical minds can divine the future course of quantum information science.''
Posners suggest that we have better chances than Preskill expected.}

%
%
%
\subsection{Quantum-computation models and universality}
\label{section:Backgrnd_models}

A general quantum process consists of 
state preparations, evolutions, and measurements.
Which operations can be implemented easily
(which states $\ket{ \psi }$ can be prepared easily, etc.)
varies from platform to platform.
Consider, for example, a nuclear-magnetic resonance (NMR) experiment.
Let $\Sites$ denote the number of nuclear spins.
Preparing the pure state $\ket{ 0 }^{ \otimes \Sites }$ is difficult.
Preparing a maximally mixed state 
$\id / 2^{\Sites - 1 }$ of $\Sites - 1$ spins,
tensored with one pure $\ket{ 0 }$,
is easier~\cite{Knill_98_Power}.
A set of quantum resources---of performable quantum operations---forms
a \emph{model for quantum computation}.
DiVincenzo catalogued the ingredients needed to realize 
a quantum-computation model physically~\cite{diVincenzo_00_Physical}.

Certain computational models are universal~\cite{Deutsch_85_Quantum}. A universal quantum computer can perform every conceivable quantum computation. Every universal model can simulate every other universal model efficiently. 

Many quantum-computation models exist.
Two prove most pertinent to this paper:
the circuit model~\cite{Deutsch_89_Quantum} and
measurement-based quantum computation (MBQC)~\cite{Briegel_01_Persistent,Raussendorf_03_Measurement,Briegel_09_Measurement}.
Other models include 
the quantum Turing machine~\cite{Deutsch_85_Quantum},
the one-clean-qubit model~\cite{Knill_98_Power},
adiabatic quantum computation~\cite{Farhi_01_Quantum},
anyonic quantum computation~\cite{Kitaev_03_Fault}, 
teleportation-based quantum computation~\cite{Gottesman_99_Demonstrating,Knill_01_Scheme,Nielsen_03_Quantum,Zhou_00_Methodology},
quantum walks on graphs~\cite{Childs_09_Universal}, and
permutational quantum computation~\cite{Jordan_09_Permutational}.

The \emph{circuit model} is used most widely~\cite{Deutsch_89_Quantum}.
One solves a problem by running a quantum circuit,
illustrated by a circuit diagram (e.g., Fig.~\ref{fig:Circuit_MF_story}).
Wires represents the qubits, which are often
prepared in pure states $\ket{ 0 }$.
Rectangles represent unitary operations $U$.
The $U$'s evolve the qubits,
implementing gates.
A rectangle inscribed with a dial represents a measurement.
Single qubits can be measured with respect to some orthonormal basis,
e.g., $\Set{ \ket{ 0 }, \ket{1} }$, wherein $\langle 0 | 1 \rangle = 0$.

\emph{Depth} quantifies a circuit's length, or complexity.
Consider grouping together
the operations that can be performed simultaneously.
For example, qubit 1 can interact with qubit 2
while qubit 3 interacts with qubit 4.
Each group of gates occurs during one time slice.
The number of time slices in a circuit equals
the circuit's depth.
Suppose that the depth does not depend on the number of qubits.
Such a circuit has \emph{constant depth}.

\emph{Primitive unitaries} can be implemented directly.
Composing primitives simulates more-complicated operations.
One universal primitive set~\cite{Kitaev_97_Quantum,NielsenC10}
is natural to compare with $\Thirtyone$ dynamics:
(i) Each qubit's state can rotate through a fixed angle $\theta$
about a fixed axis $\hat{n}$ of the Bloch sphere.\footnote{
The \emph{Bloch sphere} represents pure qubit states geometrically~\cite{NielsenC10}.
A general pure qubit state has the form
$\ket{ \psi }  =  \cos \frac{\theta}{2}  \ket{ 0 }  
+  e^{ i \varphi} \sin \frac{\theta}{2} \ket{1}$,
wherein $\theta, \varphi \in [ 0, 2 \pi )$.
The state is equivalent to the \emph{Bloch vector}
$\left( \sin \theta  \:  \cos  \varphi ,  \sin \theta  \:  \sin \varphi,
\cos  \theta \right)$.
The Bloch vector lies on the unit sphere, 
or Bloch sphere.
Points inside the sphere represent mixed states 
$\rho  \neq  \ketbra{ \psi }{ \psi }$.}
$\theta$ must be an irrational multiple of $2 \pi$.
(ii) Each qubit can rotate through a fixed angle $\theta'$
about a fixed axis $\hat{n}' \neq \hat{n}$.
(iii) Any two qubits can be entangled 
via some fixed unitary.

No unitary is known to entangle Posners'
phosphorus nuclear spins.
Hence we turn from the circuit model to MBQC~\cite{Briegel_01_Persistent,Raussendorf_03_Measurement,Briegel_09_Measurement}.
To implement MBQC, one prepares
a many-qubit entangled state $\ket{ \psi }$.
One measures single qubits adaptively.
Measurements are \emph{adaptive} if
earlier measurements' outcomes dictate
later measurements' forms.

Certain states $\ket{ \psi }$ enable one to simulate efficiently, via MBQC,
a universal quantum computer.
Example states include the 
\emph{Affleck-Kennedy-Lieb-Tasaki (AKLT) state}
on a honeycomb lattice,
$\AKLTHon$~\cite{Briegel_01_Persistent,Raussendorf_03_Measurement,VandenNest_06_Universal,Miyake_11_Quantum}.
Posners can occupy a similar state, $\AKLTPrime$.
$\AKLTPrime$ can fuel universal MBQC
(Sec.~\ref{section:AKLT}).

\subsection{Quantum error correction}
\label{section:Backgrnd_QEC}

Two sources of error threaten quantum computers.
First, the operations performed might differ from 
the target operations.
Consider, for example, trying to rotate a qubit
through an angle $\frac{\pi}{2}$ about the $z$-axis.
One might overshoot or undershoot. 
The qubit would rotate through an angle $\frac{\pi}{2}  +  \epsilon$,
for some $\epsilon  \neq 0$.

Second, a quantum computer might entangle with its environment.
The environment decoheres the computer's state.
QI leaks from the computer into the environment.

\emph{Quantum error correction} preserves QI.
Imagine wishing to process 
a state $\ket{\psi}$ of $k$ qubits.
One chooses an \emph{error-correcting code}.
The code maps $\ket{ \psi }$ to 
a state $\ket{ \bar{\psi} }$ of $n > k$ qubits.
$\ket{ \bar{\psi} }$ undergoes physical processes
that effect \emph{logical operations} on the encoded state.
The logical operations constitute a computation.

Throughout the computation, certain observables $O$ are measured.
Which $O$'s depends on the code.
The measurements' outcomes imply whether an error has occurred
and, if so, which sort of error.
The code dictates how to counteract the error.
The state is typically corrected with some unitary $U$.
After the computation and correction terminate,
the state is decoded.
The computational problem's answer is read out.

A code can \emph{detect} more errors than it can correct.
Suppose that, according to the $O$ measurements,
many errors have corrupted $\ket{ \psi }$. 
Suppose that the code cannot correct all those errors.
The state must be scrapped;
and the computation, reinitiated.
We present a quantum error-detecting code 
and an error-correcting code 
formed from states accessible to Posners (Sec.~\ref{section:QEC}).

Let us review the mathematics of 
quantum error correction and detection (QECD).
Consider encoding $k < n$ logical qubits in $n$ physical qubits.
The physical Hilbert space $\mathbb{C}^{2n}$
has dimensionality $2^n$.
A QECD code is a subspace $\Hil_\Logg  \subset  \mathbb{C}^{2n}$
of dimensionality $2^k  <  2^n$.
Let $\Basis^\comp_\Logg = \Set{ \ket{ j_\Logg } }$ denote 
the code's computational basis.
(See Sec.~\ref{section:Encodings} for
an introduction to computational bases.)

Each quantum error-correcting/-detecting code
corresponds to a set $\Set{ E_\alpha }$ of
correctable/detectable errors.
For example, a code of $n = 9$ physical qubits
has been constructed~\cite{Shor_95_Scheme}.
This code corrects the set of single-qubit Pauli errors,
$\Set{ \sigma^x_1 , \sigma^x_2 , \ldots , \sigma^x_9,
\sigma^z_1 , \ldots, \sigma^z_9 }  \, .$
The shorthand $\sigma^\alpha_j  
\equiv  \id^{ \otimes (j - 1) }  \otimes
\sigma^\alpha_j  \otimes  \id^{ \otimes ( n - j ) }$.
The ability to correct $\sigma^y$ errors follows from
the ability to correct $\sigma^x$ and $\sigma^z$.

Under what conditions can a code $\Hil_\Logg$ detect 
a set $\Set{ E_\alpha }$ of errors?
The code and set must satisfy
the \emph{quantum error-detection criteria},
\begin{align}
   \label{eq:QED_criteria}
   \langle j_\Logg  |  E_\alpha  |  k_\Logg  \rangle
   =  C_\alpha  \,  \delta_{jk}  
   \qquad  \forall  j, k, \alpha  \, .
\end{align}
The Kronecker delta is denoted by $\delta_{ j k }$. 
$C_\alpha$ denotes a constant dependent only on the error $E_\alpha$,
not on the codeword labels $j$ and $k$.
Equation~\eqref{eq:QED_criteria} decomposes into two subcriteria:
the off-diagonal criterion, in which $j \neq k$, and
the diagonal criterion, in which $j = k$.

The \emph{off-diagonal error-detecting criterion} has the form
\begin{align}
   \label{eq:QED_offdiag}
   \langle j_\Logg  |  E_\alpha  |  k_\Logg  \rangle
   =  0
   \qquad  \forall  j \neq k  \, .
\end{align}
No $E_\alpha$ maps any codeword $\ket{ k_\Logg }$ into
any other codeword $\ket{ j_\Logg }$.
The logical states retain their integrity 
under detectable errors.

The \emph{diagonal criterion} has the form
\begin{align}
   \label{eq:QED_diag}
   \langle  j_\Logg  |  E_\alpha  |  j_\Logg  \rangle
   =  C_\alpha
   \qquad  \forall  j , \alpha  \, .
\end{align}
Suppose that $\ket{ j_\Logg }$ is prepared.
The environment might effectively measure 
$\expval{ E_\alpha }$.
The environment gains no information about the state,
according to Eq.~\eqref{eq:QED_diag}:
Every codeword's expectation value equals every other codeword's.
Typical detectable errors $E_\alpha$ operate nontrivially
on just a few close-together qubits.
The codewords are \emph{locally indistinguishable} 
with respect to $\Set{ E_\alpha }$.

Local indistinguishability protects QI:
Suppose that the environment had ``learned'' about $\ket{ j_\Logg }$.
Information would have leaked out of the system.
Highly entangled states are locally indistinguishable:
Entanglement distributes information throughout the system.
Local operations cannot extract
the distributed information.

We have reviewed the error-detection criteria.
Under what conditions can a code $\Hil_\Logg$
correct $\Set{ E_\alpha }$?
The code must satisfy the 
\emph{quantum error-correction criteria}~\cite{Knill_97_Theory,Bennett_96_Mixed,Kribs_05_Unified,Kribs_05_Operator,Nielsen_07_Algebraic,Preskill_99_QEC},
\begin{align}
   \label{eq:QEC_criteria}
   \bra{ j_\Logg }  E_\beta^\dag  E_\alpha  \ket{ k_\Logg }
   =  C_{ \alpha \beta }  \,  \delta_{jk}
   \quad  \forall j, k, \alpha, \beta  \, .
\end{align}
Equation~\eqref{eq:QEC_criteria} is interpreted
similarly to Eq.~\eqref{eq:QED_criteria}.
A code that corrects $(d - 1) / 2$ errors
detects $(d - 1)$ errors.
We refer readers to~\cite{Preskill_99_QEC} for more background.

\section{Multiplicity of (no-colliding-nuclei antisymmetric) subspaces 
accessible to a Posner molecule}
\label{eq:Subtle_HilPos}


A subtlety about $\HilPos$ was glossed over in Sec.~\ref{section:How_to_enc_transf}.
Consider Eq.~\eqref{eq:Slater}.
In every term,
the spin quantum number $m_{ \pi_\alpha (j) }$ appears alongside 
the position $\mathbf{r}_{ \pi_\alpha ( j ) }$.
The tuple 
$( m_{ \pi_\alpha (j) }  ,  \mathbf{r}_{ \pi_\alpha ( j ) } )$
occupies different kets in different terms.
But $m_{ \pi_\alpha (j) }$ remains hitched to 
the same position $\mathbf{r}_{ \pi_\alpha ( j ) }$
throughout the terms.
How are the $m_{ \pi_\alpha (j) }$'s assigned to positions?

This question has a two-part answer.
The choice of coordinate system partially determines the assignments.
So do initial conditions, the pre-Posner phosphates' positions and momenta.

The choice of coordinate system determines
the $\varphi$-value associated with a given $m$-value.
For example, suppose that $m_{1} = 0$.
Should this spin variable be assigned to
$\mathbf{r}_1 = ( \phi ,  h )$, to
$\mathbf{r}_1  =  ( \phi + 2 \pi / 3 ,  h )$, 
or to $\mathbf{r}_1  =  ( \phi + 4 \pi / 3 , h )$?
(Whether $h = h_+$ or $h = h_-$ is irrelevant.)
This assignment is a convention,
because the orientation of $\hat{x}_\In$ is a convention.

We illustrate the answer's second part with an example.
Suppose that three singlets,
\begin{align}
   \label{eq:Singlet_exp}
   \ket{ \Psi^- }^{ \otimes 3 }  & =
   \frac{1}{ \sqrt{2} }  \:  
   ( \ket{ \uparrow }  \ket{ \downarrow }
   -  \ket{ \downarrow }  \ket{ \uparrow } )
   \otimes  \frac{1}{ \sqrt{2} }  \:
   (  \ket{ \uparrow }  \ket{ \downarrow }  
   -  \ket{ \downarrow }  \ket{ \uparrow } )
   \otimes  \frac{1}{ \sqrt{2} }  \:
   (  \ket{ \uparrow }  \ket{ \downarrow }  
   -  \ket{ \downarrow }  \ket{ \uparrow } )  \, ,
\end{align}
join together to form a Posner.
Molecule creation is assumed to preserve
the entanglement within each pair of spins (Sec.~\ref{section:Formal_Pos_form}).
Posner creation maps each six-spin term in~\eqref{eq:Singlet_exp}
to a sum~\eqref{eq:Slater}.
Suppose we choose an intra-Posner coordinate system
such that $\mathbf{r}_1 = ( 0, h )$, for $h = h_+$ or $h = h_-$.
Given that coordinate system,
which value should $\mathbf{r}_2$ assume?
Should the spin at $( 0 , h_\pm )$ form a singlet with
the spin at $( 0,  h_\mp )$, with the spin at $( 2 \pi / 3 ,  h_\pm )$, etc.?
Different answers generate qualitatively different Posner states:
The states transform differently under $\CThree$.
%

The correct answer, we posit, is determined by
the positions and momenta that the phosphates had
at the lip of the Lennard-Jones potential (Sec.~\ref{section:When_Pos_form}).
Different projections of the same initial state, we posit, 
would release different amounts of heat to the environment.

The PVM model in Sec.~\ref{section:Formal_Pos_form}
can now be refined:
Posner creation projects the phosphates' state
with the projector $\PiPos$ onto 
\emph{some} no-colliding-nuclei subspace $\HilPos$
of the antisymmetric subspace.
$6!$ such subspaces exist;
$6!$ possible forms are available to $\PiPos$.
One subspace and projector correspond to
entanglement between 
the $( 0 , h_\pm )$ spin and
the $( 0,  h_\mp )$ spin;
one subspace and projector correspond to
entanglement between
the $( 0 , h_\pm )$ spin and
$( 2 \pi / 3 ,  h_\pm )$ spin; etc.
Hence pre-Posner positions and momenta,
with a choice of coordinate system,
determine to which position 
each spin variable (e.g., $m_{1}$) is assigned
during Posner creation.

%
%
\section{The Posner-molecule Hilbert space $\HilPos$ has dimensionality 64.}
\label{section:Dim_anti}

When a Posner forms, we posit,
$\PiPos$ projects the phosphorus nuclei's joint state [Eq.~\eqref{eq:PiMinus}].
$\PiPos$ defines a map that preserves the dimensionality
of the space available for storing QI, 64.
A counting argument shows why.

%
%
%

Imagine that the phosphorus nuclei were classical and distinguishable.
A tuple $( m_j ,  \mathbf{r}_j )$ would label 
the $j^\th$ nucleus's state.
The spin variable $m_j$ could assume one of two possible values.
The position $\mathbf{r}_j$ could assume one of six possible values,
$\uparrow$ or $\downarrow$.
The tuple could therefore assume one of twelve possible values:
\begin{align}
   \label{eq:TwelveTuples}
   & \LParen \uparrow,  ( \phi ,  h_+ )  \RParen,   \;
   \LParen \uparrow,  ( \phi ,  h_- )  \RParen,    \;
   \LParen \uparrow,  ( \phi + 2 \pi / 3 ,  h_+ )  \RParen,    \;
   \LParen \uparrow,  ( \phi + 2 \pi / 3 ,  h_- )  \RParen,    \;
   \LParen \uparrow,  ( \phi + 4 \pi / 3 ,  h_+ )  \RParen,    \;
   \LParen \uparrow,  ( \phi + 4 \pi / 3 ,  h_- )  \RParen,   
   \\ \nonumber &
   \LParen \downarrow,  ( \phi ,  h_+ )  \RParen,     \;
   \LParen \downarrow,  ( \phi ,  h_- )  \RParen,    \;
   \LParen \downarrow,  ( \phi + 2 \pi / 3 ,  h_+ )  \RParen,    \;
   \LParen \downarrow,  ( \phi + 2 \pi / 3 ,  h_- )  \RParen,    \;
   \LParen \downarrow,  ( \phi + 4 \pi / 3 ,  h_+ )  \RParen,    \;
   \; \text{or} \;  \;
   \LParen \downarrow,  ( \phi + 4 \pi / 3 ,  h_- )  \RParen  \, .
\end{align}
The nuclei would be ``dodequits'': 
$\dim ( \Hil_\nuc )$ would equal $2 \times 6  =  12$.

Let us return to reality: The phosphorus nuclei are indistinguishable fermions.
A hextuple of nuclei can occupy
the antisymmetric basis state~\eqref{eq:Slater}.
This state is labeled by a set of six tuples.
Each tuple must differ from each other tuple,
for the state to be antisymmetric.
To label a joint state,
we choose six of the twelve possible tuples.

But we cannot choose six arbitrary tuples.
No two tuples can contain the same position:
No two nuclei can coincide.
Hence we pair up the twelve possible tuples.
Each pair's constituent tuples
have the same positions and different spin states:
\begin{enumerate}

   \item $\LParen \uparrow,  ( \phi ,  h_+ )  \RParen,   \quad
   \LParen \downarrow,  ( \phi ,  h_+ )  \RParen$
   
   \item $\LParen \uparrow,  ( \phi ,  h_- )  \RParen,   \quad
   \LParen \downarrow,  ( \phi ,  h_- )  \RParen$
   
   \item $\LParen \uparrow,  ( \phi + 2 \pi / 3 ,  h_+ )  \RParen,  \quad
   \LParen \downarrow,  ( \phi + 2 \pi / 3 ,  h_+ )  \RParen$
   
   \item $\LParen \uparrow,  ( \phi + 2 \pi / 3 ,  h_- )  \RParen,   \quad
   \LParen \downarrow,  ( \phi + 2 \pi / 3 ,  h_- )  \RParen$
   
   \item $\LParen \uparrow,  ( \phi + 4 \pi / 3 ,  h_+ )  \RParen,   \quad
   \LParen \downarrow,  ( \phi + 4 \pi / 3 ,  h_+ )  \RParen$
   
   \item $\LParen \uparrow,  ( \phi + 4 \pi / 3 ,  h_- )  \RParen, 
   \LParen \downarrow,  ( \phi + 4 \pi / 3 ,  h_- )  \RParen$

\end{enumerate}
\noindent  We have formed six pairs of tuples.
We choose one tuple from each pair,
to label an antisymmetric joint basis state.

Let us count the ways in which we can choose the six tuples.
We can choose one tuple from each pair in two ways.
We choose from each of six pairs.
Hence we have $2^6 = 64$ choices of labels for
an antisymmetric joint state.

\section{Why the Posner's Hamiltonian is expected
to conserve $S_{1 \ldots 6}^{ \zinter }$}
\label{section:Z_Cons}

A Posner's phosphorus nuclear spins 
resist decoherence for long times, according to Fisher~\cite{Fisher15}.
We infer that the Posner Hamiltonian $H_\Posner$ preserves 
$S^{ \zlab }_{ 1 \ldots 6 }$
(the $z$-component, relative to the Posner's lab frame,
of the six phosphorus nuclei's total spin).
We support this interpretation by identifying candidate interactions 
that preserve the Posner's $C_3$ symmetry.
These interactions, we show, preserve $S^{ \zlab }_{1 \ldots 6}$.
Swift \emph{et al.} studied intra-Posner interactions
more rigorously
after the present paper's initial release~\cite[Eq.~(3)]{Swift_17_Posner}.
Their findings---including their form for $H_\Posner$---are consistent with ours.

The nuclei within a molecule can interact, in general.
Intramolecule interactions include 
the Coulomb exchange, kinetic exchange, and superexchange~\cite{Ashcroft_76_Solid}.
These interactions have the Heisenberg form 
\begin{align}
   \label{eq:Heis}
   \mathbf{S}_j \cdot  \mathbf{S}_k
   =  S_j^z  S_k^z  +  S_j^+  S_k^-  +  S_j^-  S_k^+  \, .
\end{align}
The $j^\th$ single-nucleus spin operator
is denoted by $\mathbf{S}_j$.
Raising and lowering operators are denoted by
$S_j^\pm  :=  \frac{1}{2} ( S_j^x  \pm  i  S_j^y )$.

Suppose that arbitrary phosphorus nuclear spins in a Posner
interact via Eq.~\eqref{eq:Heis}:
\begin{align}
   \label{eq:Heis2}
   H_\inter  =
   \sum_{j = 1}^6  \sum_{ k < j }  
   J_{jk}  \,  \mathbf{S}_j \cdot  \mathbf{S}_k  \, .
\end{align}
The pair-dependent interaction strength
is denoted by $J_{jk}$.
This $H_\inter$ remains invariant under
permutations of the spins via $\CThree$.
$\CThree$ represents the rotation
that preserves the Posner's geometry.
Hence the Posner's intrinsic Hamiltonian
might contain $H_\inter$.

The first term in Eq.~\eqref{eq:Heis} conserves
each spin's $S^{ z }_j$,
relative to an arbitrary reference frame.
The second term does not.
But suppose that any spin flips upward via $S^+_j$.
Another spin flips downward via $S^-_k$.
The compensation preserves the total spin's $z$-component.

\section{Decomposition of the Posner-molecule Hilbert space $\HilPos$
in terms of composite spin operators}
\label{section:Group_thry}

A Posner encodes logical qubits (Sec.~\ref{section:How_to_enc}).
Three qubits correspond to the $h_+$ triangle in Fig.~\ref{fig:h_axis},
via Eq.~\eqref{eq:Simple_code2}.\footnote{
No particular nucleus can be associated with any particular 
pure spin-and-position state, by Pauli's principle.
But a spin can be associated with a position.
Loosely speaking, some nucleus $A$ occupies 
the spin state $\ket{ m_j }$
if and only if $A$ occupies 
the position state $\ket{ \mathbf{r}_j }$.} 
We label these qubits 1, 2, and 3.
The $h_-$ triangle corresponds to logical qubits 4, 5, and 6.
The 64-dimensional logical space decomposes into
a direct sum of subspaces.
Different subspaces transform in different ways
under $\mathbf{S}_{123}^2  +  \mathbf{S}_{456}^2$,
the composite spin-squared operator~\eqref{eq:S_tot_op}.
Let us derive the decomposition.
We refer readers to standard quantum-mechanics textbooks,
such as~\cite{Shankar94}, for background.

Let us focus on one triangle (one trio of qubits) first.
Each trio corresponds to a Hilbert space
$\mathbb{C}^2  \otimes  \mathbb{C}^2  \otimes  \mathbb{C}^2$.
Each factor is replaced with
the corresponding subsystem's spin quantum number,
in useful conventional notation:
$s_1  \otimes  s_2  \otimes  s_3
=  \frac{1}{2}  \otimes  \frac{1}{2}  \otimes  \frac{1}{2}$.
This tensor product can be rewritten as a direct sum.

To derive the direct sum,
we follow rules for adding angular-momentum quantum numbers.
Two spin quantum numbers, $s_1$ and $s_2$, sum as
\begin{align}
   \label{eq:AddS}
   s_\tot  =  | s_1  - s_2 | ,  | s_1 - s_2 |  +  1,  \ldots
   s_1 + s_2 - 1 ,  s_1 + s_2  \, .
\end{align}
Two magnetic spin quantum numbers, $m_1$ and $m_2$, sum as
\begin{align}
   \label{eq:AddM}
   m_\tot  =  m_1 + m_2  \, .
\end{align}
We need not use Eq.~\eqref{eq:AddM} here, however.

Since tensor products distribute across direct sums,
\begin{align}
   \label{eq:Tri_group0}
   \frac{1}{2}  \otimes  \frac{1}{2}  \otimes  \frac{1}{2}
   & =  ( 0  \oplus  1 )  \otimes  \frac{1}{2}  \\
   \label{eq:Tri_group}
   & =  \frac{1}{2}  \oplus  \left( \frac{1}{2}  \oplus  \frac{3}{2} \right)  \, .
\end{align}
We can check Eq.~\eqref{eq:Tri_group}:
A space that transforms with spin quantum number $s$
has dimensionality $2s + 1$.
That is, $s$ corresponds to 
$2s + 1$ possible magnetic spin quantum numbers $m$.
According to the LHS of Eq.~\eqref{eq:Tri_group0}, therefore, 
each triangle corresponds to a space of dimensionality
$\left(2 \times \frac{1}{2}  +  1  \right)^3  =  2^3  =  8$.
Equation~\eqref{eq:Tri_group} implies the same dimensionality:
$2 + 2 + 4  =  8$.

Each Posner consists of two triangles.
A triangle pair corresponds to the Hilbert space
$\left(  \frac{1}{2}  \oplus  \frac{1}{2}  \oplus  \frac{3}{2}  \right)^{ \otimes 2 }$.
Distributing the tensor product across the direct sums yields
\begin{align}
   \label{eq:Pos_group0}
   \left(  \frac{1}{2}  \oplus  \frac{1}{2}  \oplus  \frac{3}{2}  \right)^{ \otimes 2 }
   & =  \left(  \frac{1}{2}  \otimes  \frac{1}{2}  \right)^{ \oplus 4 }
   \oplus  \left(  \frac{1}{2}  \otimes  \frac{3}{2}  \right)^{ \oplus 4 }
   \oplus  \left(  \frac{3}{2}  \otimes  \frac{3}{2}  \right)  \\
   & =  ( 0  \oplus  1  )^{  \oplus 4 }
   \oplus  ( 1  \oplus  2  )^{  \oplus 4 }
   \oplus ( 0 \oplus 1 \oplus 2 \oplus 3 )  \\
   \label{eq:Pos_group}
   & =  0^{ \oplus 5 }  \oplus  1^{ \oplus 9 }
   \oplus  2^{ \oplus 5 }  \oplus 3  \, .
\end{align}
Let us check Eq.~\eqref{eq:Pos_group}.
According to the LHS of Eq.~\eqref{eq:Pos_group0},
a Posner corresponds to a space of dimensionality 
$(2 + 2 + 4)^2  =  8^2 = 64$.
Equation~\eqref{eq:Pos_group} implies the same dimensionality:
$5  +  (3 \times 9 )  +  (5  \times  5 )  +  7  =  64$.

\section{Preferred eigenbasis of 
the permutation operator $\CThree$}
\label{section:C_eigenspaces}

The permutation operator $\CThree$
was introduced in Sec.~\ref{section:C_symm}.
The Posner dynamics is assumed to conserve
$\CThree$, as well as $S^{ \zlab}_{1 \ldots 6}$.
An eigenbasis shared by $\CThree$ and 
$S^{ \zlab}_{1 \ldots 6}$
can facilitate the construction of natural quantum error-correcting codes
(Sec.~\ref{section:QEC}).

Several eigenbases of $\CThree$ 
are eigenbases of $S^{ \zlab }_{1 \ldots 6}$.
The operator $\mathbf{S}_{123}^2  \otimes  \mathbf{S}_{456}^2$
breaks the degeneracy satisfactorily,
as discussed in Sections~\ref{section:Eigenbasis} and~\ref{section:AKLT}.
$\CThree$, $S^{ \zlab}_{1 \ldots 6}$, and 
$\mathbf{S}_{123}^2  \otimes  \mathbf{S}_{456}^2$
share the eigenbasis in
Tables~\ref{table:Tau0_states_Pos},~\ref{table:Tau1_states}, and~\ref{table:Tau2_states}.
Each table corresponds to one value of $\tau = 0, \pm 1$
(equivalently, $\tau = 0, 1, 2$).

%
%
\begin{table*}[t] 
\begin{center} 
\begin{tabular}{|M{1.2cm}|M{2cm}|M{1.2cm}|M{1.2cm}|M{1.2cm}|M{2cm}|M{3cm}|N}
    \hline
 $\text{State}$ 
 & $s_{123}\otimes s_{456}$ 
 & $m_{123}$ 
 & $m_{456}$ 
 & $m_{\text{1...6}}$ 
 & $\tau _{123}\otimes \tau _{456}$ 
 & $\text{Decomposition}$
 &
 \\[5pt]  \hline
 $\ket{c_{\tau = 0}^1}$ 
 & $\frac{3}{2}\otimes \frac{3}{2}$ 
 & $\frac{3}{2}$ & $\frac{3}{2}$ 
 & 3 
 & $1\otimes 1$ 
 & $\ket{000}\ket{000}$  
 & 
 \\[5pt] \hline
 $\ket{c_{\tau = 0}^2}$  
 & $\frac{3}{2}\otimes \frac{3}{2}$ 
 & $\frac{3}{2}$ 
 & $\frac{1}{2}$ 
 & 2 
 & $1\otimes 1$ 
 & $\ket{000}\ket{W} $ 
 &
 \\[5pt] \hline
 $\ket{c_{\tau = 0}^3}$ 
 & $\frac{3}{2}\otimes \frac{3}{2}$ & $\frac{3}{2}$ & $-\frac{1}{2}$ & 1 & $1\otimes 1$ & $\ket{000}\ket{\bar{W}} $ 
 &
 \\[5pt] \hline
 $\ket{c_{\tau = 0}^4}$ & $\frac{3}{2}\otimes \frac{3}{2}$ & $\frac{3}{2}$ & $-\frac{3}{2}$ & 0 & $1\otimes 1$ & $\ket{000}\ket{111}$
 &
 \\[5pt] \hline
 $\ket{c_{\tau = 0}^5}$ & $\frac{3}{2}\otimes \frac{3}{2}$ & $\frac{1}{2}$ & $\frac{3}{2}$ & 2 & $1\otimes 1$ & $\ket{W}\ket{000}$  
 &
 \\[5pt] \hline
 $\ket{c_{\tau = 0}^6}$ & $\frac{3}{2}\otimes \frac{3}{2}$ & $\frac{1}{2}$ & $\frac{1}{2}$ & 1 & $1\otimes 1$ & $\ket{W}\ket{W} $ 
 &
 \\[5pt] \hline
 $\ket{c_{\tau = 0}^7}$ & $\frac{3}{2}\otimes \frac{3}{2}$ & $\frac{1}{2}$ & $-\frac{1}{2}$ & 0 & $1\otimes 1$ & $\ket{W}\ket{\bar{W}} $
 &
 \\[5pt] \hline
 $\ket{c_{\tau = 0}^8}$ & $\frac{3}{2}\otimes \frac{3}{2}$ & $\frac{1}{2}$ & $-\frac{3}{2}$ & -1 & $1\otimes 1 $& $\ket{W}\ket{111} $ 
 &
 \\[5pt] \hline
 $\ket{c_{\tau = 0}^9}$ & $\frac{3}{2}\otimes \frac{3}{2}$ & $-\frac{1}{2}$ & $\frac{3}{2}$ & 1 & $1\otimes 1$ & $\ket{\bar{W}}\ket{000} $ 
 &
 \\[5pt] \hline
 $\ket{c_{\tau = 0}^{10}}$ & $\frac{3}{2}\otimes \frac{3}{2}$ & $-\frac{1}{2}$ & $\frac{1}{2}$ & 0 & $1\otimes 1 $& $\ket{\bar{W}}\ket{W}  $
 &
 \\[5pt] \hline
 $\ket{c_{\tau = 0}^{11}}$ & $\frac{3}{2}\otimes \frac{3}{2}$ & $-\frac{1}{2}$ & $-\frac{1}{2}$ & -1 & $1\otimes 1 $& $\ket{ \bar{W}}\ket{\bar{W}}$  
 &
 \\[5pt] \hline
 $\ket{c_{\tau = 0}^{12}}$ & $\frac{3}{2}\otimes \frac{3}{2}$ & $-\frac{1}{2}$ & $-\frac{3}{2}$ & -2 & $1\otimes 1 $& $\ket{\bar{W}}\ket{111}$  
 &
 \\[5pt] \hline
 $\ket{c_{\tau = 0}^{13}}$ & $\frac{3}{2}\otimes \frac{3}{2}$ & $-\frac{3}{2}$ & $\frac{3}{2}$ & 0 & $1\otimes 1 $&  $\ket{111}\ket{000}$ 
 &
 \\[5pt] \hline
 $\ket{c_{\tau = 0}^{14}}$ & $\frac{3}{2}\otimes \frac{3}{2}$ & $-\frac{3}{2}$ & $\frac{1}{2}$ & -1 & $1\otimes 1 $&  $\ket{111}\ket{W} $ 
 &
 \\[5pt] \hline
 $\ket{c_{\tau = 0}^{15}}$ & $\frac{3}{2}\otimes \frac{3}{2}$ & $-\frac{3}{2}$ & $-\frac{1}{2}$ & -2 & $1\otimes 1 $& $\ket{ 111}\ket{\bar{W}}$  
 &
 \\[5pt] \hline
 $\ket{c_{\tau = 0}^{16}}$ & $\frac{3}{2}\otimes \frac{3}{2}$ & $-\frac{3}{2}$ & $-\frac{3}{2}$ & -3 & $1\otimes 1 $&  $\ket{111}\ket{111} $
 &
 \\[5pt] \hline
 $\ket{c_{\tau = 0}^{17}}$ 
 & $\frac{1}{2}\otimes \frac{1}{2}$ 
 & $\frac{1}{2}$ & $\frac{1}{2}$ 
 & 1 
 & $\omega \otimes \omega ^2 $
 & $\ket{\omega }\ket{\omega^2} $ 
 &
 \\[5pt] \hline
 $\ket{c_{\tau = 0}^{18}}$ & $\frac{1}{2}\otimes \frac{1}{2}$ & $\frac{1}{2}$ & $-\frac{1}{2}$ & 0 & $\omega \otimes \omega ^2 $& $\ket{\omega }\ket{\overline{\omega ^2}}$  
 &
 \\[5pt] \hline
 $\ket{c_{\tau = 0}^{19}}$ & $\frac{1}{2}\otimes \frac{1}{2}$ & $\frac{1}{2}$ & $\frac{1}{2}$ & 1 & $\omega ^2\otimes \omega  $&  $\ket{\omega ^2}\ket{\omega}  $
 &
 \\[5pt] \hline
 $\ket{c_{\tau = 0}^{20}}$ & $\frac{1}{2}\otimes \frac{1}{2}$ & $\frac{1}{2}$ & $-\frac{1}{2}$ & 0 & $\omega ^2\otimes \omega  $& $\ket{\omega^2}\ket{\bar{\omega }}  $
 &
 \\[5pt] \hline
 $\ket{c_{\tau = 0}^{21}}$ & $\frac{1}{2}\otimes \frac{1}{2}$ & $-\frac{1}{2}$ & $\frac{1}{2}$ & 0 & $\omega \otimes \omega ^2 $& $\ket{\bar{\omega }}\ket{\omega^2}  $
 &
 \\[5pt] \hline
 $\ket{c_{\tau = 0}^{22}}$ & $\frac{1}{2}\otimes \frac{1}{2}$ & $-\frac{1}{2}$ & $-\frac{1}{2}$ & -1 & $\omega \otimes \omega ^2 $& $\ket{\bar{\omega}}\ket{\overline{\omega ^2}}$  
 &
 \\[5pt] \hline
 $\ket{c_{\tau = 0}^{23}}$  & $\frac{1}{2}\otimes \frac{1}{2}$ & $-\frac{1}{2}$ & $\frac{1}{2}$ & 0 & $\omega ^2\otimes \omega  $& $\ket{ \overline{\omega^2}}\ket{\omega}$  
 &
 \\[5pt] \hline
 $\ket{c_{\tau = 0}^{24}}$ & $\frac{1}{2}\otimes \frac{1}{2}$ & $-\frac{1}{2}$ & $-\frac{1}{2}$ & -1 & $\omega ^2\otimes \omega$  &  $\ket{ \overline{\omega^2}}\ket{\bar{\omega }}$ 
 &
 \\[5pt] \hline
\end{tabular}
\caption{\caphead{Preferred eigenbasis for the $\tau = 0$ eigenspace
of the permutation operator $\CThree$:}
Twenty-four states span the eigenspace.
Each basis element equals a product of 
two three-qubit states.
The final column displays the product, 
explained in Sec.~\ref{section:Eigenbasis}.
The state's first factor represents a state of 
the qubits (labeled $j = 1, 2, 3$) in the top triangle in Fig.~\ref{fig:h_axis}.
The second factor represents a state of 
the qubits (labeled $j = 4, 5, 6$) in the bottom triangle.
Each factor is an eigenstate shared by the total-spin operators
$\mathbf{S}_{123}^2$ and $S^{ \zlab }_{123}$
or by $\mathbf{S}_{456}^2$ and $S^{ \zlab }_{456}$.
The operators are defined in Sec.~\ref{section:Eigenbasis}.
Table~\ref{table:Trio_states} displays the three-qubit eigenstates.
The spin quantum number $s_{123}$ 
denotes the eigenvalue of $\mathbf{S}_{123}^2$.
The magnetic spin quantum number $m_{123}$
denotes the eigenvalue of $S^{ \zlab }_{123}$.
$s_{456}$ and $m_{456}$ are defined analogously.
The total magnetic spin quantum number
$m_{ 1 \ldots 6 }  =  m_{123} + m_{456}$.
The notation in column two follows from~\cite{Shankar94}:
Eigenspaces of $\mathbf{S}_{123}^2  \otimes  \mathbf{S}_{456}^2$
bear the label $s_{123} \otimes s_{456}$.
Column six is notated similarly.
$\tau_{123}$ denotes the eigenvalue of 
the permutation operator
that cyclically permutes qubits 1, 2, and 3.
$\tau_{456}$ is defined analogously.
The permutation eigenvalues multiply to
$\tau_{123} \times  \tau_{456}  =  \tau$.
}
\label{table:Tau0_states_Pos}
\end{center}
\end{table*}

%
%
\begin{table*}[t] 
\begin{center} 
\begin{tabular}{|M{1.2cm}|M{2cm}|M{1.2cm}|M{1.2cm}|M{1.2cm}|M{2cm}|M{3cm}|N}
    \hline
 $\text{State}$ & $s_{123}\otimes s_{456}$ & $m_{123}$ & $m_{456}$ & $m_{\text{1...6}}$ & $\tau _{123}\otimes \tau _{456}$ & $\text{Decomposition} $
 &
 \\[5pt] \hline
 $\ket{c_{\tau = 1}^1}$ & $\frac{3}{2}\otimes \frac{1}{2}$ & $\frac{3}{2}$ & $\frac{1}{2} $& 2 & $1\otimes \omega  $& $\ket{ 000} \ket{\omega}  $
 &
 \\[5pt] \hline
 $\ket{c_{\tau = 1}^2}$ & $\frac{3}{2}\otimes \frac{1}{2}$ & $\frac{3}{2}$ & $-\frac{1}{2}$ & 1 & $1\otimes \omega $ & $\ket{000}\ket{\bar{\omega }}  $
 &
 \\[5pt] \hline
 $\ket{c_{\tau = 1}^3}$ & $\frac{3}{2}\otimes \frac{1}{2}$ & $\frac{1}{2}$ & $\frac{1}{2} $& 1 & $1\otimes \omega  $& $\ket{W}\ket{\omega}  $
 &
 \\[5pt] \hline
 $\ket{c_{\tau = 1}^4}$ & $\frac{3}{2}\otimes \frac{1}{2}$ & $\frac{1}{2}$ & $-\frac{1}{2} $& 0 & $1\otimes \omega  $& $\ket{W} \ket{\bar{\omega }}  $
 &
 \\[5pt] \hline
 $\ket{c_{\tau = 1}^5}$ & $\frac{3}{2}\otimes \frac{1}{2}$ & $-\frac{1}{2}$ & $\frac{1}{2} $& 0 & $1\otimes \omega  $& $\ket{\bar{W}}\ket{\omega}  $
 &
 \\[5pt] \hline
 $\ket{c_{\tau = 1}^6}$ & $\frac{3}{2}\otimes \frac{1}{2}$ & $-\frac{1}{2}$ & $-\frac{1}{2}$ & -1 & $1\otimes \omega $ & $\ket{\bar{W}}\ket{\bar{\omega }}  $
 &
 \\[5pt] \hline
 $\ket{c_{\tau = 1}^7}$ & $\frac{3}{2}\otimes \frac{1}{2}$ & $-\frac{3}{2}$ & $\frac{1}{2} $& -1 & $1\otimes \omega  $& $\ket{111}\ket{\omega}  $
 &
 \\[5pt] \hline
 $\ket{c_{\tau = 1}^8}$ & $\frac{3}{2}\otimes \frac{1}{2}$ & $-\frac{3}{2}$ & $-\frac{1}{2}$ & -2 & $1\otimes \omega $ & $\ket{111}\ket{\bar{\omega }}  $
 &
 \\[5pt] \hline
 $\ket{c_{\tau = 1}^9}$ & $\frac{1}{2}\otimes \frac{3}{2}$ & $\frac{1}{2}$ & $\frac{3}{2} $& 2 & $\omega \otimes 1 $& $\ket{\omega}\ket{000}  $
 &
 \\[5pt] \hline
 $\ket{c_{\tau = 1}^{10}}$ & $\frac{1}{2}\otimes \frac{3}{2}$ & $\frac{1}{2}$ & $\frac{1}{2} $& 1 &$ \omega \otimes 1 $& $\ket{\omega}\ket{W} $
 &
 \\[5pt] \hline
 $\ket{c_{\tau = 1}^{11}}$ & $\frac{1}{2}\otimes \frac{3}{2}$ & $\frac{1}{2}$ & $-\frac{1}{2}$ & 0 & $\omega \otimes 1 $& $\ket{\omega }\ket{\bar{W}}  $
 &
 \\[5pt] \hline
 $\ket{c_{\tau = 1}^{12}}$ & $\frac{1}{2}\otimes \frac{3}{2}$ & $\frac{1}{2}$ & $-\frac{3}{2}$ & -1 & $\omega \otimes 1 $& $\ket{\omega}\ket{111}  $
 &
 \\[5pt] \hline
 $\ket{c_{\tau = 1}^{13}}$ & $\frac{1}{2}\otimes \frac{1}{2}$ & $\frac{1}{2}$ & $\frac{1}{2} $& 1 & $\omega ^2\otimes \omega ^2 $&$ \ket{\omega ^2}\ket{\omega ^2}  $
 &
 \\[5pt] \hline
 $\ket{c_{\tau = 1}^{14}}$ & $\frac{1}{2}\otimes \frac{1}{2}$ & $\frac{1}{2}$ & $-\frac{1}{2} $& 0 & $\omega ^2\otimes \omega ^2 $& $\ket{\omega ^2}\ket{\overline{\omega ^2}} $
 &
 \\[5pt] \hline
 $\ket{c_{\tau = 1}^{15}}$ & $\frac{1}{2}\otimes \frac{3}{2}$ & $-\frac{1}{2}$ & $\frac{3}{2} $& 1 & $\omega \otimes 1 $&$\ket{\bar{\omega }}\ket{000}  $
 &
 \\[5pt] \hline
 $\ket{c_{\tau = 1}^{16}}$ & $\frac{1}{2}\otimes \frac{3}{2}$ & $-\frac{1}{2}$ & $\frac{1}{2} $& 0 & $\omega \otimes 1 $&$\ket{\bar{\omega }}\ket{W}  $
 &
 \\[5pt] \hline
 $\ket{c_{\tau = 1}^{17}}$ & $\frac{1}{2}\otimes \frac{3}{2}$ & $-\frac{1}{2}$ & $-\frac{1}{2}$ & -1 &$ \omega \otimes 1 $&$ \ket{\bar{\omega }}\ket{\bar{W}}  $
 &
 \\[5pt] \hline
 $\ket{c_{\tau = 1}^{18}}$ & $\frac{1}{2}\otimes \frac{3}{2}$ & $-\frac{1}{2}$ & $-\frac{3}{2}$ & -2 & $\omega \otimes 1 $& $\ket{\bar{\omega }}\ket{111}  $
 &
 \\[5pt] \hline
 $\ket{c_{\tau = 1}^{19}}$ & $\frac{1}{2}\otimes \frac{1}{2}$ & $-\frac{1}{2}$ & $\frac{1}{2}$ & 0 & $\omega ^2\otimes \omega ^2 $& $\ket{ \overline{\omega ^2}}\ket{\omega ^2}  $
 &
 \\[5pt] \hline
 $\ket{c_{\tau = 1}^{20}}$ & $\frac{1}{2}\otimes \frac{1}{2}$ & $-\frac{1}{2}$ & $-\frac{1}{2}$ & -1 &$ \omega ^2\otimes \omega ^2 $&$ \ket{\overline{\omega ^2}}\ket{\overline{\omega ^2}}  $
 &
 \\[5pt] \hline    
\end{tabular}
\caption{\caphead{Preferred eigenbasis for the $\tau = 1$ eigenspace
of the rotation symmetry operator $\CThree$:}
The notation is defined below Table~\ref{table:Tau0_states_Pos}.
}
\label{table:Tau1_states}
\end{center}
\end{table*}

%
%
\begin{table*}[t] 
\begin{center} 
\begin{tabular}{|M{1.2cm}|M{2cm}|M{1.2cm}|M{1.2cm}|M{1.2cm}|M{2cm}|M{3cm}|N}
   \hline
 $\text{State} $ & $ s_{123}\otimes s_{456} $ & $ m_{123} $ & $ m_{456} $ & $ m_{\text{1...6}} $ & $ \tau _{123}\otimes \tau _{456} $ & $ \text{Decomposition}$ 
 &
 \\[5pt] \hline
 $\ket{c_{\tau = -1}^1} $ & $ \frac{3}{2}\otimes \frac{1}{2} $ & $ \frac{3}{2} $ & $ \frac{1}{2} $ & $ 2 $ & $ 1\otimes \omega ^2 $ & $ \ket{000}\ket{\omega ^2}  $ 
 &
 \\[5pt] \hline
 $\ket{c_{\tau = -1}^2} $ & $ \frac{3}{2}\otimes \frac{1}{2} $ & $ \frac{3}{2} $ & $ -\frac{1}{2} $ & $ 1 $ & $ 1\otimes \omega ^2 $ & $ \ket{000}\ket{\overline{\omega ^2}}  $ 
 &
 \\[5pt] \hline
 $\ket{c_{\tau = -1}^3} $ & $ \frac{3}{2}\otimes \frac{1}{2} $ & $ \frac{1}{2} $ & $ \frac{1}{2} $ & $ 1 $ & $ 1\otimes \omega ^2 $ & $ \ket{W}\ket{\omega ^2}  $ 
 &
 \\[5pt] \hline
 $\ket{c_{\tau = -1}^4} $ & $ \frac{3}{2}\otimes \frac{1}{2} $ & $ \frac{1}{2} $ & $ -\frac{1}{2} $ & $ 0 $ & $ 1\otimes \omega ^2 $ & $ \ket{W}\ket{\overline{\omega ^2}}  $ 
 &
 \\[5pt] \hline
 $\ket{c_{\tau = -1}^5} $ & $ \frac{3}{2}\otimes \frac{1}{2} $ & $ -\frac{1}{2} $ & $ \frac{1}{2} $ & $ 0 $ & $ 1\otimes \omega ^2 $ & $ \ket{\bar{W}}\ket{\omega ^2}  $ 
 &
 \\[5pt] \hline
 $\ket{c_{\tau = -1}^6} $ & $ \frac{3}{2}\otimes \frac{1}{2} $ & $ -\frac{1}{2} $ & $ -\frac{1}{2} $ & $ -1 $ & $ 1\otimes \omega ^2 $ & $ \ket{\bar{W}}\ket{\overline{\omega ^2}}  $ 
 &
 \\[5pt] \hline
 $\ket{c_{\tau = -1}^7} $ & $ \frac{3}{2}\otimes \frac{1}{2} $ & $ -\frac{3}{2} $ & $ \frac{1}{2} $ & $ -1 $ & $ 1\otimes \omega ^2 $ & $ \ket{111}\ket{\omega ^2}  $ 
 &
 \\[5pt] \hline
 $\ket{c_{\tau = -1}^8} $ & $ \frac{3}{2}\otimes \frac{1}{2} $ & $ -\frac{3}{2} $ & $ -\frac{1}{2} $ & $ -2 $ & $ 1\otimes \omega ^2 $ & $\ket{111}\ket{\overline{\omega ^2}}  $ 
 &
 \\[5pt] \hline
 $\ket{c_{\tau = -1}^9} $ & $ \frac{1}{2}\otimes \frac{1}{2} $ & $ \frac{1}{2} $ & $ \frac{1}{2} $ & $ 1 $ & $ \omega \otimes \omega  $ & $ \ket{\omega}\ket{\omega}  $ 
 &
 \\[5pt] \hline
 $\ket{c_{\tau = -1}^{10}} $ & $ \frac{1}{2}\otimes \frac{1}{2} $ & $ \frac{1}{2} $ & $ -\frac{1}{2} $ & $ 0 $ & $ \omega \otimes \omega  $ & $ \ket{\omega}\ket{\bar{\omega }}  $ 
 &
 \\[5pt] \hline
 $\ket{c_{\tau = -1}^{11}} $ & $ \frac{1}{2}\otimes \frac{3}{2} $ & $ \frac{1}{2} $ & $ \frac{3}{2} $ & $ 2 $ & $ \omega ^2\otimes 1 $ & $ \ket{\omega ^2}\ket{000}  $ 
 &
 \\[5pt] \hline
 $\ket{c_{\tau = -1}^{12}} $ & $ \frac{1}{2}\otimes \frac{3}{2} $ & $ \frac{1}{2} $ & $ \frac{1}{2} $ & $ 1 $ & $ \omega ^2\otimes 1 $ & $ \ket{\omega ^2}\ket{W}  $ 
 &
 \\[5pt] \hline
 $\ket{c_{\tau = -1}^{13}} $ & $ \frac{1}{2}\otimes \frac{3}{2} $ & $ \frac{1}{2} $ & $ -\frac{1}{2} $ & $ 0 $ & $ \omega ^2\otimes 1 $ & $ \ket{\omega ^2}\ket{\bar{W}}  $ 
 &
 \\[5pt] \hline
 $\ket{c_{\tau = -1}^{14}} $ & $ \frac{1}{2}\otimes \frac{3}{2} $ & $ \frac{1}{2} $ & $ -\frac{3}{2} $ & $ -1 $ & $ \omega ^2\otimes 1 $ & $\ket{\omega ^2}\ket{111}  $ 
 &
 \\[5pt] \hline
 $\ket{c_{\tau = -1}^{15}} $ & $ \frac{1}{2}\otimes \frac{1}{2} $ & $ -\frac{1}{2} $ & $ \frac{1}{2} $ & $ 0 $ & $ \omega \otimes \omega  $ & $\ket{\bar{\omega }}\ket{\omega}  $ 
 &
 \\[5pt] \hline
 $\ket{c_{\tau = -1}^{16}} $ & $ \frac{1}{2}\otimes \frac{1}{2} $ & $ -\frac{1}{2} $ & $ -\frac{1}{2} $ & $ -1 $ & $ \omega \otimes \omega  $ & $ \ket{\bar{\omega }}\ket{\bar{\omega }}  $ 
 &
 \\[5pt] \hline
 $\ket{c_{\tau = -1}^{17}} $ & $ \frac{1}{2}\otimes \frac{3}{2} $ & $ -\frac{1}{2} $ & $ \frac{3}{2} $ & $ 1 $ & $ \omega ^2\otimes 1 $ & $ \ket{\overline{\omega ^2}}\ket{000}  $ 
 &
 \\[5pt] \hline
 $\ket{c_{\tau = -1}^{18}} $ & $ \frac{1}{2}\otimes \frac{3}{2} $ & $ -\frac{1}{2} $ & $ \frac{1}{2} $ & $ 0 $ & $ \omega ^2\otimes 1 $ & $ \ket{\overline{\omega ^2}}\ket{W}  $ 
 &
 \\[5pt] \hline
 $\ket{c_{\tau = -1}^{19}} $ & $ \frac{1}{2}\otimes \frac{3}{2} $ & $ -\frac{1}{2} $ & $ -\frac{1}{2} $ & $ -1 $ & $ \omega ^2\otimes 1 $ & $ \ket{\overline{\omega ^2}}\ket{\bar{W}}  $ 
 &
 \\[5pt] \hline
 $\ket{c_{\tau = -1}^{20}} $ & $ \frac{1}{2}\otimes \frac{3}{2} $ & $ -\frac{1}{2} $ & $ -\frac{3}{2} $ & $ -2 $ & $ \omega ^2\otimes 1 $ & $\ket{ \overline{\omega ^2}}\ket{111}  $ 
 &
 \\[5pt] \hline
\end{tabular}
\caption{\caphead{Preferred eigenbasis for the $\tau = -1$ eigenspace
(equivalently, the $\tau = 2$ eigenspace)
of the rotation symmetry operator $\CThree$:}
The notation is defined below Table~\ref{table:Tau0_states_Pos}.
}
\label{table:Tau2_states}
\end{center}
\end{table*}

\section{Quantification of 
the information encoded in
the outcome of a Posner-binding measurement:
Analysis 2}
\label{section:Quant_Bind_Outcome}

The Posner-binding measurement is analyzed 
in Sec.~\ref{section:Stick_apps}.
The measurement yields an outcome
that encodes classical information.
This information is quantified in
Sec.~\ref{section:Analyze_Stick}.
The quantification is explained alternatively here.

Imagine wishing to measure the $\tau_A$ and $\tau_B$
of Posners $A$ and $B$. 
Each measurement would yield one of three possible outcomes
(0, 1, or 2).
The pair of measurements would yield 
one of nine possible outcomes.
The pair of outcomes could be recorded in
$\lceil \log_2 (9) \rceil  =  4$ bits.

Whether two Posners bind is equivalent to 
a measurement of
whether $\tau_A + \tau_B = 0$.
The yes-or-no answer constitutes one bit.
You forfeit three of the bits you wanted,
measuring just whether the Posners bind.
Three is the number of bits you would need
to specify the value of $( \tau_A , \tau_B )$,
given that $\tau_A + \tau_B \neq 0$.
Why? Suppose that $\tau_A + \tau_B \neq 0$.
$( \tau_A , \tau_B )$ can equal one of six possible values,
$( 0 , 1 )$, $( 0 , 2 )$, $(1 , 0 )$, $(1 , 1 )$, $(2 , 0 )$, or $(2 , 2 )$.
Specifying one of six possible values requires 
$\lceil \log_2 ( 6 )  \rceil  =  3$ bits.\footnote{
Imagine learning, instead, that $\tau_A + \tau_B = 0$.
Given this information, would you need three bits
to specify the value of $( \tau_A , \tau_B )$?
No: $( \tau_A , \tau_B )$ can assume one of three possible values.
Hence you would need $\lceil  \log_2 (3)  \rceil = 2$ bits.
But you could encode the tuple's value in three bits.}
Hence measuring Posner binding is equivalent to 
each of two QI processes:
\begin{enumerate}
   \item
   Measuring $( \tau_A , \tau_B )$ 
and coarse-graining away three bits
(all information except whether $\tau_A + \tau_B = 0$).

   \item
   Measuring the Bell basis
and coarse-graining away one bit
(whether a $+$ outcome or a $-$ outcome occurred).
\end{enumerate}


%
%
%
\section{How to prepare, with Posner operations, 
states used in incoherent teleportation}
\label{section:Tele_state_app}


Section~\ref{section:Teleportation} details how Posners can
teleport QI incoherently.
The protocol involves states 
$\ket{ \plust }  =  \frac{1}{ \sqrt{3} }  
( \ket{ \zerot }  +  \ket{ \onet }  +  \ket{ \twot } )$ and 
$\ket{ \psi }  =  c_0  \ket{ \zerot }  +  c_1  \ket{ \onet }  +  c_2  \ket{ \twot }$.
Each $\ket{ \jt }$ denotes an arbitrary state in
the $\tau = j$ subspace.
How can Posner operations (Sec.~\ref{section:Abstract2}) prepare 
a $\ket{ \plust }$ and a $\ket{ \psi }$?
One protocol is described below.
Other protocols may await discovery.

Each state is of one Posner and is pure.
Hence the Posner contains three singlets.
Consider preparing three singlets via operation~\ref{item:Bell}.
Consider rotating one spin 
about the $y_\lab$-axis,
through an angle $\theta$.

Consider forming a Posner from the spins,
via operation~\ref{item:Form_Pos}.
Let the singlets be arranged as in Fig.~\ref{fig:Prep_teleport}.
Recall that a Posner contains
two triangles of phosphorus nuclear spins
(Sec.~\ref{section:Posner_geo}).
One triangle sits at $\zinter = h_+$;
and the other triangle, at $\zinter = h_-$.
Each triangle contains one singlet
(illustrated with a green, wavy line).
One singlet extends from the $h_+$ triangle
to the $h_-$ triangle.
(How a singlet corresponds to 
positions in a Posner
is discussed in Sec.~\ref{section:How_to_enc_transf}
and App.~\ref{eq:Subtle_HilPos}.)
The red hoop encircles the rotated spin.
The rotated spin
is entangled with a spin in the same triangle.

Let $\ket{ \phi ( \theta ) }$ denote the Posner's state.
$\ket{ \phi ( \theta ) }$ can have weight on
each $\tau = j$ eigenspace:
\begin{align}
   \label{eq:Phi_theta}
   \ket{ \phi ( \theta ) }  =  \sum_{ j = 0}^2
   \sum_{ \lambda_j = 1 }^{ d_j }
   C_{ j , \lambda_j  } ( \theta ) 
   \ket{ c_{\tau = j}^{ \lambda_j } }  \, .
\end{align}
The $\tau = j$ eigenspace has degeneracy $d_j$.
The degeneracy parameter is denoted by $\lambda_j$.
The coefficients $C_{ j , \lambda_j  } ( \theta )$
satisfy the normalization condition
$\sum_{ j }  \sum_{ \lambda_j }
\left\lvert  C_{ j , \lambda_j  } ( \theta )  \right\rvert^2  =  1$.

The dependence on $\theta$ can be calculated analytically:
The state has an amount
\begin{align}
   \label{eq:Weight0}
   \sum_{ \lambda_0 = 1 }^{24}  
   \left\lvert  C_{ 0 , \lambda_0 } ( \theta ) \right\rvert^2
   =   \frac{1}{6} [ \cos ( 2 \theta )  +  2 ]
\end{align}
of weight on the $\tau = 0$ eigenspace, an amount
\begin{align}
   \label{eq:Weight1}
   \sum_{ \lambda_1 = 1 }^{20}  
   \left\lvert  C_{ 1 , \lambda_1 } ( \theta ) \right\rvert^2
   =  \frac{1}{12}  [  4 - \cos ( 2 \theta ) ]
\end{align}
on the $\tau = 1$ eigenspace, and an amount
\begin{align}
   \label{eq:Weight2}
   \sum_{ \lambda_2 = 1 }^{20}  
   \left\lvert  C_{ 2 , \lambda_2 } ( \theta ) \right\rvert^2
   =  \frac{1}{12}  [  4 - \cos ( 2 \theta ) ]
\end{align}
on the $\tau = 2$ eigenspace.

At which $\theta$-value does
the weight on each eigenspace equal
the weight on every other?
Let us equate~\eqref{eq:Weight0},~\eqref{eq:Weight1}, and~\eqref{eq:Weight2}.
Solving for the angle yields
$\theta = \frac{\pi}{4}$.
The corresponding state can serve as
the equal-weight superposition $\ket{ \plust }$:
\begin{align}
   \ket{ \plust }  =  \ket{ \phi \left( \pi / 4 \right) }  \, .
\end{align}
The basis vectors $\ket{ \jt }$ inherit the definition
\begin{align}
   \label{eq:Plus_basis}
   \ket{ \jt }  =  \sum_{ \lambda_j = 1 }^{ d_j }
   C_{ j , \lambda_j }  
    \left( \pi / 4 \right) 
   \ket{ c_j^{ \lambda_j } }  \, .
\end{align}

Now, let $\theta$ assume an arbitrary value.
Information about $\ket{ \phi ( \theta ) }$
can be teleported incoherently:
\begin{align}
   \label{eq:Psi_phi}
   \ket{ \psi }  =  \ket{ \phi ( \theta ) }  \, .
\end{align}
Granted, $\ket{ \phi ( \theta ) }$ might not decompose as
$\sum_{ j = 0 }^2  c_j  \ket{ \jt }$,
in terms of the $\ket{ \jt }$'s defined in Eq.~\eqref{eq:Plus_basis}.
Yet the incoherent-teleportation protocol continues to work:
Equation~\eqref{eq:Psi_phi} defines
new basis elements $\ket{ \jt ( \theta ) }$:
\begin{align}
   \label{eq:Plus_basis_psi}
   \ket{ \jt ( \theta ) }  =  \sum_{ \lambda_j = 1 }^{ d_j }
   C_{ j , \lambda_j } ( \theta )
   \ket{ c_j^{ \lambda_j } }  \, .
\end{align}

States $\ket{ \jt }$ of Posner $A$ appear in 
Eqs.~\eqref{eq:ToTeleport} and~\eqref{eq:Tel_3_state2a}.
Each such $\ket{ \jt }$ must be replaced with 
a $\ket{ \jt ( \theta ) }$.
The projector $\Pi_{AB}$ transforms the $\ket{ \jt ( \theta ) }$'s
as it would transform the $\ket{ \jt }$'s.

\begin{figure}[tb]
\centering
\includegraphics[width=.45\textwidth, clip=true]{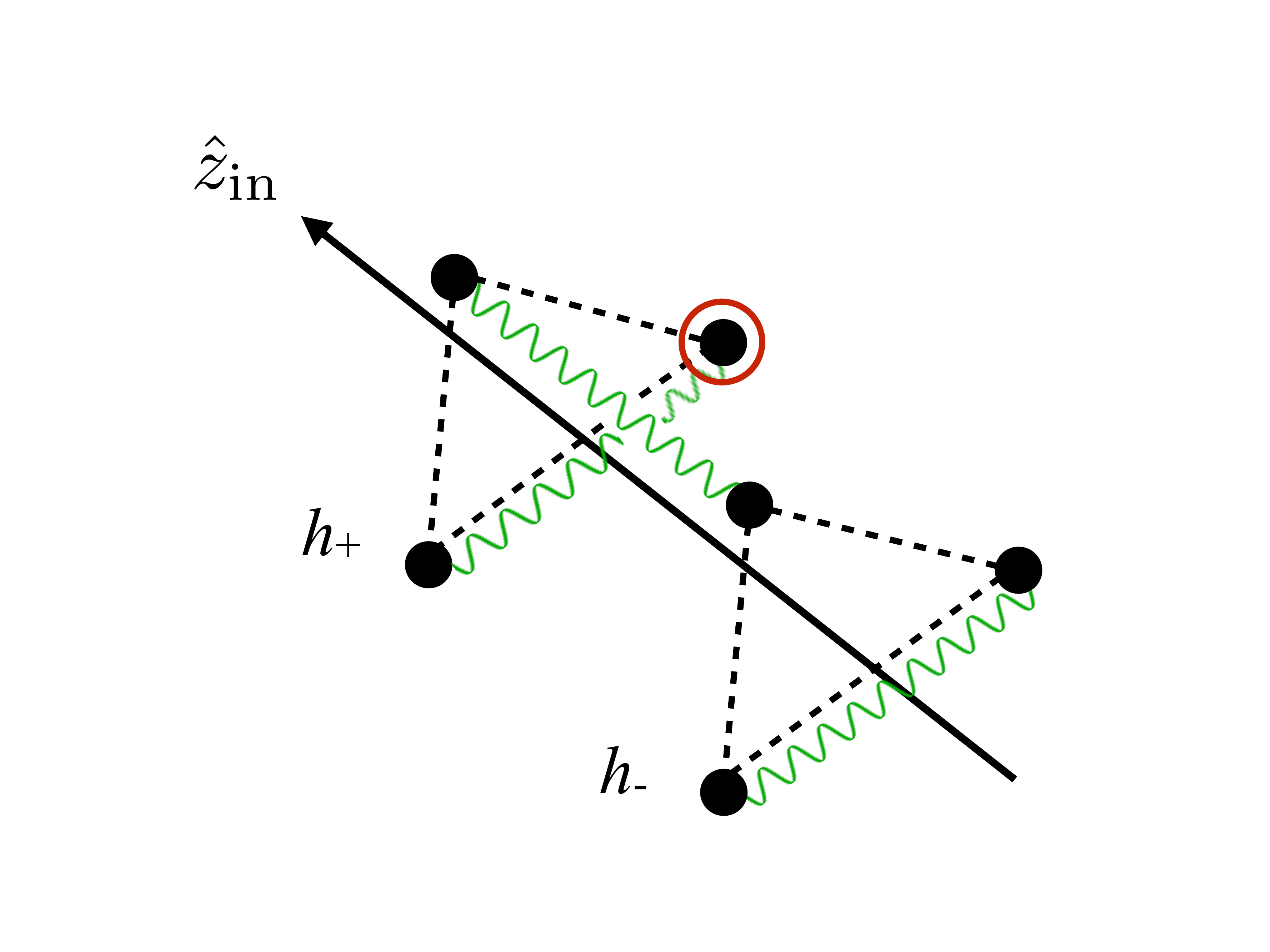}
\caption{\caphead{Posner-molecule state
usable in incoherent teleportation:}
Each black dot represents a phosphorus nuclear spin.
The internal $z$-axis $\hat{z}_\In$
remains fixed with respect to the atoms' positions.
Three spins sit at $\zinter = h_+$;
and three spins, at $\zinter = h_-$.
The spins occupy a pure state of three singlets.
Each green, wavy line represents one singlet.
The red hoop encircles a spin
that has been rotated through an angle $\theta$.
The rotation is about the $y_\lab$-axis,
which remains fixed relative to the lab that contains the Posner.
The angle labels the Posner's state, $\ket{ \phi ( \theta ) }$.
Instances of $\ket{ \phi ( \theta ) }$ can serve as 
the $\ket{ \plust }$ and the $\ket{ \psi }$
in incoherent teleportation (Sec.~\ref{section:Tele_protocol}).}
\label{fig:Prep_teleport}
\end{figure}
\section{Frustrated-lattice intuition about projecting onto the $\tau = 0$ subspace}
\label{section:Frust_int}

We can understand Eq.~\eqref{eq:Proj_prod2}
in terms of a frustrated lattice.
Consider a triangular lattice of three sites, $A$, $B$, and $C$.
Let a spin-1 DOF occupy each site.
The site-$K$ magnetic spin quantum number $m_K = 0, \pm 1$
stands in place of $\tau_K$. 

Let us regress to Eq.~\eqref{eq:Proj_prod1}.
We ignore the final $m - 3$ identity operators in each term.
How does $\Pi_{123}$ transform the lattice's state?
Consider multiplying out the terms in the RHS.
We label as a \emph{cross-term} each term that contains
at least one $\Pi_{ \tau_K = 0 }$
and one $\Pi_{ \tau_K = \pm 1 }$,
for some $K = A, B, C$.
These projectors annihilate each other;
the cross-terms vanish.
Each surviving term in $\Pi_{123}$
contains only $\tau_K = 0$ projectors
or only $\tau_K = \pm 1$ projectors.

Each $\tau_K = \pm 1$ projector 
represents an antiferromagnetic interaction
between two lattice sites.
The $\tau_K = \pm 1$ projectors form a term 
that represents a frustrated lattice.
No set $( \tau_A, \tau_B, \tau_C )$
satisfies all the constraints
encoded in the frustration term.
Hence the lattice must occupy its $\tau_A = \tau_B = \tau_C = 0$ subspace.

\section{PEPS representation of $\AKLTPrime$}
\label{section:PEPS}

The AKLT$'$ PEPS is a repeating pattern of two tensors,
$T^+$ and $T^-$ (Fig.~\ref{fig:PEPS}).
We will focus primarily on $T^+$. 
The tensor has six indices.
Three ($v^+_1$, $v^+_2$, and $v^+_3$) are virtual.
Three more indices ($a^+_1$, $a^+_2$, and $a^+_3$) are physical.

Each small, black dot represents a virtual spin.
Each short leg, extending upward from the plane occupied by the large circle,
represents a physical qubit.
We denote the physical qubits' computational-basis states
by $\ket{ a^+_1  \,  a^+_2  \,  a^+_3 }$.
For each $j = 1, 2, 3$, the physical index $a^+_j  =  0 , 1$.

Each long leg, extended across the plane occupied by the large circle,
represents a virtual index.
The $v^+_1$ and $v_2^+$ lines represents singlets.
Consider, as an illustration, the physical qubit
associated with $a^+_1$.
This qubit forms a singlet with
some physical qubit in another tensor.
Suppose that $a^+_1  =  0$.
The $T^+$ physical qubit points upward.
Hence the partner physical qubit must point downward:
The partner qubit's $a$ must equal one.
This necessity is conveyed to the second tensor
by the virtual index $v_1^+$: 
If $a^+_1  =  0$, 
$T^+_{ a^+_1 ,  a^+_2 , a^+_3 , v^+_1 ,  v^+_2 , v^+_3 }  \neq 0$
only if $v^+_1 = 0$.

The virtual index $v^+_3$ differs from $v_1^+$ and $v_2^+$:
The tensor lacks isotropy.
$v^+_3$ connects two tensors
associated with the same Posner,
$T^+$ and $T^-$.
The two tensors, together, determine
which $\CThree$ eigenspace the Posner occupies.
Hence $v^+_3$ carries
not only ``singlet'' information about
one physical qubit.
$v^+_3$ conveys also
how the $T^+$ qubit trio transforms
under $\CThree_3$
(the final column in Table~\ref{table:Trio_states}). 
This $\CThree$ information dictates
how the $T^-$ physical qubits must transform,
such that the Posner occupies the $\tau = 0$ eigenspace.

We ascribe to $v^+_3$ a tuple $( \tilde{v}^+_3 ,  \tau^+ )$.
The first entry conveys information about 
the $a^+_{ 3 }$ physical qubit.
Only if $\tilde{v}^+_3  =  a^+_{ 3 }$ can
the tensor have a nonzero value.
The second entry, $\tau^+$, equals 0, 1, or 2.
Hence $v^+_3$ assumes one of six possible values:
\begin{align}
   v^+_3  =  ( \tilde{v}^+_3 ,  \tau^+ )  
   & =  \Set{ ( 0, 0) ,  (0, 1 ) ,  ( 0 , 2 ) ,  ( 1, 0) ,  (1, 1 ) ,  ( 1 , 2 ) }  \\
   \label{eq:v3}
   & =  \Set{ 0 , 1 , 2 , 3 , 4 , 5 }  \, .
\end{align}
Hence $v^+_3$ has a bond dimension of six.

Having overviewed the tensor's six indices,
we consider the whole tensor,
$T^+_{ a^+_1 ,  a^+_2 , a^+_3 , v^+_1 ,  v^+_2 , v^+_3 }$.
This tensor equals
the coefficient that multiplies
the computational-basis state $\ket{ a^+_1 ,  a^+_2 , a^+_3 }$
when the virtual indices have the values
$v^+_1$,  $v^+_2$, and $v^+_3$.
Suppose, for simplicity, that
the $T^+$ triangle lacked connections to any other triangles.
The triangle would occupy the physical state
\begin{align}
   ( \const )  \sum_{ a^+_1 ,  a^+_2 , a^+_3 , v^+_1 ,  v^+_2 , v^+_3 }
   T^+_{ a^+_1 ,  a^+_2 , a^+_3 , v^+_1 ,  v^+_2 , v^+_3 }
   \ket{ a^+_1 ,  a^+_2 , a^+_3 }  \, .
\end{align} 
The $v$'s do not label the ket, because they are virtual.

The tensor can be evaluated, 
with help from Table~\ref{table:Trio_states},
after a normalization convention is chosen.
We illustrate with three examples.

First, let us evaluate $T^+_{ 0 0 0 0 0 0 }$.
Since $v^+_3 = 0$, Eq.~\eqref{eq:v3} implies that
$T^+_{ 0 0 0 0 0 0 }$ can $\neq  0$ only if $a^+_3 = 0$.
Indeed, $a^+_3 = 0$.
In fact, every $a$ vanishes.
This tensor equals the coefficient of
the one-triangle state
$\ket{ a^+_1  \,  a^+_2  \,  a^+_3 }  =  \ket{ 0 0 0 }$.
This state occupies the $\tau = 0$ eigenspace,
according to Table~\ref{table:Trio_states}.
We choose the following normalization condition:
$\ket{ 0 0 0 }$ appears once,
with a unit coefficient, in the table's second column.
Hence we choose for $T^+_{ 0 0 0 0 0 0 }$ to equal one.

The second example consists of
$T^+_{ a^+_1  \,  a^+_2  \,  a^+_3  \,  0 0 1 }$,
wherein the $a$'s have arbitrary values. 
According to the final three indices
[and Eq.~\eqref{eq:v3}],
the tensor can be nonzero only if
$a^+_j  =  0$  for all $j = 1, 2, 3$.
That is, $T^+_{ a^+_1  \,  a^+_2  \,  a^+_3  \,  0 0 1 }  =  0$ 
except, perhaps, if the coefficient of
$\ket{ a^+_1  \,  a^+_2  \,  a^+_3 }  =  \ket{ 0 0 0 }$.

The tensor's final index implies that $v^+_3 = 1$.
Hence, by Eq.~\eqref{eq:v3}, the qubit trio
transforms under $\CThree$ with $\tau^+ = 1$.
No qubit-trio state (i) transforms with $\tau^+$ and
(ii) equals a linear combination of computational-basis states including $\ket{ 0 0 0 }$,
by Table~\eqref{table:Trio_states}.
Hence $T^+_{ 0 0 0  0 0 1 }  =  0$.

The final example consists of $T^+_{ 1 0 0 1 0 0 }$.
The physical indices ``agree with'' the virtual indices:
$a^+_1 = v^+_1$,  $a^+_2 = v^+_2$,  
and [by Eq.~\eqref{eq:v3}]
$a^+_3 = \tilde{v}^+_3$.
Hence the tensor does not necessarily vanish.
This tensor multiplies the physical one-triangle ket 
$\ket{ a^+_1  \,  a^+_2  \,  a^+_3 }  =  \ket{ 1 0 0 }$.
This ket appears three times in
the second column of Table~\ref{table:Trio_states}.
Only one of those appearances is relevant:
Since $v^+_3  =  0$, Eq.~\eqref{eq:v3} implies that
$\tau_+ = 0$.
Hence the physical qubit trio occupies the first ket
in the $\ket{W}$ decomposition
(in the third row of Table~\ref{table:Trio_states}).
This ket multiples a $\frac{1}{ \sqrt{3} }$
in the table's second column.
We might wish to ascribe the value $\frac{1}{ \sqrt{3} }$
to $T^+_{ 1 0 0 1 0 0 }$.

But the physical qubits' state is constructed from singlets.
Singlets carry minus signs.
We must incorporate these minus signs
into our convention.
We choose for the tensor to carry a factor of $( - 1 )^{ a_j^+ }$
for each $j = 1 , 2 , 3$.
Hence $T^+_{ 1 0 0 1 0 0 }  =  - \frac{1}{ \sqrt{3} }$.

\section{How logical qubits could be rotated}
\label{section:One_Qubit_Gates}

Firing neurons, we propose, generate a magnetic field
that could rotate Posners' phosphorus nuclear spins significantly.
We review the interaction Hamiltonian.
Then, we quantify the magnetic field 
generated by firing neurons.
We form the rotation unitary $U(t)$,
then infer the time $t_\rot$ for which
the spin must rotate.
$t_\rot$, we expect, is much less than
the time $t_\fire$ for which a neuron fires.
But a spin could rotate significantly
over several firings.
The spin would not decohere significantly during this time,
if in a Posner.

These estimates are order-of-magnitude.
We often focus on the best possible case.

%
%
\textbf{Hamiltonian:}
Consider a spin  of magnetic moment $\bm{\mu}$.
A magnetic field $\mathbf{B}$ can evolve the spin
under the Hamiltonian 
$H_{\rm mag}  =  -  \bm{\mu}  \cdot  \mathbf{B}$.
The $^{31}$P nuclear spin has a magnetic moment of magnitude
$\mu  =  1.13 \mu_{\rm N}$~\cite{Fuller_76_Nuclear,Stone_15_Nuclear}.
The Bohr magneton is denoted by
$\mu_{\rm N}  =  \frac{ e \hbar}{ 2 m_{\rm p} }$;
and the proton mass, by $m_{\rm p}$.

%
%
\textbf{Magnetic-field strength:}
Firing causes a current to run down a neuron.
The current generates a magnetic field $\mathbf{B}$,
by the Biot-Savart law.
Luo \emph{et al.} model the in-brain $\mathbf{B}$
generated by neural tissue~\cite{Luo_11_Modeling}.
Table~3 on their p. 15 suggests that
the field can reach tens of nano-Tesla (nT).
Hence we approximate $B  :=  |  \mathbf{B}  |  
\approx  10^{-8}$ T.

Subtleties merit bearing in mind.
First, the in-tissue field has
a mean of $10^{-2} - 10^{-1}$ nT and
a standard deviation of $10^{-2} - 10^{-1}$ nT~\cite{Luo_11_Modeling}.
Our focus on the best case justifies
the use of a greater $B$.

Second, the study of the in-brain $\mathbf{B}$
has fluctuated over the past decade
(e.g.,~\cite{Xue_06_Direct,Blagoev_07_Modelling,Luo_11_Modeling,Jones_12_Detectability}).
Relatedly, magnetoencephelography (MEG) has guided studies of the in-brain field.
But MEG measures the field outside the skull. 
``Many current sources in the cortex are expected to cancel,''
Blagoev \emph{et al.} write~\cite{Blagoev_07_Modelling},
``leading to a small magnetic field outside the skull.
Hence, using the MEG-measured magnetic field strength to calculate
the magnitude of the field within the cortex
might lead to an underestimation.''
If $\mathbf{B}$ is stronger than believed,
single-qubit unitaries can be implemented more quickly than expected.

%
%
\textbf{Angle of rotation:}
Let $\hat{n}$ denote the axis of rotation.
Let $\theta$ denote the angle through which
a spin rotates.
How large must $\theta$ be for
the global quantum state to change significantly?
One might na\"{i}vely guess $\pi$.

But quantum-cognition spins are prepared in singlets
$\ket{ \Psi^- }$
(via operation~\ref{item:Bell}).
$\ket{ \Psi^- }$ remains invariant under
arbitrary identical rotations of both qubits,
$U_\alpha  :=
U_{ \hat{n} }   \left(  \frac{\theta}{2}  \right)
\otimes  U_{ \hat{n} }   \left(  \frac{\theta}{2}  \right)$.
$U_\alpha$ transforms $\ket{ \Psi^- }$ as does
the identity operation, $\id$:
$\ket{ \Psi^- }  =  U_\alpha  \ket{ \Psi^- }$.
Hence the qubits need not rotate physically
to undergo $U_\alpha$ effectively.

Consider rotating the qubits oppositely, physically, with 
$U_\beta
:=  U_{ \hat{n} } \left(  \frac{\theta}{2}  \right)
\otimes  U_{ \hat{n} }  \left(  -  \frac{\theta}{2}  \right)$.
The two unitaries combine group-theoretically:
$\ket{ \Psi^- }  =  U_\alpha  \ket{ \Psi^- }
\mapsto  U_\beta  U_\alpha  \ket{ \Psi^- }
=  [ U_{ \hat{n} } ( \theta )  \otimes   \id ]  \ket{ \Psi^- }$.
Hence rotating qubit 1 through $\theta$ counterclockwise
is equivalent to
(i) rotating qubit 1 through $\frac{\theta}{2}$ counterclockwise while
(ii) rotating qubit 2 through $\frac{\theta}{2}$ clockwise.

Hence the time $t_\rot$ required to rotate a qubit effectively 
through an angle $\pi$
equals the time required to rotate a qubit physically
through an angle $\theta  =  \frac{\pi}{2}$.
The unitary
\begin{align}
   \label{eq:U_Form1}
   \exp  \left(  - i  \,   \frac{ \theta }{ 2 }  \, 
   \hat{ \mathbf{n} }  \cdot  \bm{\sigma}  \right)  
\end{align}
rotates a qubit through an angle $\theta$.
Since $\theta = \frac{ \pi }{ 2 }$, 
$\frac{ \theta }{ 2 }  =  \frac{ \pi }{ 4 }$.
The order-of-magnitude estimate will eliminate the $\frac{1}{2}$,
but the half is worth being aware of.

%
%
\textbf{Equation of unitaries and solution for $t_\rot$:}
The Hamiltonian generates the unitary
\begin{align}
   \label{eq:U_Form2}
   \exp  \left( - \frac{ i }{ \hbar }  \,  H_{\rm mag} t  \right)
   \approx  \exp  \left(  i  
   \frac{ 1.13 e }{ 2 m_{\rm p} } B  t  \right)  \, .
\end{align}
The exponentials~\eqref{eq:U_Form2} and~\eqref{eq:U_Form1}
equal each other when $t = t_\rot$.
We equate the exponentials' arguments
and neglect order-one factors:
$\frac{ e B }{ m_{\rm p} }  t_\rot  
\approx  1$.
Solving for the time scale's order of magnitude yields
\begin{align}
   t_\rot  
   \approx  \frac{ m_{\rm p} }{ e B }
   \approx  \frac{ 10^{ -27 }  \text{ kg} }{
   \left( 10^{ - 19 }  \text{ C}  \right)
   \left(  10^{ -8 }  \text{ T}  \right) }
   \approx 1  \text{ s.}
\end{align}
Let us compare this required rotation time to
the duration $t_\fire$ of one neuron firing.

%
%
\textbf{Duration of neuron firing:}
Xue \emph{et al.} attribute $5-10$ ms to a firing~\cite{Xue_06_Direct}, 
citing~\cite{Destexhe_98_Dendritic,Williams_02_Dependence,Gulledge_05_Synaptic}.
We therefore approximate $t_\fire  \approx  10 \text{ ms}  
=  10^{-2}$ s.
One firing does not last long enough
to rotate a qubit through an angle $\frac{\pi}{2}$:
$t_\fire  \approx  10^{-2} \text{ s}
\ll  1 \text{ s}  
\approx  t_\rot$.
But $10^2$ firings could rotate the neuron enough.

%
%
\textbf{Frequency of neuron firing:}
Blagoev \emph{et al.} write that
``a `typical' neuron spikes 0.1 to 10 times 
a second''~\cite{Blagoev_07_Modelling}.
We focus on the best case of ten firings per second.
One hundred firings would consume about ten seconds.
Hence rotating a qubit through 
an angle $\sim \frac{\pi}{2}$ would take 
$t'_\rot  \approx  10$ s.
Let us compare this rotation time to
two time scales that characterize the qubit.

\textbf{Comparison with spin lifetime:}
Consider a phosphorus nuclear spin
in a lone phosphate.
The spin is expected to have a lifetime of
$\sim 1 \text{ s}  \ll  t'_\rot$~\cite{Fisher15} .
The spin will decohere before rotating appreciably.
But a spin in a Posner is expected to have
a lifetime of $\sim 10^5 - 10^6$ s~\cite{Fisher15}.
In-Posner qubits could undergo 
$\sim 10^4 - 10^5$ single-qubit gates before decohering.

\textbf{Comparison with diffusion time:}
A Posner could diffuse between the neuron firings.
Let us estimate the distance diffused.
We estimated the Posner's diffusion constant in Sec.~\ref{section:diV}:
$D  \sim  10^{ - 10 }  \text{ m}^2 / \text{s}$.
Solving $D  \sim  \frac{ \ell^2 }{ t }$ for distance yields
$\ell \sim  \sqrt{ D t }
\sim  \sqrt{ \left(  10^{ - 10 }  \text{ m}^2 / \text{s}  \right)
\left(10  \text{ s}  \right)  }
=  1$ mm.

One millimeter equals approximately another relevant length:
The in-tissue magnetic field appears
as a function of 
two-dimensional position in~\cite[Fig.~3]{Luo_11_Modeling}.
The tallest spikes represent field strengths
$B  \approx  10$ nT.
About a millimeter separates neighboring peaks.
Hence a Posner could diffuse from peak to peak,
rotating maximally during each firing.

We do not expect a Posner to hit 
peak after peak typically.
But we have presented the best possible case.
At best, a qubit could effectively rotate through
angles up to $\pi$.

\end{appendices}

%
%
\bibliographystyle{h-physrev}
\bibliography{CognitionRefs}

\end{document}